%% file: main.tex
\PassOptionsToPackage{dvipsnames}{xcolor}
\documentclass[conference]{99-config/IEEEtran}
\input{99-config/my-packages.tex}
\input{99-config/my-commands.tex}

\input{99-config/my-proof-connectors.tex}

\def\BibTeX{{\rm B\kern-.05em{\sc i\kern-.025em b}\kern-.08em
    T\kern-.1667em\lower.7ex\hbox{E}\kern-.125emX}}

\begin{document}
\input{99-config/mes_acronymes.tex}

\title{Network-Calculus Service Curves of the\\Interleaved Regulator}

\author{\IEEEauthorblockN{Ludovic Thomas}
\IEEEauthorblockA{%
École Polytechnique Fédérale de Lausanne (EPFL)
}
\and
\IEEEauthorblockN{Jean-Yves Le Boudec}
\IEEEauthorblockA{%
École Polytechnique Fédérale de Lausanne (EPFL)
}
}

\maketitle

\begin{abstract}
    \input{0-abstract/0-abstract.tex}
\end{abstract}
\input{0-abstract/keywords.tex}

\input{1-introduction/1-introduction}

\input{2-background/0-background.tex}

\input{3-related-work/0-related-work.tex}

\input{4-system-model/0-system-model}
\input{5-fifo-sff/0-fifo-sff.tex}

\input{6-service-curves/0-service-curves.tex}
\input{8-conclusion/0-conclusion}

\bibliography{./biblio/IEEEabrv,./biblio/zotero}

\appendices
\input{9-appendix/0-appendix}

\end{document}

%% file: 99-config/my-packages.tex
\usepackage[utf8]{inputenc}\DeclareUnicodeCharacter{2212}{-}
\usepackage[nolist]{acronym}
\usepackage{amsmath,amsfonts,amsthm}
\usepackage{amssymb,mathrsfs}
\usepackage{graphicx}
\usepackage{textcomp}
\usepackage{tikz}
\usepackage{pgfplots}
\usepackage{multirow}
\usepackage{makecell}
\usepackage{algorithm}
\usepackage{algpseudocode}
\usepackage[caption=false,font=footnotesize]{subfig}
\usepackage{framed}
\usepackage{soul}
\usepackage{bm}
\usepackage[normalem]{ulem}
\usepackage{pdfcomment}
\usepackage{import}
\usepackage{tkz-euclide}
\usepackage{stmaryrd}
\usepackage{xparse}
\usepackage{mathtools}
\usepackage[shortlabels]{enumitem}
\usepackage{nicematrix}
\pgfplotsset{compat=1.18}
\usepackage{bbm}
\usepackage{cancel}
\usepackage{xcolor}

%% file: 99-config/my-commands.tex
\usepgfplotslibrary{external,fillbetween}
\tikzset{every picture/.style={line width=0.75pt}}
\tikzset{cross/.style={cross out, draw=black, fill=none, minimum size=2*(#1-\pgflinewidth), inner sep=0pt, outer sep=0pt}, cross/.default={2pt}}
\usetikzlibrary{patterns.meta, positioning, shapes, trees, automata, calc, fit, automata, shapes.misc, arrows, patterns, math, decorations.pathmorphing, arrows.meta}
\def\BibTeX{{\rm B\kern-.05em{\sc i\kern-.025em b}\kern-.08em
		T\kern-.1667em\lower.7ex\hbox{E}\kern-.125emX}}
	     
\providecommand{\DIFdel}[1]{} 
\renewcommand{\DIFdel}[1]{\color{red}\sout{#1}\color{black}} 
\providecommand{\DIFadd}[1]{} 
\renewcommand{\DIFadd}[1]{\color{blue}#1\color{black}}

\newcommand{\tagg}[1]{%
	\colorbox{green}{#1}
}

\theoremstyle{definition}
\newtheorem{definition}{Definition}
\newtheorem{theorem}{Theorem}
\newtheorem{lemma}{Lemma}
\newtheorem{proposition}{Proposition}

\newcommand{\ie}{\textit{i.e.,~}}
\newcommand{\eg}{\textit{e.g.,~}}
\newcommand{\resp}{resp., }



%% file: 99-config/my-proof-connectors.tex
\def\proofEndSysModel{Each result in this paper is associated with an intuition of the proof, and the formal proofs are in Appendix~\ref{appendix:proofs}}

\def\proofSpring{Appendix~\ref{proof:spring}}

\def\proofPropPfrScModel{Appendix~\ref{proof:fifo:pfr-sff-explain}}

\def\proofLemmaMinimumScExplain{Appendix~\ref{proof:lem:fifo-sff:no-explain:minimum-sc-sff}}

\def\proofIrNoScExplain{Appendix~\ref{proof:thm:fifo-sff:no-explain:ir-has-no-explaination}}

\def\proofIrStrictSc{Appendix~\ref{proof:thm:ir-service-curves:strict-sc}}

\def\proofPropLimitIrStrictSc{Appendix~\ref{proof:prop:ir-service-curves:strict-sc:limit}}

\def\proofThmLimitIndiv{Appendix~\ref{proof:thm:ir-service-curves:limit-ir-residual}}

\def\proofThmLimitAggreg{Appendix~\ref{proof:thm:ir-service-curves:limit-sc}}

%% file: 99-config/mes_acronymes.tex
\begin{acronym}[CP-OFDMX]
	\acro{XXXXX}{\tagg{List of acronyms for Holly}\tagg{Will be removed in final version}}
	\acro{ACP}{aggregate computation pipeline}
	\acro{AFDX}{Avionics Full-dupleX switched Ethernet}
	\acro{ATS}{asynchronous traffic shaping}
	\acro{AVB}{Audio Video Bridging}
	\acro{CAN}{Controller Area Network}
	\acro{CBQS}{class-based queuing subsystem}
	\acro{CBS}{credit-based scheduler}
	\acro{CDT}{control-data traffic}
	\acro{CEV}{crew exploration vehicle}
	\acro{COTS}{commercial off the shelf}
	\acro{DAG}{directed acyclic graph}
	\acro{DetNet}{deterministic networking}
	\acro{DNC}{deterministic network calculus}
	\acro{ETE}{end-to-end}
	\acro{EP}{elimination-pending}
	\acro{FIFO}{first in, first out}
	\acro{FP}{fixed-priority}
	\acro{FRER}{frame replication and elimination for redundancy}
	\acro{HOL}{head-of-line}
	\acro{HSR}{High-availability Seamless Redundancy}
	\acro{HTTP}{Hypertext Transfer Protocol}
	\acro{IEC}{International Electrotechnical Committee}
	\acro{IEEE}{Institute of Electrical and Electronics Engineers}
	\acro{IETF}{Internet Engineering Task Force}
	\acro{IMA}{Integrated Modular Avionics}
	\acro{iPRP}{IP parallel redundancy protocol}
	\acro{IR}{interleaved regulator}
	\acro{LCAN}{low-cost acyclic network}
	\acro{MCU}{micro-controller unit}
	\acro{MFAS}{minimum feedback arc set}
	\acro{MFVS}{minimum feedback vertex set}
	\acro{MOST}{Media Oriented Systems Transport}
	\acro{NC}{network calculus}
	\acro{NoC}{networks on chip}
	\acro{OSI}{Open Systems Interconnection}
	\acro{PBOO}{pay burst only once}
	\acro{PFR}{per-flow regulator}
	\acro{PMOC}{pay multiplexing only at convergence points}
	\acro{PMOO}{pay multiplexing only once}
	\acro{PEF}{packet-elimination function}
	\acro{POF}{packet-ordering function}
	\acro{PRF}{packet-replication function}
	\acro{PREF}[PREFs]{packet replication and elimination functions}
	\acro{PREOF}[PREOFs]{packet replication, elimination and ordering functions}
	\acro{PRP}{Parallel Redundancy Protocol}
	\acro{QoS}{quality of service}
	\acro{RAMS}{Reliability, Availability, Maintainability, and Safety}
	\acro{RBO}{reordering byte offset}
	\acro{REG}{regulator}
	\acro{RSTP}{Rapid Spanning Tree Protocol}
	\acro{RTE}{Real-Time Ethernet}
	\acro{RTO}{reordering late time offset}
	\acro{SFA}{single-flow analysis}
	\acro{SNC}{stochastic network calculus}
	\acro{TAS}{Time-Aware Shaping}
	\acro{TCP}{Transmission Control Protocol}
	\acro{TFA}{total-flow analysis}
	\acro{TP}{turn prohibition}
	\acro{TSN}{time-sensitive networking}
	\acro{VBR}{variable-bit-rate}
	\acro{VIU}{vehicle interface unit}
\end{acronym}

%% file: 0-abstract/0-abstract.tex
The interleaved regulator (implemented by IEEE TSN Asynchronous Traffic Shaping) is used in time-sensitive networks for reshaping the flows with per-flow contracts.
When applied to an aggregate of flows that come from a FIFO system, an interleaved regulator
that reshapes the flows with their initial contracts
does not increase the worst-case delay of the aggregate.
This shaping-for-free property supports the computation of end-to-end latency bounds 
and the validation of the network's timing requirements.
A common method to establish the properties of a network element is to obtain a network-calculus service-curve model.
The existence of such a model for the interleaved regulator remains an open question.
If a service-curve model were found for the interleaved regulator, then the analysis of this mechanism would no longer be limited to the situations where the shaping-for-free holds, which would widen its use in time-sensitive networks.
In this paper, we investigate if network-calculus service curves can capture the behavior of the interleaved regulator.
We find that an interleaved regulator placed outside of the shaping-for-free requirements (after a non-FIFO system) can yield unbounded latencies.
Consequently, we prove that no network-calculus service curve exists to explain the interleaved regulator's behavior.
It is still possible to find non-trivial service curves for the interleaved regulator.
However, their long-term rate cannot be large enough to provide any guarantee (specifically, we prove that for the regulators that process at least four flows with the same contract, the long-term rate of any service curve is upper bounded by three times the rate of the per-flow contract).

%% file: 0-abstract/keywords.tex

%% file: 1-introduction/1-introduction.tex
\section{Introduction}
\label{sec:intro}

Time-sensitive networks, as specified by the \ac{TSN} task group of the \ac{IEEE}, support safety-critical applications in the aerospace, automation, and automotive domains \cite{farkasTSNBasicConcepts2018}.
To do so, time-sensitive networks provide a deterministic service with guaranteed bounded latencies.

These guarantees must be validated through a deterministic worst-case timing analysis that can be performed with network calculus.
This mathematical framework obtains worst-case performance bounds by modeling the flows with the concept of arrival curves and the network elements with the concept of service curves.
Service curves constrain the minimum amount of service that network elements provide to a flow or aggregate of flows.

Time-sensitive networks can also rely on a set of mechanisms that improve the traditional forwarding process of an output port.
The traffic regulators are such hardware elements that support higher scalability and efficiency of time-sensitive networks.
Placed after a multiplexing stage, they reshape the flows with per-flow shaping curves (by delaying packets if required) and remove the increase of the flows' burstiness due to their interference with other flows.

Traffic regulators come in two flavors: \acfp{PFR} and \acfp{IR}. 
A \ac{PFR} processes a unique flow.
It stores the packets of the flow in a \ac{FIFO} queue and releases the head-of-line packet as soon as doing so does not violate the configured shaping curve for the flow.
In contrast, the \ac{IR} processes an aggregate of flows with a unique \ac{FIFO} queue.
Each flow has its own configured shaping curve, but the \ac{IR} analyses only the head-of-line packet and releases it as soon as doing so does not violate the shaping curve of the associated flow. 
The packets in the \ac{IR} queue are blocked by the head-of-line even if they belong to other flows.
This second flavor is implemented within \ac{IEEE} \ac{TSN} under the name \acf{ATS} \cite{ieee8021Qcr}.

In time-sensitive networks that contain traffic regulators, end-to-end latency bounds are obtained from the knowledge of the shaping curves enforced by the regulators and from the essential ``shaping-for-free'' property.
It states that the traffic regulators do not increase the worst-case latency of the flow (or of the flow aggregate) under certain conditions that depend on the type of the regulator (\ac{PFR} or \ac{IR}).
Most analyses of traffic regulators rely on this property.

For the \ac{PFR}, the shaping-for-free property is well understood because 
a \ac{PFR} with a concave shaping curve 
can be modeled with a context-agnostic service curve, \ie a service curve that only depends on the configuration of the \ac{PFR} but not on the context in which the \ac{PFR} is placed.
This 
service curve proves the shaping-for-free property when the \ac{PFR} is placed in the appropriate context.
When the \ac{PFR}'s context deviates from the shaping-for-free requirements, the context-agnostic service curve still provides performance bounds for the \ac{PFR}, and slight deviations of the context lead to bounded delay penalties.
On the contrary, no known non-trivial context-agnostic service curve exists for the \ac{IR}.
The non-trivial service-curve models published in the literature \cite{mohammadpourLatencyBacklogBounds2018,zhaoQuantitativePerformanceComparison2022} are context dependent and always assume that the shaping-for-free property holds. 
Without a context-agnostic service-curve model, performance bounds cannot be obtained for an \ac{IR} placed outside the shaping-for-free requirements, which restrains its use in time-sensitive networks.

In this paper, we investigate if the behavior of the \ac{IR} can be modeled by a context-agnostic network-calculus service curve.
Our contributions are:

\noindent $\bullet$ As opposed to the shaping-for-free property when the \ac{IR} is placed after a \ac{FIFO} system, we prove that the \ac{IR} can yield unbounded latencies when placed after a non-\ac{FIFO} system, even if the latter is \ac{FIFO}-per-flow and lossless. 

\noindent $\bullet$ We prove that the shaping-for-free property of the \ac{IR} cannot be explained by any network-calculus service-curve model.

\noindent $\bullet$ For any \ac{IR} that processes at least four flows, we prove that any context-agnostic service curve for an individual flow is upper bounded by a constant.

\noindent $\bullet$ We exhibit a strict service curve of the \ac{IR} and a function that upper bounds any other context-agnostic strict service curve.

\noindent $\bullet$ For any \ac{IR} that processes at least four flows, we show that the long-term rate of any context-agnostic service curve is upper bounded by three times the rate of the per-flow contract.

The paper is organized as follows.
We provide the background on network-calculus service curves and regulators in Section~\ref{sec:background}.
We discuss the related work in Section~\ref{sec:related-work} and provide the system model in Section~\ref{sec:sysmodel}.
We then analyze the role of the \ac{FIFO} assumption in the shaping-for-free property of the \ac{IR} in Section~\ref{sec:fifo-sff}.
Afterward, we discuss the context-agnostic service curves of the \ac{IR} in Section~\ref{sec:ir-service-curves}.
We provide our conclusive remarks in Section~\ref{sec:conclusion}.

%% file: 2-background/0-background.tex
\section{Background}
\label{sec:background}

In time-sensitive networks, performance metrics such as flows' end-to-end latencies have to be bounded in the worst case, not in average.
The network-calculus framework \cite{leboudecNetworkCalculus2001,Chang2000,bouillardDeterministicNetworkCalculus2018} can provide such performance bounds. 
It describes the data traffic with cumulative functions, such as $R^\mathtt{A}$, where $R^\mathtt{A}(t)$ is the amount of data that cross the observation point $\mathtt{A}$ between an arbitrary time reference $0$ and $t$. 
Cumulative functions belong to $\mathfrak{F}_0=\{\mathfrak{f}:\mathbb{R}^+\rightarrow\mathbb{R}^+|\mathfrak{f}(0)=0\}$.

\input{2-background/1-nc-sc.tex}
\input{2-background/2-packetizers.tex}
\input{2-background/2-residual-service-curve.tex}
\input{2-background/2-traffic-regulators-sff.tex}

%% file: 2-background/1-nc-sc.tex
\subsection{Network-Calculus Service Curves}
\label{sec:background:nc-sc}

A causal network system $S$ offers a service curve $\beta$ if (a) $\beta$ is wide-sense increasing and (b) for any input cumulative function $R^\mathtt{A}(t)$, the resulting output traffic $R^\mathtt{B}(t)$ verifies
\begin{equation}\label{eq:background:sc-def}
    \forall t\ge 0, \quad R^\mathtt{B}(t) \ge (R^{\mathtt{A}}\otimes\beta)(t)
\end{equation}
where $\otimes$ describes the min-plus convolution\footnote{Defined by $(R^{\mathtt{A}}\otimes\beta):t\mapsto \inf_{s\ge0}(R^{\mathtt{A}}(s)+\beta(t-s))$}.
Common service curves are of the form \emph{rate latency} $\beta_{R,T}:t\mapsto R\cdot|t-T|^+$ with rate $R$ and latency $T$, where $|\cdot|^+ = \max(\cdot,0)$.

Some network systems provide stronger guarantees through a \emph{strict} service curve. 
A causal network system $S$ offers a strict service curve $\beta^{\text{strict}}$ if (a) $\beta^{\text{strict}}$ is wide-sense increasing and (b) during any interval $]s,t]$ in which the system is never empty (a so-called \emph{backlogged period}), the output $R^\mathtt{B}$ verifies
\begin{equation}\label{eq:background:strict-sc-def}
    \forall t\ge s \ge 0, \quad R^\mathtt{B}(t) - R^\mathtt{B}(s) \ge \beta^{\text{strict}}(t-s)
\end{equation}

Such a system then also offers $\beta^{\text{strict}}$ as a (simple) service curve \cite[Prop. 1.3.5]{leboudecNetworkCalculus2001}.
We say that a regulator offers a \emph{context-agnostic} service curve [\resp strict service curve] if \eqref{eq:background:sc-def} [\resp \eqref{eq:background:strict-sc-def}] holds for any packetized input $R^\mathtt{A}(t)$, without any other assumption on the upstream systems.

Reciprocally, 
a wide-sense increasing function $\alpha\in\mathfrak{F}_0$ is an arrival curve for the traffic at $\mathtt{A}$ if

\begin{equation}\label{eq:background:strict-ac-def}
    \forall\ t \ge s \ge 0, \quad R^\mathtt{A}(t) - R^\mathtt{A}(s) \le \alpha(t-s)
\end{equation}
We also note $R^\mathtt{A}\sim\alpha$ and say that the traffic is $\alpha$-constrained.
Common arrival curves are of the form \emph{leaky bucket} $\gamma_{r,b}$ with a rate $r$ and a burst $b$: $\forall t>0, \gamma_{r,b}(t)=rt+b$.

Given some arrival-curve and service-curve constraints, network-calculus results provide delay and backlog bounds at a network element.
A common approach for computing end-to-end performance bounds in time-sensitive networks consists in obtaining an arrival-curve model for each flow and a service-curve model for each network element.
Service-curve models for most \ac{IEEE} \ac{TSN} mechanisms can be found in \cite{maileNetworkCalculusResults2020, zhaoLatencyAnalysisMultiple2020}.



\begin{figure}
    \resizebox{\linewidth}{!}{\input{./figures/2023-02-fluid-service-curve}}
    \caption{\label{fig:sys-model:fluid-service-curve} Notations of Definition~\ref{def:fluid-service-curve}. $Z$ offers the \emph{fluid service curve} $\beta$ if it can be modelled as the concatenation of $Z'$ followed by a packetizer, where $Z'$ offers the service curve $\beta$.}
\end{figure}
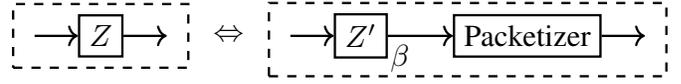


%% file: figures/2023-02-fluid-service-curve.tex
\begin{tikzpicture}
    \tikzstyle{n} = [draw]
    \node[n] at (0,0) (pfr) {$Z$};
    \node[n] at (3,0) (gs) {$Z'$};
    \node[n, anchor=west] at ($(gs.east)+(0.75,0)$) (pl) {Packetizer};

    \draw[->] ($(pfr.west)+(-0.5,0)$) -- (pfr.west) node[pos=0,anchor=center](pfrIn){};
    \draw[->] (pfr.east) -- ($(pfr.east)+(0.5,0)$) node[pos=1,anchor=center](pfrOut){};
    \draw[->] ($(gs.west)+(-0.5,0)$) -- (gs.west) node[pos=0,anchor=center](modelIn){};
    \draw[->] (gs.east) -- (pl.west);
    \draw[->] (pl.east) -- ($(pl.east)+(0.5,0)$) node[pos=1,anchor=center](modelOut){};

    \node[anchor=west, xshift=-0.1cm] at (gs.south east) (gst2) {$\beta$};

    \node[draw, dashed, fit={(pfrIn) (pfr) (pfrOut)}] (pfrBox) {};
    \node[draw, dashed, fit={(gs) (pl) (modelIn) (modelOut) ([yshift=0.2cm] gst2.south)}] (modelBox) {};

    \path (pfr) -- (gs) node[pos=0.5, anchor=center]{$\Leftrightarrow$};
\end{tikzpicture}

%% file: 2-background/2-packetizers.tex
\subsection{The Packetizer and Fluid Service Curves}
\label{sec:background:fluid-sc}

In packet-switching time-sensitive networks, the stream of data at an observation point $\mathtt{A}$ can either be fluid (bit-by-bit, \eg on the transmission links) or packetized (packet-by-packet, \eg within the switches).
A packetizer 
transforms a fluid stream into a packetized stream by releasing the packet's bits only when the last bit is received.
It does 
not increase the end-to-end latency bounds \cite[Thm. 1.7.5]{leboudecNetworkCalculus2001}.

When a system $Z$ with packetized input and output can be split into a fluid service-curve element followed by a packetizer, we say that $Z$ offers a \emph{fluid} service curve.

\begin{definition}[Fluid service curve\label{def:fluid-service-curve}]
    Consider a function $\beta\in\mathfrak{F}_0$ and a system $Z$ with packetized input and output.
    We say that $Z$ offers $\beta$ as a \emph{fluid service curve} if there exists a system $Z'$ that offers the service curve $\beta$ such that $Z$ can be realized
    as the concatenation of $Z'$, followed by a packetizer (Figure~\ref{fig:sys-model:fluid-service-curve}).    
\end{definition}



%% file: 2-background/2-residual-service-curve.tex
\subsection{Individual Service Curve for a Flow}
\label{sec:background:indiv-sc}

In time-sensitive networks, the service modeled by the service curves of Sections~\ref{sec:background:nc-sc} and \ref{sec:background:fluid-sc} is shared between the flows of the aggregate $\mathcal{F}$.
In \eqref{eq:background:sc-def}, the cumulative arrival function $R^{\mathtt{A}}$ of the aggregate at $\mathtt{A}$ is the sum of the individual arrival functions $\{R_f^{\mathtt{A}}\}_{f\in\mathcal{F}}$ for each flow $f$ in the aggregate $\mathcal{F}$:
$\forall t\in\mathbb{R}^+, R^{\mathtt{A}}(t)=\sum_{f\in\mathcal{F}}R_f^{\mathtt{A}}(t)$.

We say that a system $S$ offers to flow $g$ the \emph{individual} arrival curve $\beta_g$ if 
(a) $\beta_g$ is wide-sense increasing, and 
(b) for any cumulative function $R_g^{\mathtt{A}}$ of the flow $g$ at the input $\mathtt{A}$ of $S$, the cumulative function $R_g^{\mathtt{B}}$ of $g$ at its output $\mathtt{B}$ verifies
\begin{equation}\label{eq:background:indiv-sc:def}
    \forall t\ge 0, R_g^{\mathtt{B}}(t)\ge (R_g^{\mathtt{A}} \otimes \beta_g)(t)
\end{equation}

%% file: 2-background/2-traffic-regulators-sff.tex
\subsection{Traffic Regulators and their Shaping-For-Free Properties}
\label{sec:related-work:traffic-reg-sff}

Traffic regulators are hardware elements placed before a multiplexing stage to remove the increased burstiness due to interference with other flows in previous hops. 
They enable the computation of guaranteed latency bounds in networks with cyclic dependencies \cite{mohammadpourLatencyBacklogBounds2018,spechtUrgencyBasedSchedulerTimeSensitive2016, thomasCyclicDependenciesRegulators2019}. They come in two flavors.

A \acf{PFR} is a causal, lossless, \ac{FIFO} system configured for a unique flow $f$ with a shaping curve $\sigma_f$ (Figure~\ref{fig:background:pfr-ir}a).
It stores the packets of $f$ in order of arrival and releases the head-of-line (\acs{HOL}) packet at the earliest time such that the resulting output has $\sigma_f$ as an arrival curve.
In a network with multiple flows, there is one \ac{PFR} per flow.

The \acf{IR} is a causal, lossless, and \ac{FIFO} system that processes an aggregate $\mathcal{F}=\{f_1,f_2,\dots\}$ of several flows, each one with its own shaping curve ($\sigma_{f_1},\sigma_{f_2},\dots$, see Figure~\ref{fig:background:pfr-ir}b).
It stores all the packets of the aggregate $\mathcal{F}$ in order of arrival into a single \ac{FIFO} queue and only looks at the \acf{HOL} packet.
The \ac{HOL} packet $p$ is released as soon as doing so does not violate the configured shaping curve for the associated flow $f_i$:
Packet $p$ can either be immediately released (if the resulting traffic for $f_i$ at the \ac{IR}'s output is $\sigma_{f_i}$-constrained) or delayed to the earliest date that ensures that $f_i$ is $\sigma_{f_i}$-constrained at the \ac{IR}'s output.
This delay depends on the shaping curve $\sigma_{f_i}$ for the associated flow $f_i$ and the history of departure dates for previous packets of the same flow.
During this delay, any other packet $p'$ in the queue is blocked by the \ac{HOL} packet $p$, even if $p'$ belongs to another flow $f_j$ and even if $p'$ could be immediately released without violating the shaping curve $\sigma_{f_j}$ for its flow $f_j$.

\begin{figure}
    \resizebox{\linewidth}{!}{\input{./figures/2023-02-pfr-ir}}
    \caption{\label{fig:background:pfr-ir} Two flavors of traffic regulators.
    With \acp{PFR}, we need one \acs{PFR} per flow.
    In contrast, the \acs{IR} uses a single \acs{FIFO} queue to shape several flows.}
\end{figure}
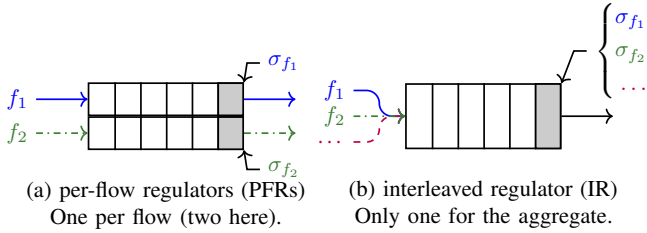 

Traffic regulators can delay individual packets, but there exist specific conditions in which they do not increase the worst-case latency bounds of the flows.
This fundamental \emph{shaping-for-free} property is central in the analysis of time-sensitive networks with traffic regulators.
It slightly differs for the two flavors.

\begin{theorem}[\label{thm:pfr-for-free}Shaping-for-free property of the \ac{PFR} {\cite[Thm. 3]{leboudecTheoryTrafficRegulators2018}}]
    Consider a flow $f$ with input arrival curve $\alpha_f$ that crosses in sequence a causal system $S$ followed by a \ac{PFR} (Figure~\ref{fig:related-work:shaping-for-free:pfr}).
    If the \ac{PFR} is configured with $\sigma_f\ge \alpha_f$
    and if $S$ is \ac{FIFO} for $f$, then the worst-case delay $D_{f,S+\ac{PFR}}$ of $f$ through the concatenation is equal to the worst-case delay $D_{f,S}$ of the flow through the previous system $S$.
\end{theorem}

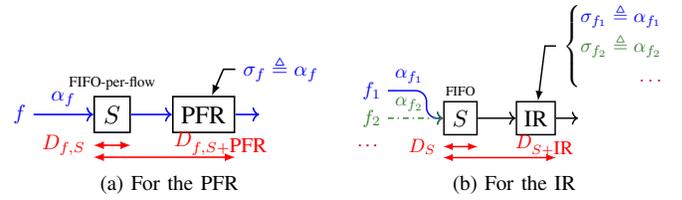
\begin{figure}\centering
    \def\lsize{0.49\linewidth}
    \subfloat[For the \acs{PFR}]{
        \label{fig:related-work:shaping-for-free:pfr}
        \resizebox{\lsize}{!}{\input{./figures/2022-02-sff-pfr.tex}}
    }   
    \subfloat[For the \acs{IR}]{
        \label{fig:related-work:shaping-for-free:ir}
        \resizebox{\lsize}{!}{\input{./figures/2022-02-sff-ir.tex}}
    }
    \caption{\label{fig:related-work:shaping-for-free} Shaping-for-free properties of the traffic regulators. 
    For the \ac{PFR}, the system $S$ only needs to be \ac{FIFO}-per-flow.
    For the \ac{IR}, $S$ must be \ac{FIFO} for the aggregate.}
\end{figure}

\begin{theorem}[\label{thm:ir-for-free}Shaping-for-free property of the \ac{IR} {\cite[Thm. 4]{leboudecTheoryTrafficRegulators2018}}]
    Consider an aggregate $\mathcal{F}=\{f_1,f_2,\dots\}$ with input arrival curves $\{\alpha_{f}\}_{f\in\mathcal{F}}$ that crosses in sequence a causal system $S$ followed by an \ac{IR} (Figure~{\ref{fig:related-work:shaping-for-free:ir}}).
    If the \ac{IR} is configured with $\forall i, \sigma_{f_i}\ge \alpha_{f_i}$ and if $S$ is \ac{FIFO} for the aggregate, then the worst-case delay $D_{S+\ac{IR}}$ of the aggregate $\mathcal{F}$ through the concatenation is equal to the worst-case delay $D_{S}$ of the aggregate through the previous system $S$ only.
\end{theorem}

Theorems~\ref{thm:pfr-for-free} and \ref{thm:ir-for-free} exhibit two fundamental differences.
First, Theorem \ref{thm:ir-for-free} only ensures that the worst-case delay of the aggregate is not increased, whereas Theorem \ref{thm:pfr-for-free} guarantees that the worst-case delay of the individual flow is preserved. 
Within an aggregate, the first bound can be larger than the latter, \eg when the flows have different packet sizes.
Second, Theorem \ref{thm:pfr-for-free} only requires the previous system $S$ to be \ac{FIFO} for each flow individually (\ac{FIFO}-per-flow), whereas the same system is required to be globally \ac{FIFO} for Theorem~\ref{thm:ir-for-free}.




%% file: figures/2023-02-pfr-ir.tex
\begin{tikzpicture}
    \tikzstyle{reg} = [minimum width=2.4cm, minimum height=1cm, draw]
    \tikzstyle{regpfr} = [minimum width=2.4cm, minimum height=0.5cm, draw]
    \tikzstyle{f} = [blue]
    \tikzstyle{g} = [OliveGreen, dashdotted]
    \tikzstyle{h} = [dashed, purple]

    \node[regpfr,anchor=south] at (0,0) (pfr) {};
    \node[regpfr,anchor=north] at (0,0) (pfrb) {};
    \node[anchor=west, reg] at ($(pfr.south east)+(2.5,0)$) (ir) {};

    \foreach \n in {pfr,ir,pfrb}{
        \foreach \x in {0.4,0.8,1.2,1.6,2.0}{
            \draw ($(\n.north east)+(-\x,0)$) -- ($(\n.south east)+(-\x,0)$);
        }
        \draw[fill={black!20}] ([xshift=-\pgflinewidth/2,yshift=-\pgflinewidth/2]\n.north east) --  ([yshift=-\pgflinewidth/2]$(\n.north east)+(-0.4,0)$) -- ([yshift=\pgflinewidth/2]$(\n.south east)+(-0.4,0)$) -- ([xshift=-\pgflinewidth/2,yshift=\pgflinewidth/2]\n.south east) -- cycle;
    }

    \draw[->, f] ($(pfr.west)+(-0.8,0)$) -- (pfr.west) node[pos=0,left]{$f_1$};
    \draw[->, f] (pfr.east) -- ($(pfr.east)+(0.8,0)$);
    \draw[->, g] ($(pfrb.west)+(-0.8,0)$) -- (pfrb.west) node[pos=0,left]{$f_2$};
    \draw[->, g] (pfrb.east) -- ($(pfrb.east)+(0.8,0)$);
    \draw[->, rounded corners=0.2cm, f] ($(ir.west)+(-0.8,0.4)$) -|  ($(ir.west)+(-0.4,0)$)  node[pos=0,left]{$f_1$} -- (ir.west);
    \draw[->, rounded corners=0.2cm, h] ($(ir.west)+(-0.8,-0.4)$) -|  ($(ir.west)+(-0.4,0)$)  node[pos=0,left]{\dots} -- (ir.west);
    \draw[->, g] ($(ir.west)+(-0.8,0)$) -- (ir.west) node[pos=0,left]{$f_2$};
    \draw[->] (ir.east) -- ($(ir.east)+(0.8,0)$);

    \node[anchor=north] at ([yshift=-0.3cm] pfrb.south) {\makecell{(a)~\acfp{PFR}\\One per flow (two here).}};
    \node[anchor=north] at ([yshift=-0.3cm] ir.south) {\makecell{(b)~\acf{IR}\\Only one for the aggregate.}};

    \node[anchor=east, f] at ($(pfr.north east)+(1,0.3)$) (pfrConf) {$\sigma_{f_1}$};
    \node[anchor=east, g] at ($(pfrb.south east)+(1,-0.3)$) (pfrConfb) {$\sigma_{f_2}$};
    \node[anchor=east] at ($(ir.north east)+(1.5,0.5)$) (irConf) {$\left\lbrace\begin{aligned}{\color{blue}\sigma_{f_1}}\\{\color{OliveGreen}\sigma_{f_2}}\\{\color{purple} \dots}\end{aligned}\right.$};

    \draw[->] (pfrConf.west) -- ++(-0.2,0) -- (pfr.north east);
    \draw[->] (pfrConfb.west) -- ++(-0.2,0) -- (pfrb.south east);
    \draw[->] (irConf.west) --  ++(-0.2,0) -- (ir.north east); 

\end{tikzpicture}

%% file: figures/2022-02-sff-pfr.tex
\begin{tikzpicture}
    \tikzstyle{n} =[draw, minimum height=0.5cm, minimum width=0.5cm, scale=1.2]
    \tikzstyle{f} = [blue]
    \tikzstyle{g} = [red, dashdotted]
    \tikzstyle{h} = [dashed, purple]

    \node at (1,0) (xv) {};
    \node at (0.4,0) (xv2) {};

    \node[n] at (0,0) (pfr) {\ac{PFR}};
    \node[n] at ($(pfr.west) - (xv)$) (s) {$S$};
    
    \node[anchor=south, scale=0.7] at (s.north) {\ac{FIFO}-per-flow};

    \draw[->, f] ($(s.west) - (xv)$) -- (s.west) node[pos=0,left]{$f$} node[pos=0.5,above] {$\alpha_f$};
    \draw[->, f] (s.east) -- (pfr.west);
    \draw[->, f] (pfr.east) -- ($(pfr.east) + (xv2)$);

    \draw[latex-latex,red] ($(s.south west) + (0,-0.2)$) --  ($(s.south east) + (0,-0.2)$) node[pos=0, left] {$D_{f,S}$};
    \draw[latex-latex,red] ($(s.south west) + (0,-0.4)$) --  ($(pfr.south east) + (0,-0.4)$) node[pos=0.9, anchor=south, yshift=-0.1cm] {$D_{f,S+\acs{PFR}}$};

    \node[anchor=south west, f] at ($(pfr.north east) + (0,0.1)$) (pfrconf) {$\sigma_f\triangleq \alpha_f$};
    \draw[-latex] (pfrconf.west) -- ++(-0.2,0) -- (pfr);

\end{tikzpicture}

%% file: figures/2022-02-sff-ir.tex
\begin{tikzpicture}
    \tikzstyle{n} =[draw, minimum height=0.5cm, minimum width=0.5cm, scale=1.2]
    \tikzstyle{f} = [blue]
    \tikzstyle{g} = [OliveGreen, dashdotted]
    \tikzstyle{h} = [dashed, purple]

    \node at (1,0) (xv) {};
    \node at (0.3,0) (xvhalf) {};
    \node at (0.4,0) (xv2) {};
    \node at (0,0.5) (yv) {};

    \node[n] at (0,0) (ir) {\ac{IR}};
    \node[n] at ($(ir.west) - (xv)$) (s) {$S$};
    
    \node[anchor=south, scale=0.7] at (s.north) {\ac{FIFO}};

    \draw[->, f, rounded corners=0.2cm] ($(s.west) - (xv) + (yv)$) -| ($(s.west) - (xvhalf)$)  node[pos=0,left]{$f_1$} node[pos=0,anchor=south west] {$\alpha_{f_1}$} -- (s.west);
    \draw[->, g] ($(s.west) - (xv)$) -- (s.west) node[pos=0,left]{$f_2$} node[pos=0,anchor=south west] {$\alpha_{f_2}$};
    \path[->, h, rounded corners=0.2cm] ($(s.west) - (xv) - (yv)$) -| ($(s.west) - (xvhalf)$)  node[pos=0,left]{\dots} -- (s.west);
    \draw[->] (s.east) -- (ir.west);
    \draw[->] (ir.east) -- ($(ir.east) + (xv2)$);

    \draw[latex-latex,red] ($(s.south west) + (0,-0.2)$) --  ($(s.south east) + (0,-0.2)$) node[pos=0, left] {$D_{S}$};
    \draw[latex-latex,red] ($(s.south west) + (0,-0.4)$) --  ($(ir.south east) + (0,-0.4)$) node[pos=0.9, anchor=south, yshift=-0.1cm] {$D_{S+\acs{IR}}$};

    \node[anchor=south west] at ($(ir.north east) + (0,0.1)$) (irconf) {$\left\lbrace\begin{aligned}
        {\color{blue}\sigma_{f_1} \triangleq \alpha_{f_1}}\\
        {\color{OliveGreen}\sigma_{f_2} \triangleq \alpha_{f_2}}\\
        {\color{purple} \dots}
    \end{aligned}\right.$};
    \draw[-latex] (irconf.west) -- ++(-0.2,0) -- (ir);

\end{tikzpicture}

%% file: 3-related-work/0-related-work.tex
\section{Related Work on the Modeling of\\ Traffic Regulators}
\label{sec:related-work}

In time-sensitive networks with traffic regulators, end-to-end latency bounds for the flows are obtained by combining the shaping-for-free property for traffic regulators with service-curve-based network-calculus results for other systems.
This differentiated treatment of network elements (traffic regulators vs. other systems with service-curve models) restrains the choice of the end-to-end analysis method.
Methods based on \ac{TFA} \cite[\S 3.2]{schmittDISCONetworkCalculator2006} can be adapted to networks with traffic regulators \cite{mohammadpourLatencyBacklogBounds2018,thomasCyclicDependenciesRegulators2019}.
Other approaches, such as \acf{SFA} \cite[\S 3.3]{schmittDISCONetworkCalculator2006}, \acf{PMOO} \cite{schmittImprovingPerformanceBounds2008} and flow prolongation \cite{bondorfBetterBoundsWorse2017} provide tighter end-to-end latency bounds than \ac{TFA} in several types of networks \cite{bondorfQualityCostDeterministic2017}, but they heavily rely on service-curve models.
In addition, service-curve models provide continuity and differentiability properties, which allows for synthesizing network designs from the performance requirements, as shown by Geyer and Bondorf in \cite{geyerNetworkSynthesisDelay2022}.
Hence, a need exists for obtaining service-curve models for all elements of time-sensitive networks, including traffic regulators such as \acp{PFR} and \acp{IR}.

The \acf{PFR} was introduced under the name \emph{packetized greedy shaper} in \cite[\S 1.7.4]{leboudecNetworkCalculus2001}.
Le Boudec and Thiran proved in \cite[\S 1.7.4]{leboudecNetworkCalculus2001} that if $\sigma_f$ is concave and such that $\lim_{t\rightarrow 0^+} \sigma_f(t)$ is larger than the maximum packet size of $f$, then the \ac{PFR} offers $\sigma_f$ as a fluid service curve \cite[Thm. 1.7.3]{leboudecNetworkCalculus2001}.
This model proves Theorem~\ref{thm:pfr-for-free}.
%
In this paper, we say that the curve $\sigma_f$ \emph{explains the shaping-for-free property} of the \ac{PFR} and we formally define this notion in Section~\ref{sec:fifo-sff:no-explain}.
Due to its network-calculus service-curve model, the behavior of the \ac{PFR} can also be studied in situations where Theorem~\ref{thm:pfr-for-free} does not apply.
In \cite{thomasWorstCaseDelayBounds2022}, the consequence of redundancy mechanisms -- that can affect the \ac{FIFO} property -- is studied, and end-to-end latency bounds are obtained for flows in networks with redundancy mechanisms and \acp{PFR}.
In \cite[\S IV.A]{leboudecTheoryTrafficRegulators2018}, Le Boudec also provides an input-output characterization of the \ac{PFR}. 
This type of model does not rely on the concept of service curve but describes the \ac{PFR}'s output packet sequence as a function of the input packet sequence.




The \acf{IR} was introduced by Specht and Samii under the name \emph{Urgency-Based Scheduler} \cite{spechtUrgencyBasedSchedulerTimeSensitive2016}.
As opposed to the \ac{PFR}, its shaping-for-free property was proved without the concept of service curves: with a trajectorial approach in \cite{spechtUrgencyBasedSchedulerTimeSensitive2016} and with an input-output characterization in \cite[\S V]{leboudecTheoryTrafficRegulators2018}.
The equivalence between the theoretical model of the \ac{IR} and the \ac{TSN} implementation (\emph{Asynchronous Traffic Shaping}, \cite{ieee8021Qcr}) was proved by Boyer in \cite{boyerEquivalenceTheoreticalModel2022}, who also provides a second input-output characterization \cite[\S 3.3]{boyerEquivalenceTheoreticalModel2022}.
The only useful service curves that are known for the \ac{IR} are only valid when the \ac{IR} is placed in a context that meets the conditions of Theorem~\ref{thm:ir-for-free}.
A first context-dependent service curve is provided in \cite[\S IV.A.1]{mohammadpourLatencyBacklogBounds2018} and then slightly improved in \cite[\S III.B.1]{zhaoQuantitativePerformanceComparison2022}.
In contrast, the only context-agnostic service curve known for the \ac{IR} is the trivial function $t\mapsto 0$.
In \cite{hamscherUsingMathematicalProgramming2022}, Hamscher mentions the first conjecture on a non-trivial context-agnostic service curve for the \ac{IR} and uses a mathematical-programming approach for hardening their conjecture pending formal proof.
The conjecture was not shared (and, to our knowledge, has not been published at the time of this writing). 
However, the presentation triggered discussions on whether the \ac{IR}'s behavior could be captured by context-agnostic service curves.
These discussions motivated this paper.

Hence, two questions remain open: 
Beyond the function $\beta:t\mapsto 0$, what other context-agnostic service curves does the \ac{IR} provide? 
Do any of them explain Theorem~\ref{thm:ir-for-free}?
We address these two questions in this paper.

%% file: 4-system-model/0-system-model.tex
\section{System Model and Notations}
\label{sec:sysmodel}

\begin{table}
    \caption{\label{tab:motations} Notations}
    \resizebox{\linewidth}{!}{\input{99-config/notations.tex}}
\end{table}

We consider an asynchronous packet-switching time-sensitive network that contains traffic regulators.
We focus on a particular traffic regulator within this network.
It can either be a \acf{PFR} that processes a single flow $\mathcal{F}=\{f\}$ with shaping curve $\sigma_f$ or an \acf{IR} that processes an aggregate $\mathcal{F}$ with leaky-bucket shaping curves $\{\sigma_f\}_{f\in\mathcal{F}}=\{\gamma_{r_f,b_f}\}_{f\in\mathcal{F}}$.
We focus on the subset of flows $\mathcal{F}$ that cross the regulator.
We model any other network elements (queues, schedulers, switching fabrics, transmission links, \dots) or sequence of network elements crossed by the flows $\mathcal{F}$ between their sources, the regulator, and their destinations as black-box systems.
Each system has a traffic input and a traffic output and is only assumed to be causal and lossless: it neither produces nor loses any data internally.
Data is produced at the flows' sources and consumed at the flows' destinations.

A \emph{trajectory} $x$ is a description of all the events in the network (packet arrival, packet departure).
It is \emph{acceptable} if all known constraints are satisfied.
For an observation point $\mathtt{M}$, we denote by $R^{x,\mathtt{M}}$ [\resp $R_f^{x,\mathtt{M}}$] the cumulative function of the aggregate [\resp of the flow $f\in\mathcal{F}$] at observation point $\mathtt{M}$ in trajectory $x$.
If the stream is packetized at $\mathtt{M}$, we call $\mathscr{M}^x$ the packet sequence that describes the packets' arrival date, size, and associated flow at $\mathtt{M}$ in trajectory $x$.

For an input packet sequence $\mathscr{B}^x$ at the input $\mathtt{B}$ of the regulator (Figure~\ref{fig:sysmodel:traffic-regulators:simple-ir}), we use the equivalent input-output characterizations of traffic regulators from \cite{leboudecTheoryTrafficRegulators2018} and \cite{boyerEquivalenceTheoreticalModel2022} to obtain the output packet sequence $\mathscr{D}^x$ at the output $\mathtt{D}$ of the regulator.

We list the notations in Table~\ref{tab:motations}.
\proofEndSysModel{}.

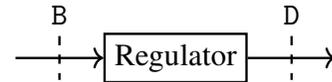
\begin{figure}
    \centering
        \resizebox{0.5\linewidth}{!}{\input{./figures/2023-04-ir-two-observation-points.tex}}
        \caption{\label{fig:sysmodel:traffic-regulators:simple-ir} Input [\resp output] observation point $\mathtt{B}$ [\resp $\mathtt{D}$] for a regulator.}
\end{figure}

%% file: 99-config/notations.tex
\begin{tabular}{r|ll}
    \multicolumn{3}{c}{Common Operators} \\
    $a \vee b $ & $=\max(a,b)$ & Maximum of $a$ and $b$.\\
    $a \wedge b $ & $=\min(a,b)$ & Minimum of $a$ and $b$.\\
    $|c|^+$ & $=\max(0,c)$ & \\
    $\lfloor x\rfloor$ & $=\max\{n\in\mathbb{N}|n\le x\}$ & Floor function.\\ 
    $\mathfrak{f} \otimes \mathfrak{g}$ & $t\mapsto\inf_{s\le t}\mathfrak{f}(s)+\mathfrak{g}(t-s)$ & Min-plus convolution\\
    $\mathfrak{f}\ \overline{\otimes}\ \mathfrak{g}$ & $t\mapsto\sup_{s\le t}\mathfrak{f}(s)+\mathfrak{g}(t-s)$ & Max-plus convolution\\
    $\mathfrak{f}\ \overline{\oslash}\ \mathfrak{g}$ & $t\mapsto\inf_{u\ge0}\mathfrak{f}(t+u)+\mathfrak{g}(u)$ & Max-plus deconvolution\\
    $\mathfrak{F}_0$ & $=\{\mathfrak{f}:R^+\rightarrow R^+|\mathfrak{f}(0)=0\}$ & Set of curves\\
    \hline
    \multicolumn{3}{c}{Common Curves} \\
    $\gamma_{r,b}$  & $t\mapsto \left\lbrace \begin{aligned} 0 &\quad \text{ if } t = 0\\ rt+b &\quad \text{ if } t>0\end{aligned}\right.$      &   Leaky-bucket curve.\\
    $\beta_{R,T}$   & $t\mapsto R|t-T|^+$ & Rate-latency curve. \\
    \hline
    \multicolumn{3}{c}{Flows} \\
    $f\in\mathcal{F}$ & A flow $f$ in the set of flows $\mathcal{F}$\\
    $\{\sigma_f\}_{f\in\mathcal{F}}$ &  \multicolumn{2}{l}{A set of shaping curves for the flows $\mathcal{F}$} \\
    $L_f^{\min}$, $L_f^{\max}$ &  \multicolumn{2}{l}{Minimum [\resp maximum] packet size of flow $f$} \\
    \hline
    \multicolumn{3}{c}{Trajectory Description} \\
    $x$ & \multicolumn{2}{l}{A trajectory: Description of all the events in the network} \\
    $\mathtt{M}$ & \multicolumn{2}{l}{An observation point} \\
    $\mathscr{M}^x$ & \multicolumn{2}{l}{Packet sequence at $\mathtt{M}$ in trajectory $x$} \\
    $R^{x,\mathtt{M}}$ & \multicolumn{2}{l}{Cumulative function of the aggregate\dots} \\
    {[\resp $R^{x,\mathtt{M}}_f$]} & \multicolumn{2}{l}{\dots[\resp of $f$] at $\mathtt{M}$ in Trajectory $x$.} \\
    $R^{x,\mathtt{M}} \sim \alpha$ & \multicolumn{2}{l}{$R^{x,\mathtt{M}}$ is constrained by $\alpha$, Equation \eqref{eq:background:strict-ac-def}} \\
    \hline
    \multicolumn{3}{c}{Parameters of the Spring adversary (Section~\ref{sec:fifo-sff:spring})} \\
    $I$ &  \multicolumn{2}{l}{Expected spacing for same-flow packets after the \ac{IR}} \\
    $d$ &  \multicolumn{2}{l}{Maximum delay in the Spring-controlled system $S_1$} \\
    $\epsilon$ &  \multicolumn{2}{l}{Margin (minimum packet spacing after $S_1$)} \\
    $\tau$ & \multicolumn{2}{l}{Period of the six-packet-long profile}

\end{tabular}

%% file: figures/2023-04-ir-two-observation-points.tex
\begin{tikzpicture}
    \node[draw] at (0,0) (ir) {Regulator};
    \draw[->] ($(ir.west)+(-1,0)$) -- (ir.west) node[pos=0.5,anchor=center](tin){};
    \draw[->] (ir.east) -- ($(ir.east)+(1,0)$)  node[pos=0.5,anchor=center](tout){};

    \node at (0,0.25) (vv) {};
    \draw[dashed] ($(tin)+(vv)$) -- ($(tin)-(vv)$)      node[pos=0,above]{$\mathtt{B}$};
    \draw[dashed] ($(tout)+(vv)$) -- ($(tout)-(vv)$)    node[pos=0,above]{$\mathtt{D}$};
\end{tikzpicture}

%% file: 5-fifo-sff/0-fifo-sff.tex
\section{Limits of the Shaping-For-Free Property for the Interleaved Regulator}
\label{sec:fifo-sff}

The shaping-for-free property is a strong attribute of the \acf{IR}.
However, it is context dependent: 
It makes assumptions on the context in which the \ac{IR} is placed.
In this section, we investigate the limits of these assumptions.

First, we observe that Theorem~\ref{thm:ir-for-free} requires the upstream system to be \ac{FIFO}.
In Section~\ref{sec:fifo-sff:spring}, we prove that removing this assumption makes the \ac{IR} unstable: it can yield unbounded latencies.
We then prove in Section~\ref{sec:fifo-sff:no-explain} that there exists no service-curve model of the \ac{IR} that can explain Theorem~\ref{thm:ir-for-free}.

\input{5-fifo-sff/1-non-fifo.tex}

\input{5-fifo-sff/2-nc-service-curve-model.tex}

%% file: 5-fifo-sff/1-non-fifo.tex
\subsection{Instability of the \acs{IR} when Placed after a Non-\acs{FIFO} System}
\label{sec:fifo-sff:spring}

In this subsection, we discuss the role of the \ac{FIFO} assumption in Theorem~\ref{thm:ir-for-free}.
When removed, we prove that the \ac{IR} can yield unbounded latencies.
Specifically, we prove


\begin{theorem}[\label{thm:instab}Instability of the \ac{IR} after a non-\ac{FIFO} system]
    Consider an \ac{IR} that processes three or more flows with the same leaky-bucket shaping curve for the first three flows: $\forall f_i\in\{f_1,f_2,f_3\}, \sigma_{f_i}=\gamma_{r,b}$ with $r>0$ and $b$ greater than the maximum packet size of $f_1,f_2,f_3$.
    For any $D>0$, there exists a system $S_1$ and a source $\phi$ (Figure~\ref{fig:fifo-sff:spring:n1}), such that:
    \begin{enumerate}[1/]
        \item each flow $f_i$ is $\sigma_{f_i}$-constrained at the source $\phi$, \label{enum:spring:source-constraint} \label{enum:spring:first}
        \item $S_1$ is causal, lossless and \ac{FIFO}-per-flow, \label{enum:spring:s-fifo-per-flow-lossless}
        \item when the system $S_1$ is placed 
        after the source as in Figure~\ref{fig:fifo-sff:spring:n1}, then the delay of each flow within $S_1$ is upper-bounded by $D$, \label{enum:spring:delay-through-s}
        \item when the \ac{IR} is placed 
        after $S_1$ as in Figure~\ref{fig:fifo-sff:spring:n1}, then the delay of any flow within the \ac{IR} is not bounded.\label{enum:spring:instable}  \label{enum:spring:last}
    \end{enumerate}
\end{theorem}

The proof of Theorem~\ref{thm:instab} in \proofSpring{} relies on an adversarial traffic generation that we call ``Spring''. 
%
Spring is an adversary that knows the values of $b,r$ and $D$ in Theorem~\ref{thm:instab} and controls the source $\phi$ and the system $S_1$ of Figure~\ref{fig:fifo-sff:spring:n1} such that Properties~\ref{enum:spring:first} to \ref{enum:spring:last} of Theorem~\ref{thm:instab} hold.
It defines the constants $I,d,\epsilon$ and $\tau$ as follows
\begin{equation}\label{eq:spring-constants}
    \resizebox{\linewidth}{!}{$I\triangleq \frac{b}{r}; 0<d<\min\left(I,D\right); 0<\epsilon<\min(I-d,\frac{d}{3}); \tau\triangleq 3I+3\epsilon-d$}
\end{equation}


\begin{figure}
    \resizebox{\linewidth}{!}{\input{./figures/2022-11-network-n1}}
    \caption{\label{fig:fifo-sff:spring:n1} The network $\mathcal{N}_1$ and the Spring-generated Trajectory 1 that yields unbounded latencies in the \acs{IR} when $S_1$ is not assumed \ac{FIFO}.}
\end{figure}
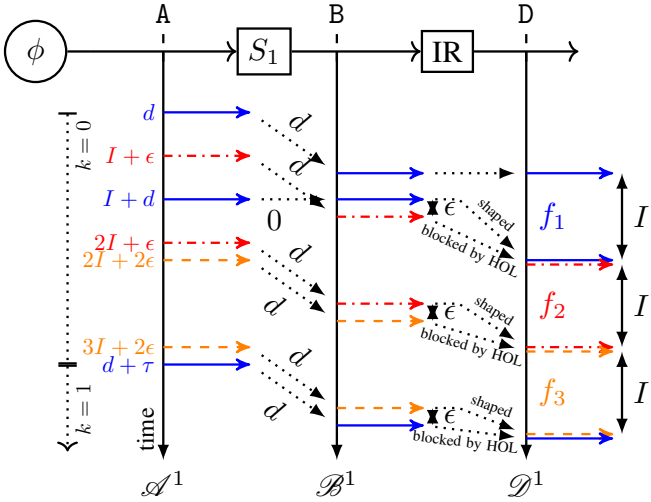



\paragraph*{Intuition}

Trajectory $1$ generated by Spring is illustrated in Figure~\ref{fig:fifo-sff:spring:n1}.
All packets have the size $b$.
The far-left timeline shows the packet sequence $\mathscr{A}^1$ for the three flows at the output of the Spring-controlled source.
A sequence of six packets is repeated with period $\tau$.
Only the period $k=0$ is shown.

The dotted arrows that lead to the second timeline highlight each packet's delay in the Spring-controlled system $S_1$ and the resulting packet sequence $\mathscr{B}^1$.
The main property of Trajectory $1$ is that the first packet of the dash-dotted red flow $f_2$ and the second packet of the solid blue flow $f_1$ have exchanged their order at $\mathtt{B}$ compared to their order at $\mathtt{A}$.
This is because the former suffers a delay $d$ through $S_1$, but the latter does not suffer any delay.
Note that the Spring-controlled system $S_1$ is not \ac{FIFO} but remains causal, lossless, and \ac{FIFO}-per-flow with a delay bound $d< D$.

The dotted arrows that link the second to the third timeline describe the behavior of the \ac{IR} (not controlled by Spring) when provided with the input sequence $\mathscr{B}^1$.
For example, the first packet of $f_1$ is immediately released by the \ac{IR} because the network was previously empty.
However, the second solid blue packet of $f_1$ is shaped (delayed) by the \ac{IR} because releasing it would violate the $\gamma_{r,b}$ shaping constraint for $f_1$ at the output of the \ac{IR}.
This packet is released as soon as doing so does not violate the $\gamma_{r,b}$ constraint, \ie $I$ seconds after the previous packet.
Because of this, the first dash-dotted red packet of the flow $f_2$ is blocked by the \acf{HOL}.
And the second packet of $f_2$ is shaped and delayed to ensure a distance of $I$ from the previous packet of $f_2$.

As a result, it takes $3I$ seconds for the \ac{IR} to output the six packets of the first period, whereas they entered the \ac{IR} within $\tau$ seconds. 
As $\tau < 3I$, we can generate a constant build-up of delay and backlog in the \ac{IR} by repeating the six-packet-long profile every $\tau$ seconds.






%% file: figures/2022-11-network-n1.tex
\begin{tikzpicture}
    \tikzstyle{n} = [draw]

    \node[n] at (1,0) (s1) {$S_1$};
    \node[n, anchor=west] at (2.8,0) (ir1) {$\text{IR}$};
    \node[anchor=east] at (-1,0) (n1) {};
    \node[anchor=east, circle, draw] at (n1.west) (phi) {$\phi$};
    \draw[-] (phi) -- (n1.east);
    \draw[->] (n1.east) -- (s1.west) node[pos=0.5,anchor=center](ta1){};
    \draw[->] (s1.east) -- (ir1.west) node[pos=0.35](tb1){};
    \draw[->] (ir1.east) -- ($(ir1.east)+(1.2,0)$)node[pos=0.5](td1){};

    \draw[dashed] let \p3=(ta1.center) in (\x3,0.2) -- (\x3,-0.2) node[pos=0,above]{$\mathtt{A}$};
    \draw[dashed] let \p3=(tb1.center) in (\x3,0.2) -- (\x3,-0.2) node[pos=0,above]{$\mathtt{B}$};
    \draw[dashed] let \p3=(td1.center) in (\x3,0.2) -- (\x3,-0.2) node[pos=0,above]{$\mathtt{D}$};

    \def\mlinelength{-4.7}
    \draw[-latex] (ta1.center) -- ($(ta1.center)+(0,\mlinelength)$) node[pos=1,rotate=90,scale=0.8,anchor=south west]{time} node[pos=1,anchor=north]{$\mathscr{A}^1$};
    \draw[-latex] (tb1.center) -- ($(tb1.center)+(0,\mlinelength)$)  node[pos=1,anchor=north]{$\mathscr{B}^1$};
    \draw[-latex] (td1.center) -- ($(td1.center)+(0,\mlinelength)$)  node[pos=1,anchor=north]{$\mathscr{D}^1$};

    \tikzstyle{f1} = [-stealth', blue]
    \tikzstyle{f2} = [-stealth', red, dashdotted]
    \tikzstyle{f3} = [-stealth', orange, dashed]
    
    \def\myI{1}
    \def\myD{0.7}
    \def\myEpsilon{0.2}
    \tikzstyle{ldd} = [scale=0.75]
    \draw[f1] ($(ta1.center)-(0,\myD)$) -- ++ (1,0)                     node[pos=0, left, ldd] {$d$}                  node[pos=1,anchor=center](1aa){}; 
    \draw[f2] ($(ta1.center)-(0,\myI+\myEpsilon)$) -- ++ (1,0)          node[pos=0, left, ldd] {$I+\epsilon$}        node[pos=1,anchor=center](2aa){};   
    \draw[f1] ($(ta1.center)-(0,\myI+\myD)$) -- ++ (1,0)                node[pos=0, left, ldd] {$I+d$}               node[pos=1,anchor=center](1ba){};
    \draw[f2] ($(ta1.center)-(0,2*\myI+\myEpsilon)$) -- ++ (1,0)        node[pos=0, left, ldd] {$2I+\epsilon$}       node[pos=1,anchor=center](2ba){};
    \draw[f3] ($(ta1.center)-(0,2*\myI+2*\myEpsilon)$) -- ++ (1,0)      node[pos=0, left, ldd] {$2I+2\epsilon$}      node[pos=1,anchor=center](3aa){};
    \draw[f3] ($(ta1.center)-(0,3*\myI+2*\myEpsilon)$) -- ++ (1,0)      node[pos=0, left, ldd] {$3I+2\epsilon$}      node[pos=1,anchor=center](3ba){};

    \draw[f1] ($(ta1.center)-(0,\myD+3*\myI+3*\myEpsilon-\myD)$) -- ++ (1,0)                     node[pos=0,left, ldd] {$d+\tau$}                  node[pos=1,anchor=center](1aak1){};

    \draw[f1] ($(tb1.center)-(0,2*\myD)$) -- ++ (1,0)                   node[pos=0,anchor=center](1ab){} node[pos=1,anchor=center](1abEnd){}; 
    \draw[f1] ($(tb1.center)-(0,\myI+\myD)$) -- ++ (1,0)                node[pos=0,anchor=center](1bb){} node[pos=1,anchor=center](1bbEnd){}; 
    \draw[f2] ($(tb1.center)-(0,\myI+\myD+\myEpsilon)$) -- ++ (1,0)     node[pos=0,anchor=center](2ab){} node[pos=1,anchor=center](2abEnd){}; 
    \draw[f2] ($(tb1.center)-(0,2*\myI+\myEpsilon+\myD)$) -- ++ (1,0)   node[pos=0,anchor=center](2bb){} node[pos=1,anchor=center](2bbEnd){}; 
    \draw[f3] ($(tb1.center)-(0,2*\myI+2*\myEpsilon+\myD)$) -- ++ (1,0) node[pos=0,anchor=center](3ab){} node[pos=1,anchor=center](3abEnd){}; 
    \draw[f3] ($(tb1.center)-(0,3*\myI+2*\myEpsilon+\myD)$) -- ++ (1,0) node[pos=0,anchor=center](3bb){} node[pos=1,anchor=center](3bbEnd){}; 

    \draw[f1] ($(tb1.center)-(0,2*\myD+3*\myI+3*\myEpsilon-\myD)$) -- ++ (1,0)                   node[pos=0,anchor=center](1abk1){} node[pos=1,anchor=center](1abk1End){}; 

    \tikzstyle{da} = [-latex, dotted]
    \draw[da] (1aa) -- (1ab) node[pos=0.4, above, sloped]{$d$};
    \draw[da] (1ba) -- (1bb) node[pos=0.2, below, sloped]{$0$};
    \draw[da] (2aa) -- (2ab) node[pos=0.4, above, sloped]{$d$};
    \draw[da] (2ba) -- (2bb) node[pos=0.4, above, sloped]{$d$};
    \draw[da] (3aa) -- (3ab) node[pos=0.4, below, sloped]{$d$};
    \draw[da] (3ba) -- (3bb) node[pos=0.4, above, sloped]{$d$};
    \draw[da] (1aak1) -- (1abk1) node[pos=0.4, below, sloped]{$d$};

    \draw[{latex[length=1mm, width=1mm]}-{latex[length=1mm, width=1mm]}] ($(1bbEnd.center)+(0.1,0)$) -- ($(2abEnd.center)+(0.1,0)$) node[pos=0.5, right] {$\epsilon$};
    \draw[{latex[length=1mm, width=1mm]}-{latex[length=1mm, width=1mm]}] ($(2bbEnd.center)+(0.1,0)$) -- ($(3abEnd.center)+(0.1,0)$) node[pos=0.5, right] {$\epsilon$};
    \draw[{latex[length=1mm, width=1mm]}-{latex[length=1mm, width=1mm]}] ($(3bbEnd.center)+(0.1,0)$) -- ($(1abk1End.center)+(0.1,0)$) node[pos=0.5, right] {$\epsilon$};
 
    \draw[f1] ($(td1.center)-(0,2*\myD)$) -- ++ (1,0)                       node[pos=0,anchor=center](1ad){} node[pos=1,anchor=center](1adEnd){}; 
    \draw[f1] ($(td1.center)-(0,2*\myD+\myI)$) -- ++ (1,0)                  node[pos=0,anchor=center](1bd){} node[pos=1,anchor=center](1bdEnd){}; 
    \draw[f2] ($(td1.center)-(0,2*\myD+\myI+0.05)$) -- ++ (1,0)             node[pos=0,anchor=center](2ad){} node[pos=1,anchor=center](2adEnd){}; 
    \draw[f2] ($(td1.center)-(0,2*\myD+2*\myI)$) -- ++ (1,0)                node[pos=0,anchor=center](2bd){} node[pos=1,anchor=center](2bdEnd){}; 
    \draw[f3] ($(td1.center)-(0,2*\myD+2*\myI+0.05)$) -- ++ (1,0)           node[pos=0,anchor=center](3ad){} node[pos=1,anchor=center](3adEnd){}; 
    \draw[f3] ($(td1.center)-(0,2*\myD+3*\myI)$) -- ++ (1,0)                node[pos=0,anchor=center](3bd){} node[pos=1,anchor=center](3bdEnd){}; 
    \draw[f1] ($(td1.center)-(0,2*\myD+3*\myI+0.05)$) -- ++ (1,0)           node[pos=0,anchor=center](1adk1){} node[pos=1,anchor=center](1adk1End){}; 

    \draw[{latex[length=1mm, width=1mm]}-{latex[length=1mm, width=1mm]}] ($(1adEnd.center)+(0.1,0)$) -- ($(1bdEnd.center)+(0.1,0)$) node[pos=0.5, right] (if1) {$I$};
    \draw[{latex[length=1mm, width=1mm]}-{latex[length=1mm, width=1mm]}] ($(2adEnd.center)+(0.1,0)$) -- ($(2bdEnd.center)+(0.1,0)$) node[pos=0.5, right] (if2) {$I$};
    \draw[{latex[length=1mm, width=1mm]}-{latex[length=1mm, width=1mm]}] ($(3adEnd.center)+(0.1,0)$) -- ($(3bdEnd.center)+(0.1,0)$) node[pos=0.5, right] (if3) {$I$};

    \node[f1,anchor=east] at ($(if1.west)+(-0.5,0)$) {$f_1$};
    \node[f3,anchor=east] at ($(if3.west)+(-0.5,0)$) {$f_3$};
    \node[f2,anchor=east] at ($(if2.west)+(-0.5,0)$) {$f_2$};

    \draw[da] (1abEnd) -- (1ad); 
    \draw[da] (1bbEnd) -- ($(1bbEnd)+(0.4,0)$) -- (1bd)     node[pos=0.5,above,scale=0.5, sloped]{shaped}; 
    \draw[da] (2abEnd) -- (2ad)                             node[pos=0.5,below,scale=0.5, sloped]{blocked by HOL}; 
    \draw[da] (2bbEnd) -- ($(2bbEnd)+(0.4,0)$) -- (2bd)     node[pos=0.5,above,scale=0.5, sloped]{shaped}; 
    \draw[da] (3abEnd) -- (3ad)                             node[pos=0.5,below,scale=0.5, sloped]{blocked by HOL}; 
    \draw[da] (3bbEnd) -- ($(3bbEnd)+(0.4,0)$) -- (3bd)     node[pos=0.5,above,scale=0.5, sloped]{shaped}; 
    \draw[da] (1abk1End) -- (1adk1)                         node[pos=0.4,below,scale=0.5, sloped]{blocked by HOL}; 

    \draw[dotted, |-|] let \p1=($(ta1.center)+(-1.1,0)$), \p2=(1aa.center), \p3=(1aak1.center) in (\x1,\y2) -- (\x1,\y3) node[pos=0,anchor=north east, rotate=90, scale=0.7]{$k=0$};
    \draw[dotted, |->] let \p1=($(ta1.center)+(-1.1,0)$), \p2=(1aak1.center) in (\x1,\y2) -- ++(0,-1) node[pos=0.9, anchor=north west, rotate=90, scale=0.7]{$k=1$};

\end{tikzpicture}

%% file: 5-fifo-sff/2-nc-service-curve-model.tex
\subsection{The Shaping-for-Free Property of the Interleaved Regulator Cannot be Explained by a Service Curve}
\label{sec:fifo-sff:no-explain}

Theorem~\ref{thm:instab} shows that the \ac{IR} does not provide any context-agnostic delay guarantees as a stand-alone network element.
In contrast, if a system $Z$ offers a context-agnostic service curve $\beta$ that explains a context-dependent property (\eg shaping-for-free), then $\beta$ continues to hold when $Z$ is placed in a context that differs from the assumptions of the context-dependent property.
$\beta$ can be used to compute the consequences of the deviation from the assumptions and their resulting penalties on performance bounds.
For such a system, slight deviations from the assumptions should lead to small delay penalties.

This is the case for the \ac{PFR}, for which we can find a service-curve model that explains its shaping-for-free property.
We formally define this notion as follows.

\begin{definition}[A curve explains the shaping-for-free property\label{def:fifo-sff:no-explain:curve-explains-sff}]
    Consider a set of flows $\mathcal{F}$ and a set of shaping curves $\boldsymbol{\sigma}=\{\sigma_f\}_{f\in\mathcal{F}}$.
    We say that $\beta^{\boldsymbol{\sigma}}\in\mathfrak{F}_0$ explains the shaping-for-free property if and only if: For any causal, lossless and \ac{FIFO} systems $Z'$ and $S$, if $Z'$ offers $\beta^{\boldsymbol{\sigma}}$ as a service curve, then the worst-case delay of the aggregate between $\mathtt{A}$ and $\mathtt{B}$ (Figure~\ref{fig:fifo-sff:no-explain:curve-explains-sff}) over all the trajectories $X=\{x | \forall f \in\mathcal{F}, R_f^{\mathtt{A},x}\sim\sigma_f\}$ equals the worst-case delay between $\mathtt{A}$ and $\mathtt{C}$ over the same set $X$.
\end{definition}
\begin{figure}
\centering
    \resizebox{0.6\linewidth}{!}{\input{./figures/2023-02-curve-explains-sff.tex}}
    \caption{\label{fig:fifo-sff:no-explain:curve-explains-sff} Notations of Definition~\ref{def:fifo-sff:no-explain:curve-explains-sff}. $\beta^{\boldsymbol{\sigma}}$ explains the shaping-for-free property if any system $Z'$ that offers $\beta^{\boldsymbol{\sigma}}$ as a service curve does not increase the worst-case delay of the flows when placed after any \acs{FIFO} system $S$.}
\end{figure}
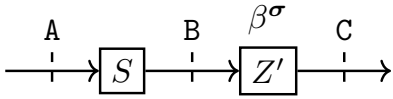

As per this definition, a function $\beta^{\boldsymbol{\sigma}}$ explains the shaping-for-free property if any system $Z'$ that offers $\beta^{\boldsymbol{\sigma}}$  as a service curve does not increase the worst-case delay of the aggregate $\mathcal{F}$ when placed in a context that meets the assumptions of Theorems~\ref{thm:pfr-for-free} and \ref{thm:ir-for-free} (\ie placed after a FIFO system $S$ with flows that are initially constrained by their shaping curves).
For the \ac{PFR}, we have a positive result:

\begin{proposition}\label{prop:fifo-sff:no-explain:pfr-sff-explain}
    Consider a \ac{PFR} that shapes a single flow $\mathcal{F}=\{f\}$ with a concave shaping curve $\sigma_f$ such that $\lim_{t\rightarrow 0}\sigma_f(t)\ge L_f^{\max}$. 
    Then the \ac{PFR} offers the fluid service curve $\beta^{\boldsymbol{\sigma}}=\sigma_f$ that explains the shaping-for-free property.
\end{proposition}

The formal proof in \proofPropPfrScModel{} directly derives from \cite[Thm. 1.7.3]{leboudecNetworkCalculus2001}.
Note that Proposition~\ref{prop:fifo-sff:no-explain:pfr-sff-explain} contains two statements: 
(1) $\sigma_f$ is a fluid service curve of the \ac{PFR}.
(2) $\sigma_f$ explains the shaping-for-free (Definition~\ref{def:fifo-sff:no-explain:curve-explains-sff}).

\begin{figure}\centering
    \def\lsize{0.37\linewidth}
    \def\rsize{0.97\linewidth-\lsize}
    \subfloat[]{
        \label{fig:fifo-sff:no-explain:pfr-model:pfr}
        \resizebox{\lsize}{!}{\input{./figures/2023-03-pfr-single.tex}}
    }%
    \subfloat[]{
        \label{fig:fifo-sff:no-explain:pfr-model:concatenation}
        \resizebox{\rsize}{!}{\input{./figures/2023-03-pfr-concat.tex}}
    }
    \caption{\label{fig:fifo-sff:no-explain:pfr-model} Application of Proposition~\ref{prop:fifo-sff:no-explain:pfr-sff-explain} to prove Theorem~\ref{thm:pfr-for-free}. (a) A \ac{PFR} placed in the conditions of Theorem~\ref{thm:pfr-for-free}. (b) The equivalent model as per Proposition~\ref{prop:fifo-sff:no-explain:pfr-sff-explain}.}
\end{figure}
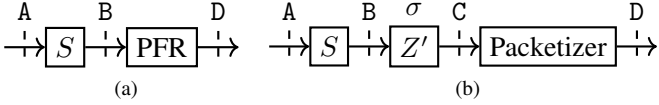

Let us discuss why these two statements prove Theorem~\ref{thm:pfr-for-free}.
Consider a causal, lossless, and \ac{FIFO} system $S$ and a \ac{PFR} configured with $\sigma_f$ placed after $S$ (Figure~\ref{fig:fifo-sff:no-explain:pfr-model:pfr}).
By combining the first statement of Proposition~\ref{prop:fifo-sff:no-explain:pfr-sff-explain} with Definition~\ref{def:fluid-service-curve}, the \ac{PFR} can be realized as the concatenation of $Z'$ followed by a packetizer (Figure~\ref{fig:fifo-sff:no-explain:pfr-model:concatenation}), where $Z'$ is a causal, lossless, and \acs{FIFO} system that offers $\beta^{\boldsymbol{\sigma}}=\sigma_f$ as a service curve.
We then combine the second statement of Proposition~\ref{prop:fifo-sff:no-explain:pfr-sff-explain} with Definition~\ref{def:fifo-sff:no-explain:curve-explains-sff}.
We obtain that if $\sigma_f$ is an arrival curve for $f$ at the input of $S$, then the worst-case delay of the flow $f$ through $S$ equals the worst-case delay of the flow through the concatenation of $S$ and $Z'$.
Finally, the packetizer does not increase the worst-case latency bounds \cite[Thm. 1.7.1]{leboudecTheoryTrafficRegulators2018}, which proves Theorem~\ref{thm:pfr-for-free}.


As opposed to the \ac{PFR}, we prove that no fluid service curve can explain the shaping-for-free property of the \ac{IR}:


\begin{theorem}\label{thm:fifo-sff:no-explain:ir-has-no-explaination}
    An \ac{IR} that processes at least three flows with the same leaky-bucket shaping curve does not have any fluid service curve that explains its shaping-for-free property.
\end{theorem}

To prove Theorem~\ref{thm:fifo-sff:no-explain:ir-has-no-explaination}, we rely on the following lemma, that we prove in \proofLemmaMinimumScExplain{}.

\begin{lemma}\label{lem:fifo-sff:no-explain:minimum-sc-sff}
    If $\beta^{\boldsymbol{\sigma}}$ explains the shaping-for-free (Definition~\ref{def:fifo-sff:no-explain:curve-explains-sff}), then $\beta^{\boldsymbol{\sigma}} \ge \sum_{f\in\mathcal{F}}\sigma_f$
\end{lemma}

By reusing Spring (from the proof of Theorem~\ref{thm:instab}), we then exhibit a trajectory that shows that a function larger than $\sum_{f\in\mathcal{F}}\sigma_f$ cannot be a fluid service curve of the \ac{IR}.
The formal proof of Theorem~\ref{thm:fifo-sff:no-explain:ir-has-no-explaination} is in \proofIrNoScExplain{}.

%% file: figures/2023-02-curve-explains-sff.tex
\begin{tikzpicture}
    \tikzstyle{n} = [draw]

    \node[n] at (0,0) (s) {$S$};
    \node[n, anchor=west] at ($(s.east)+(1,0)$) (z) {$Z'$};
    \node[anchor=south] at (z.north) {$\beta^{\boldsymbol{\sigma}}$};

    \draw[->] ($(s.west)+(-1,0)$) -- (s.west) node[pos=0.5,anchor=center](ta){};
    \draw[->] (s.east) -- (z.west) node[pos=0.5,anchor=center](tb){};
    \draw[->] (z.east) -- ++(1,0) node[pos=0.5,anchor=center](tc){};

    \node (v) at (0,0.2) {};
    \draw[dashed] ($(ta)+(v)$) -- ($(ta)-(v)$) node[pos=0,above]{$\mathtt{A}$};
    \draw[dashed] ($(tb)+(v)$) -- ($(tb)-(v)$) node[pos=0,above]{$\mathtt{B}$};
    \draw[dashed] ($(tc)+(v)$) -- ($(tc)-(v)$) node[pos=0,above]{$\mathtt{C}$};

\end{tikzpicture}

%% file: figures/2023-03-pfr-single.tex
\begin{tikzpicture}
    \node[draw] at (0,0) (s) {$S$};
    \node[draw, anchor=west] at ($(s.east)+(0.5,0)$) (pfr) {PFR};

    \draw[->] ($(s.west)+(-0.5,0)$) -- (s.west) node[pos=0.5,anchor=center](ta){};
    \draw[->] (s.east) -- (pfr.west)            node[pos=0.5,anchor=center](tb){};
    \draw[->] (pfr.east) -- ++(0.5,0)           node[pos=0.5,anchor=center](td){};
    \phantom{
        \node[anchor=center] at (pfr.center) (zprime) {$Z'$};
        \node[anchor=south] at (zprime.north) {$\sigma$}    ;
    }
    \node (v) at (0,0.2) {};
    \draw[dashed] ($(ta.center)+(v)$) -- ($(ta.center)-(v)$) node[pos=0,above]{$\mathtt{A}$};
    \draw[dashed] ($(tb.center)+(v)$) -- ($(tb.center)-(v)$) node[pos=0,above]{$\mathtt{B}$};
    \draw[dashed] ($(td.center)+(v)$) -- ($(td.center)-(v)$) node[pos=0,above]{$\mathtt{D}$};
\end{tikzpicture}

%% file: figures/2023-03-pfr-concat.tex
\begin{tikzpicture}
    \node[draw] at (0,0) (s) {$S$};
    \node[draw, anchor=west, label=above:{$\sigma$}] at ($(s.east)+(0.5,0)$) (z) {$Z'$};
    \node[draw, anchor=west] at ($(z.east)+(0.5,0)$) (pl) {Packetizer};

    \draw[->] ($(s.west)+(-0.5,0)$) -- (s.west) node[pos=0.5,anchor=center](ta){};
    \draw[->] (s.east) -- (z.west)              node[pos=0.5,anchor=center](tb){};
    \draw[->] (z.east) -- (pl.west)             node[pos=0.5,anchor=center](tc){};
    \draw[->] (pl.east) -- ++(0.5,0)            node[pos=0.5,anchor=center](td){};

    \node (v) at (0,0.2) {};
    \draw[dashed] ($(ta.center)+(v)$) -- ($(ta.center)-(v)$) node[pos=0,above]{$\mathtt{A}$};
    \draw[dashed] ($(tb.center)+(v)$) -- ($(tb.center)-(v)$) node[pos=0,above]{$\mathtt{B}$};
    \draw[dashed] ($(tc.center)+(v)$) -- ($(tc.center)-(v)$) node[pos=0,above]{$\mathtt{C}$};
    \draw[dashed] ($(td.center)+(v)$) -- ($(td.center)-(v)$) node[pos=0,above]{$\mathtt{D}$};
\end{tikzpicture}

%% file: 6-service-curves/0-service-curves.tex
\section{Service Curves of The Interleaved Regulator}
\label{sec:ir-service-curves}

In the previous section, we use the sensitivity of the \acf{IR} to the \ac{FIFO} assumption in Theorem~\ref{thm:ir-for-free} to prove that the \ac{IR} has no fluid service curve that can explain its shaping-for-free property.

In this section, we show that the \ac{IR} still offers a family of non-trivial context-agnostic service curves. 
In particular, we exhibit a non-bounded strict service curve for the \ac{IR} (Theorem~\ref{thm:ir-service-curves:strict-sc}).
This strict-service-curve model is of interest for understanding the behavior of the \ac{IR} in situations that differ from the shaping-for-free property.
It can be used to model the \ac{IR} in service-curve-oriented analysis like \ac{SFA} and \ac{PMOO}.

However, any service curves of the \ac{IR} are also fluid service curves of the \ac{IR}, and we know from Theorem~\ref{thm:fifo-sff:no-explain:ir-has-no-explaination} that they must be weak because they cannot explain its shaping-for-free property.
Indeed, their long-term rate is upper bounded:
For an \ac{IR} that processes more than four flows, the long-term rate of any of its service curves is upper bounded by three times the rate enforced for a single flow, as we show in Theorem~\ref{thm:ir-service-curves:limit-sc}.

For an \ac{IR} that processes at least four flows, we also prove that any individual service curve $\beta_g$ offered to a single flow $g$ is upper bounded by its minimum packet size $L_{g}^{\min}$ (Theorem~\ref{thm:ir-service-curves:limit-ir-residual}).

The section is organized as follows. 
First, we obtain a strict service curve of the \ac{IR} (Theorem~\ref{thm:ir-service-curves:strict-sc}) by using the input/output models of \cite{leboudecTheoryTrafficRegulators2018,boyerEquivalenceTheoreticalModel2022}.
Then, we use the Spring trajectory from Section~\ref{sec:fifo-sff} to obtain upper bounds on the individual service curve of the \ac{IR} for a single flow (Theorem~\ref{thm:ir-service-curves:limit-ir-residual}).
Last, we use Theorem~\ref{thm:ir-service-curves:limit-ir-residual} to upper-bound the long-term rate of any context-agnostic service curve of the \ac{IR} for the aggregate (Theorem~\ref{thm:ir-service-curves:limit-sc}).

\input{6-service-curves/1-strict-sc.tex}

\input{6-service-curves/2-individual-service.tex}

\input{6-service-curves/3-sc.tex}

%% file: 6-service-curves/1-strict-sc.tex
\subsection{A Strict Service Curve for the Aggregate}
\label{sec:ir-service-curves:strict-sc}

Even though Spring's Trajectory from Section~\ref{sec:fifo-sff:spring} generates a constant build-up of delay in the \ac{IR}, we can observe in Figure~\ref{fig:fifo-sff:spring:n1} that the \ac{IR} continuously outputs two packets every $I$ seconds.
In fact, we can find a minimum output rate whenever the \ac{IR} is non-empty.
This shows that the \ac{IR} offers a strict service curve as defined in Section~\ref{sec:background:nc-sc}:



\begin{figure}
    \centering
        \resizebox{0.95\linewidth}{!}{\input{./figures/2023-04-strict-service-curves-ir-proof.tex}}
        \caption{\label{fig:ir-service-curves:strict-sc:result-figure} Three different strict service curves of the \ac{IR} (Theorem~\ref{thm:ir-service-curves:strict-sc}).
        The step dashed blue function $\beta_0$ is the first curve obtained in the theorem's proof.
        Its supper-additive closure, the solid red function $\beta_\text{sc}$ is then also a strict service curve, as well as the dash-dotted green rate-latency curve $\beta_{\frac{L^{\min}}{I^{\max}},I^{\max}}$.}
    \end{figure}
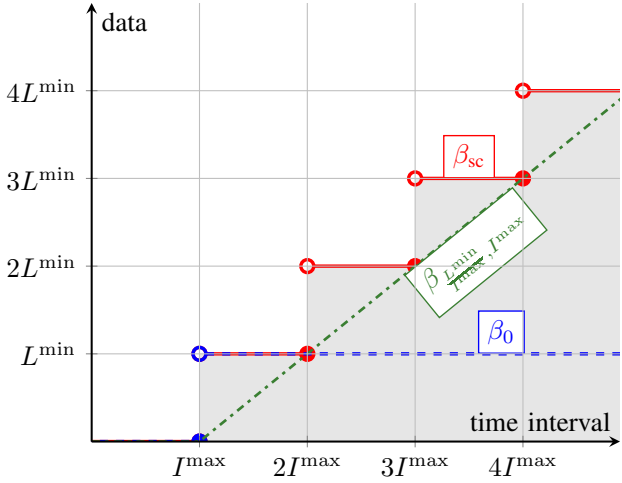

\begin{theorem}[IR strict service curve]\label{thm:ir-service-curves:strict-sc}
    Consider an \ac{IR} that processes an aggregate of flows $\mathcal{F}$ with leaky-bucket shaping curves: $\forall f\in\mathcal{F}, \sigma_f=\gamma_{r_f,b_f}$ with $r_f> 0$ and $b_f\ge L_f^{\max}$, where $L_f^{\min}$ [\resp $L_f^{\max}$] is the minimum [\resp maximum] packet size of $f$.
    Define
    \begin{equation}
        L^{\min} = \min_{f\in\mathcal{F}}L_f^{\min} \quad I^{\max}=\max_{f\in\mathcal{F}}\frac{L_f^{\max}}{r_f}
    \end{equation}
    Then the staircase curve $\beta_{\text{sc}}:t\mapsto \left\lfloor t/I^{\max}\right\rfloor\cdot L^{\min}$ (with $\lfloor\cdot\rfloor$ the floor function) and the rate-latency curve $\beta_{R,T}$ and $T=I^{\max}$ and $R=L^{\min}/I^{\max}$ (Figure~\ref{eq:ir-service-curves:strict-sc:supper-additive-closure}) are context-agnostic strict service curves of the \ac{IR}.
\end{theorem}

To prove this result, we consider a non-empty \ac{IR}.
The output time of the head-of-line packet is given by the input/output characterizations of \cite{leboudecTheoryTrafficRegulators2018,boyerEquivalenceTheoreticalModel2022}.
It can be upper-bounded.
From this we obtain that the dashed-blue curve $\beta_0$ in Figure~\ref{fig:ir-service-curves:strict-sc:result-figure} is a strict service curve of the \ac{IR}, and so is its supper-additive closure \cite[Prop 5.6]{bouillardDeterministicNetworkCalculus2018}, defined as the function
\begin{equation}\label{eq:ir-service-curves:strict-sc:supper-additive-closure}
    \beta_0 \vee \left(\beta_0 \overline{\otimes} \beta_0\right) \vee \left(\left(\beta_0 \overline{\otimes} \beta_0\right) \overline{\otimes} \beta_0\right) \vee \dots 
\end{equation}
where $\vee$ is the maximum and $\overline{\otimes}$ is the max-plus convolution\footnote{Defined by $(\mathfrak{f} \overline{\otimes} \mathfrak{g}):t\mapsto\sup_{s\le t}\mathfrak{f}(s)+\mathfrak{g}(t-s)$}. 
The computation of \eqref{eq:ir-service-curves:strict-sc:supper-additive-closure} gives $\beta_{\text{sc}}$, shown in solid red in Figure~\ref{fig:ir-service-curves:strict-sc:result-figure}.
Any wide-sense increasing curve that remains below the $\beta_{\text{sc}}$ service curve is also a strict service curve of the \ac{IR} \cite[Prop 5.6]{bouillardDeterministicNetworkCalculus2018}.
This is the case for the rate-latency service curve $\beta_{L^{\min}/I^{\max},I^{\max}}$ shown with a dash-dotted green line in Figure~\ref{fig:ir-service-curves:strict-sc:result-figure}.
The formal proof of Theorem~\ref{thm:ir-service-curves:strict-sc} is in \proofIrStrictSc.

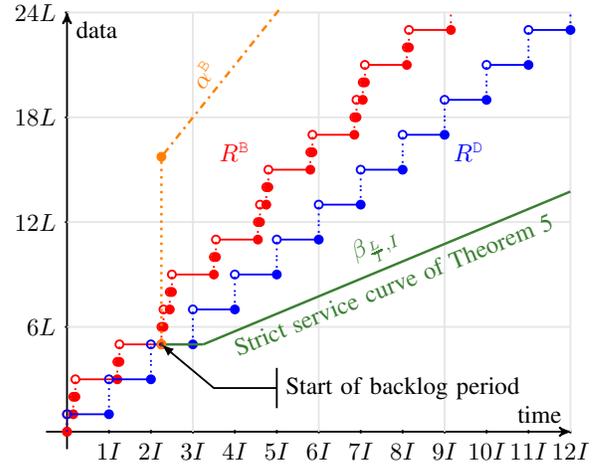
\begin{figure}
    \centering
        \resizebox{0.9\linewidth}{!}{\input{./figures/2023-04-application-strict-sc.tex}}
        \caption{\label{fig:ir-service-curves:strict-sc:application} Application of Theorem~\ref{thm:ir-service-curves:strict-sc} to the Spring trajectory of Section~\ref{sec:fifo-sff:spring}. The dotted red [\resp blue] curve is the cumulative function at the input [\resp output] of the \ac{IR}. The dash-dotted orange leaky-bucket curve is an arrival curve of the aggregate at the input of the \ac{IR}. The solid green rate-latency curve is a strict service curve of the \ac{IR}, as proved by Theorem~\ref{thm:ir-service-curves:strict-sc}.}
    \end{figure}

\emph{Application to the Situation of Section~\ref{sec:fifo-sff:spring}.}
Figure~\ref{fig:ir-service-curves:strict-sc:application} shows the cumulative arrival function $R^{\texttt{B}}$ at the input of the \ac{IR} as well as the cumulative departure function $R^{\texttt{D}}$ at the output of the \ac{IR}, in Spring's trajectory\footnote{With the notations of \eqref{eq:spring-constants}, the parameters used in Figure~\ref{fig:ir-service-curves:strict-sc:application} are: $D=0.86I, d=0.85I$, $\epsilon=0.05I$.} described in Section~\ref{sec:fifo-sff:spring} and Figure~\ref{fig:fifo-sff:spring:n1}.
We also provide the arrival curve $\alpha^{\texttt{B}}$ of the aggregate at the input of the \ac{IR}.

In Section~\ref{sec:fifo-sff:spring}, all packets of the aggregate have the same size $L$ and all three flows have the same leaky-bucket shaping curve $\sigma_{f_1}=\sigma_{f_2}=\sigma_{f_3}=\gamma_{r,b}$ with $b=L$.
The application of Theorem~\ref{thm:ir-service-curves:strict-sc} gives that $\beta_{\frac{L}{I},{I}}=\beta_{r,\frac{b}{r}}$ is a context-agnostic strict service curve of the \ac{IR} for the aggregate.
In Figure~\ref{thm:ir-service-curves:strict-sc}, we place this curve in green at the beginning of the backlog period.
We confirm that when the \ac{IR} is non-empty, the cumulative output $R^{\mathtt{B}}$ is larger than the strict service curve.
Also, it is clear that the horizontal deviation between the arrival curve $\alpha^{\mathtt{B}}$ and the service curve $\beta_{\frac{L}{I},{I}}$ is not bounded.
This means that the network-calculus theory cannot provide a latency bound with this service-curve model, which is consistent with Theorem~\ref{thm:instab}.




In Figure~\ref{fig:ir-service-curves:strict-sc:application}, the rate $\frac{L}{I}=r$ of the green service curve $\beta_{\frac{L}{I},{I}}$  does not follow the long-term rate $\frac{2L}{I}=2r$ of the output cumulative function $R^{\mathtt{D}}$.
However, no better rate for the rate-latency context-agnostic strict service curve can be achieved:

\begin{proposition}[Upper Bound on the Strict Service Curve\label{prop:ir-service-curves:strict-sc:limit}]
    Consider an \ac{IR} that processes an aggregate $\mathcal{F}$ with leaky-bucket shaping curves $\{\sigma_f\}_{f\in\mathcal{F}}=\{\gamma_{r_f,b_f}\}_{f\in\mathcal{F}}$.   
    Consider a curve $\beta^{\text{strict}}$ and assume that $\beta^{\text{strict}}$ is a context-agnostic strict service curve of the \ac{IR}. Then
    \begin{equation}
        \forall t \ge 0, \quad \beta^{\text{strict}}(t) \le \min_{f\in\mathcal{F}}\sigma_f(t)
    \end{equation}
    In particular, 
    if $\beta^{\text{strict}}=\beta_{R,T}$ is a rate-latency curve, then $R \le \min_{f\in\mathcal{F}}r_f$.
\end{proposition}

To prove this result, we consider a set of trajectories $\{x_f\}_{f\in\mathcal{F}}$.
For each flow $f\in\mathcal{F}$, the trajectory $x_f$ is obtained by having only flow $f$ send packets of size $b_f$ to the \ac{IR} at twice the frequency allowed by its shaping curve.
Then the backlog of the \ac{IR} quickly becomes non-empty, but the cumulative output of the \ac{IR}, $R^{x_f,\mathtt{D}}=R_f^{x_f,\mathtt{D}}$ is constrained by the shaping curve $\sigma_f$: $\forall 0\le s \le t, R^{x_f,\mathtt{D}}(t)-R^{x_f,\mathtt{D}}(s) \le \sigma_f(t-s)$.
Combined with the definition of a strict service curve \eqref{eq:background:strict-sc-def}, this gives $\beta^{\text{strict}} \le \sigma_f$.
This is valid for all trajectories $\{x_f\}_{f\in\mathcal{F}}$, hence the result.
The formal proof is in \proofPropLimitIrStrictSc{}.

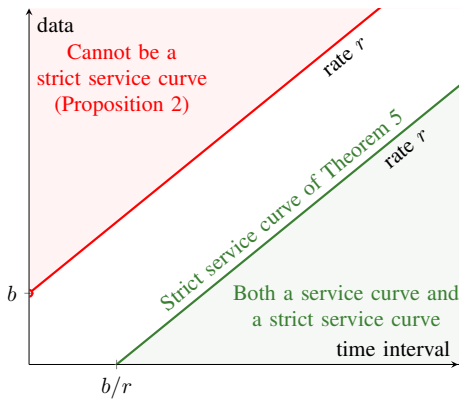
\begin{figure}
\centering
    \resizebox{0.7\linewidth}{!}{\input{./figures/2023-04-envelope-sc.tex}}
    \caption{\label{fig:ir-service-curves:strict-sc:curve-area} Areas of proven, possible, and impossible strict service curves for an \ac{IR} that processes the flows with the same shaping curve $\gamma_{r,b}$ and if all packets have size $b$.}
\end{figure}

Figure~\ref{fig:ir-service-curves:strict-sc:curve-area} shows the areas of proven, possible, and impossible strict service curves of an \ac{IR} that processes the flows with the same leaky-bucket shaping curve $\gamma_{r,b}$, assuming that all the packets have the size $b$.
Any wide-sense increasing function that remains in the green area is a strict service curve of the \ac{IR}, as proven by Theorem~\ref{thm:ir-service-curves:strict-sc}.
In contrast, any function that enters the red area cannot be a strict service curve of the \ac{IR}, as proven by Proposition~\ref{prop:ir-service-curves:strict-sc:limit}.

Incidentally, any wide-sense increasing function that remains within the green area is also a (simple) service curve of the \ac{IR}.
So far, as opposed to the strict-service-curve property, we have not obtained any limit for the simple service curve \eqref{eq:background:sc-def}.
To obtain this limit, we first need to consider the \emph{individual} service curve offered to any flow by the \ac{IR}, which we analyze in the following subsection.

%% file: figures/2023-04-strict-service-curves-ir-proof.tex
\begin{tikzpicture}
    \tikzstyle{ma} = [
        shorten >=-3pt,
        shorten <=-3pt
    ]

    \tikzstyle{beta1} = [blue, dashed]
    \tikzstyle{beta2} = [red]
    \tikzstyle{beta3} = [OliveGreen, dashdotted, line width=1pt]

    \begin{axis}[
        axis x line=center,
        axis y line=center,
        grid=both,
        xmin=0,
        ymin=0,
        xmax=4.9,
        ymax=5,
        ytick={1,2,3,4},
        yticklabels={$L^{\min}$,$2L^{\min}$,$3L^{\min}$,$4L^{\min}$},
        xtick={1,2,3,4},
        xticklabels={$I^{\max}$,$2I^{\max}$,$3I^{\max}$,$4I^{\max}$},
        axis on top,
        xlabel={time interval},
        ylabel={data}
    ]
        
    \path[fill=black!10] (axis cs:1,0) -- (axis cs:1,1) -- (axis cs:2,1) -- (axis cs:2,2) -- (axis cs:3,2) -- (axis cs:3,3) -- (axis cs:4,3) -- (axis cs:4,4) -- (axis cs:5,4) -- (axis cs:5,0) -- cycle;
    
    \draw[line width=1.25pt, beta2, -{Circle[length=6pt,width=6pt]}, ma] (axis cs:0,0) -- (axis cs:1,0);
    \draw[line width=1.25pt, beta2, {Circle[open, length=6pt,width=6pt]}-{Circle[length=6pt,width=6pt]}, ma] (axis cs:1,1) -- (axis cs:2,1);
    \draw[line width=1.25pt, beta2, {Circle[open, length=6pt,width=6pt]}-{Circle[length=6pt,width=6pt]}, ma] (axis cs:2,2) -- (axis cs:3,2);
    \draw[line width=1.25pt, beta2, {Circle[open, length=6pt,width=6pt]}-{Circle[length=6pt,width=6pt]}, ma] (axis cs:3,3) -- (axis cs:4,3) node[pos=0.5,above,fill=white,draw,solid,line width=0.5pt]{$\beta_{\text{sc}}$};
    \draw[line width=1.25pt, beta2, {Circle[open, length=6pt,width=6pt]}-{Circle[length=6pt,width=6pt]}, ma] (axis cs:4,4) -- (axis cs:5,4);
    \draw[line width=1.25pt, beta1, -{Circle[length=6pt,width=6pt]}, ma] (axis cs:0,0) -- (axis cs:1,0);
    \draw[line width=1.25pt, beta1, {Circle[open,length=6pt,width=6pt]}-, ma] (axis cs:1,1) -- (axis cs:5,1)  node[pos=0.7,above,fill=white,draw,solid,line width=0.5pt]{$\beta_0$};
    \draw[line width=1.25pt, beta3] (axis cs:0,0) -- (axis cs:1,0) -- (axis cs:5,4) node[pos=0.6,sloped,below,fill=white,draw,solid,line width=0.5pt]{$\beta_{\frac{L^{\min}}{I^{\max}},I^{\max}}$};

    \end{axis}
\end{tikzpicture}

%% file: figures/2023-04-application-strict-sc.tex
\begin{tikzpicture}[xscale=0.3,yscale=0.25]

    \tikzstyle{ma} = [
        shorten >=-2pt,
        shorten <=-2pt
    ]

    \tikzstyle{min} =  [red]
    \tikzstyle{mout} = [blue]
    \tikzstyle{f1} = [dotted]
    \tikzstyle{f2} = [dotted]
    \tikzstyle{f3} = [dotted]
    \tikzstyle{sc} = [line width=1pt, OliveGreen]
    \tikzstyle{ac} = [line width=1pt, orange]
    \tikzstyle{connect} = [ma, {Circle[open,length=4pt,width=4pt]}-{Circle[length=4pt,width=4pt]}]

    \tikzmath{
        int \x;
        int \y;
        int \k;
        int \XMAX;
        int \YMAX;
        int \KMAX;
        real \sI;
        real \sD;
        real \sTau;
        real \sEpsilon;
        real \sShift;
        real \xIn;
        real \xOut;
        real \lastXIn;
        real \lastXOut;
        \KMAX = 3;
        \YMAX = (\KMAX+1)*6;
        \XMAX = (\KMAX+1)*3;
        \sI = 2;
        \sD = 1.7;
        \sBigD = 1.72;
        \sEpsilon = 0.1;
        \sTau = 3 * \sI + 3 * \sEpsilon - \sD;%
        \sShift = -2 * \sD;
        real \GRAD;
        \GRAD = 0.1;
        {
            \draw[-stealth'] (-1,0) -- (\XMAX*\sI,0) node[pos=1,anchor=south east]{time};
            \draw[-stealth'] (0,-1) -- (0,\YMAX) node[pos=1,anchor=north west]{data};
        };
        for \x in {1,..,\XMAX}{
            {
                \draw (\x*\sI,\GRAD) -- (\x*\sI,-\GRAD) node[pos=1,below]{$\x I$};
            }; 
        };
        for \x in {3,6,..,\XMAX}{
            {
                \draw[color=black!10] (\x*\sI,0) -- (\x*\sI,\YMAX);
            }; 
        };
        for \y in {6,12,..,\YMAX}{
            {
                \draw (\GRAD,\y) -- (-\GRAD,\y) node[pos=1,left]{$\y L$};
                \draw[color=black!10] (0,\y) -- (\XMAX*\sI,\y);
            }; 
        };
        \y = 0;
        \lastXIn = 0;
        \lastXOut = 0;
        \xBacklog = -1;
        \yBacklog = -1;
        for \k in {0,..,\KMAX}{
            \xIn = \sShift + 2*\sD + \k*\sTau;
            {
                \draw[f1,min] (\xIn,\y) -- ++(0,1);
                \draw[min, connect] (\lastXIn,\y) -- (\xIn,\y);
            };
            \xOut = 3*\k*\sI;
            {
                \draw[f1,mout] (\xOut,\y) -- ++(0,1);
            };
            \y = \y + 1;
            \lastXIn = \xIn;
            \lastXOut = \xOut;
            \xIn = \sShift + \sI + \sD + \k*\sTau;
            {
                \draw[f1,min] (\xIn,\y) -- ++(0,1);
                \draw[min, connect] (\lastXIn,\y) -- (\xIn,\y);
            };
            \xOut = 3*\k*\sI + \sI; 
            {
                \draw[f1,mout] (\xOut,\y) -- ++(0,1);
                \draw[mout, connect] (\lastXOut,\y) -- (\xOut,\y);
            };
            if (\xIn < \lastXOut) then {
                if (\xBacklog < 0) then {
                    \xBacklog = \lastXIn-\sEpsilon;
                    \yBacklog = \y-2;
                };
            };
            \y = \y + 1;
            \lastXIn = \xIn;
            \lastXOut = \xOut;
            \xIn = \sShift + \sI + \sEpsilon + \sD + \k*\sTau;
            {
                \draw[f2,min] (\xIn,\y) -- ++(0,1);
                \draw[min, connect] (\lastXIn,\y) -- (\xIn,\y);
            };
            \xOut = 3*\k*\sI + \sI; 
            {
                \draw[f2,mout] (\xOut,\y) -- ++(0,1);
            };
            \y = \y + 1;
            \lastXIn = \xIn;
            \lastXOut = \xOut;
            \xIn = \sShift + 2*\sI + \sEpsilon + \sD + \k*\sTau;
            {
                \draw[f2,min] (\xIn,\y) -- ++(0,1);
                \draw[min, connect] (\lastXIn,\y) -- (\xIn,\y);
            };
            \xOut = 3*\k*\sI + 2*\sI; 
            {
                \draw[f2,mout] (\xOut,\y) -- ++(0,1);
                \draw[mout, connect] (\lastXOut,\y) -- (\xOut,\y);
            };
            if (\xIn < \lastXOut) then {
                if (\xBacklog < 0) then {
                    \xBacklog = \lastXIn-\sEpsilon;
                    \yBacklog = \y-2;
                };
            };
            \y = \y + 1;
            \lastXIn = \xIn;
            \lastXOut = \xOut;
            \xIn = \sShift + 2*\sI + 2*\sEpsilon + \sD + \k*\sTau;
            {
                \draw[f3,min] (\xIn,\y) -- ++(0,1);
                \draw[min, connect] (\lastXIn,\y) -- (\xIn,\y);
            };
            \xOut = 3*\k*\sI + 2*\sI; 
            {
                \draw[f3,mout] (\xOut,\y) -- ++(0,1);
            };
            \y = \y + 1;
            \lastXIn = \xIn;
            \lastXOut = \xOut;
            \xIn = \sShift + 3*\sI + 2*\sEpsilon + \sD + \k*\sTau;
            {
                \draw[f3,min] (\xIn,\y) -- ++(0,1);
                \draw[min, connect] (\lastXIn,\y) -- (\xIn,\y);
            };
            \xOut = 3*\k*\sI + 3*\sI; 
            {
                \draw[f3,mout] (\xOut,\y) -- ++(0,1);
                \draw[mout, connect] (\lastXOut,\y) -- (\xOut,\y);
            };
            if (\xIn < \lastXOut) then {
                if (\xBacklog < 0) then {
                    \xBacklog = \lastXIn-\sEpsilon;
                    \yBacklog = \y-2;
                };
            };
            \y = \y + 1;
            \lastXIn = \xIn;
            \lastXOut = \xOut;
        };
        \scR = 1/\sI;
        \scT = \sI;
        \acR = 3*(1/\sI);
        \acB = 3 + 3*\acR*\sBigD;
        \scY = (\XMAX*\sI - (\xBacklog + \scT)) * \scR + \yBacklog;
        \acX = ((\YMAX - \yBacklog - \acB) / \acR + \xBacklog);
        {
            \draw[sc] (\xBacklog,\yBacklog) -- (\xBacklog + \scT, \yBacklog);
            \draw[sc] (\xBacklog + \scT, \yBacklog) -- (\XMAX*\sI,\scY) node[pos=0.5,above, sloped]{$\beta_{\frac{L}{I},I}$} node[pos=0.5,below,rotate=22]{Strict service curve of Theorem~\ref{thm:ir-service-curves:strict-sc}};    
            \draw[ac, dotted, ma,{Circle[open,length=4pt,width=4pt]}-] (\xBacklog,\yBacklog) -- (\xBacklog, \yBacklog + \acB);
            \draw[ac, dashdotted, ma, {Circle[length=4pt,width=4pt]}-] (\xBacklog, \yBacklog + \acB) -- (\acX,\YMAX) node[pos=0.5, sloped, above] {$\alpha^{\texttt{B}}$};
            \node[anchor=west] at (5*\sI,2.5) (sob) {Start of backlog period};
            \draw (sob.north west) -- (sob.south west);
            \draw[-latex] (sob.west) -- ++(-3,0) -- (\xBacklog,\yBacklog);
            \node[anchor=west, blue] at (9*\sI, 16) {$R^{\texttt{D}}$};
            \node[red] at (4*\sI, 16) {$R^{\texttt{B}}$};
        };
    }   

\end{tikzpicture}

%% file: figures/2023-04-envelope-sc.tex
\begin{tikzpicture}
    \tikzstyle{ma} = [
        shorten >=-2pt,
        shorten <=-2pt
    ]

    \tikzstyle{alpha} =[line width=1pt, ma, {Circle[open, length=4pt,width=4pt]}-, red, mark=none]
    \tikzstyle{beta} = [line width=1pt, ma, OliveGreen, mark=none]

    \begin{axis}[
        axis x line=center,
        axis y line=center,
        xmin=0,
        ymin=0,
        xmax=4.9,
        ymax=5,
        xlabel={time interval},
        ylabel={data},
        ytick={1},
        yticklabels={$b$},
        xtick={1},
        xticklabels={$b/r$},
    ]
        
    \addplot[domain=1:5, beta, name path=b] {x-1} node[pos=0.5,above,sloped] {Strict service curve of Theorem~\ref{thm:ir-service-curves:strict-sc}} node[pos=0.8,below,black,sloped]{rate $r$};
    \addplot[domain=0:5, alpha, name path=a] {x+1} node[pos=0.7,below,sloped,black]{rate $r$};
    \addplot[domain=1:5, name path=axis] {0};
    \addplot[domain=0:5, draw=none, name path=mtop] {5};

    \addplot [
        fill=OliveGreen, 
        fill opacity=0.05
    ]
    fill between[
        of=axis and b,
        soft clip={domain=1:5},
    ];

    \addplot [
        fill=red, 
        fill opacity=0.05
    ]
    fill between[
        of=a and mtop,
        soft clip={domain=0:5},
    ];

    \node[anchor=west, red] at (0,4) {\makecell{Cannot be a \\ strict service curve\\ (Proposition~\ref{prop:ir-service-curves:strict-sc:limit})}};
    \node[anchor=east, OliveGreen] at (5,0.8) {\makecell{Both a service curve and \\ a strict service curve}};

    \end{axis}
\end{tikzpicture}

%% file: 6-service-curves/2-individual-service.tex
\subsection{Upper-Bound on the Individual Service Curve}
\label{sec:ir-service-curves:indiv-sc}

Consider an \ac{IR} that processes an aggregate $\mathcal{F}$ and take a flow $g\in\mathcal{F}$.
In this section, we are interested in the service the \ac{IR} guarantees to $g$, \ie in an \emph{individual} service curve of the \ac{IR} for $g$.

If each flow $f\in\mathcal{F}\backslash\{g\}$ enters the \ac{IR} with an arrival curve $\alpha_{f}^{\mathtt{B}}$ that is equal or smaller than its shaping curve $\sigma_f$, then none of the packets of the flows $\mathcal{F}\backslash\{g\}$ is ever delayed by the \ac{IR}.
In this case, the \ac{IR} acts as a \ac{PFR} for $g$ and provides an individual fluid service curve $\sigma_g$ to $g$.

In a more likely setting, though, the \ac{IR} reshapes flows that exhibit an input arrival curve (strictly) larger than their configured shaping curve. 
In such a case, and if the \ac{IR} processes more than four flows, no useful individual service curve for $g$ can be obtained:

\begin{theorem}[Upper-Bound on the Individual Service curve\label{thm:ir-service-curves:limit-ir-residual}]
    Consider an \ac{IR} that processes an aggregate $\mathcal{F}$ of at least four flows and with the same leaky-bucket shaping curve for at least three of them: $\forall f_i \in \{f_1,f_2,f_3\}, \sigma_{f_i}=\gamma_{r,b}$.
    Consider a flow $g\in \mathcal{F}\backslash\{f_1,f_2,f_3\}$ and assume that for each $f_i \in \{f_1,f_2,f_3\}, \gamma_{r,b_i}$ is an arrival curve for $f_i$ at the input $\mathtt{B}$ of the \ac{IR} (Figure~\ref{fig:sysmodel:traffic-regulators:simple-ir}), with $b_1>b$, $b_2\ge b$, $b_3\ge b$ (permutating the indexes if required).
    Last, consider a curve $\beta_g\in\mathfrak{F}_0$ that can depend on $\{\sigma_f\}_{f\in\mathcal{F}}$ and on $\{\alpha_h^{\mathtt{B}}\}_{h\in\mathcal{F}\backslash\{g\}}$. 

    If $\beta_g$ is an individual service curve of the \ac{IR} for the flow $g$, then $\beta_g$ is upper-bounded by $g$'s minimum packet size $L_g^{\min}$.
\end{theorem}

Theorem~\ref{thm:ir-service-curves:limit-ir-residual} shows that any individual-service-curve model of the \ac{IR} for a flow $g$ can only guarantee that one single packet of $g$ will ever cross the \ac{IR} over the entire network's lifetime.
Hence, no useful context-agnostic service curve exists to model the service offered to a single flow by the \ac{IR} (each flow is likely to send many packets).

 
To prove Theorem~\ref{thm:ir-service-curves:limit-ir-residual}, we reuse the Spring trajectory of Theorem~\ref{thm:instab} for the three flows $f_1,f_2,f_3$.
This trajectory creates a constant build-up of delay and backlog inside the \ac{IR}.
We then consider a fourth flow, $g$, and its first packet of size $L$.
On one hand, if $\beta_g$ is an individual service curve of the \ac{IR} for the flow $g$, then the delay of the first packet of $g$ through the \ac{IR} can be upper bounded by $\inf\{t\in\mathbb{R}^+|\beta_g(t) \ge L\}$.
On the other hand, the first packet of $g$ can suffer a delay as large as desired within the Spring trajectory:
We send it when the accumulated delay in the \ac{IR} is large enough.
The combination of the two above observations provides the result.
The formal proof is in \proofThmLimitIndiv{}.

%% file: 6-service-curves/3-sc.tex
\subsection{Limits on the Aggregate Service Curve}
\label{sec:ir-sc:limit-sc}

Let us go back to the analysis of the service offered to the aggregate, by focusing on simple (non strict) service-curve models.
One consequence of Theorem~\ref{thm:ir-service-curves:limit-ir-residual} is that the long-term rate of the service curve for the aggregate is upper-bounded by three times the rate of a single contract.

\begin{theorem}[Maximum long-term rate of any service curve\label{thm:ir-service-curves:limit-sc}]
    Consider an \ac{IR} that processes an aggregate $\mathcal{F}$ of at least four flows and with the same leaky-bucket shaping curve for at least three of them: $\forall f_i\in\{f_1,f_2,f_3\}, \sigma_{f_i} = \gamma_{r,b}$.
    Consider a curve $\beta\in\mathfrak{F}_0$ that can depend on $\{\sigma_f\}_{f\in\mathcal{F}}$.
    If the \ac{IR} offers $\beta$ as a context-agnostic service curve, then the long-term rate of $\beta$ \cite[Def. 12.1]{bouillardDeterministicNetworkCalculus2018} is upper-bounded by three times the rate of a single contract, \ie
    \begin{equation}
        \liminf_{t\rightarrow +\infty} \frac{\beta(t)}{t} \le 3r
    \end{equation}
\end{theorem}



To prove Theorem~\ref{thm:ir-service-curves:limit-sc}, we pick one flow $g\in\mathcal{F}\backslash\{f_1,f_2,f_3\}$ and a traffic arrival for $h\in\mathcal{F}\backslash\{g\}$ that meets the requirements of Theorem~\ref{thm:ir-service-curves:limit-ir-residual}.
The \ac{IR} is a \ac{FIFO} system.
Hence, if $\beta$ is a service curve of the \ac{IR}, then for any $\theta\in\mathbb{R}^+$, the curve $\beta_g^{\theta}$ defined by
\begin{equation}\label{eq:background:indiv-sc:fifo-residual}
    \beta_g^{\theta}: t\mapsto \left|\beta(t)-\sum_{f\in\mathcal{F}\backslash\{g\}
    }\alpha_f(t-\theta)\right|^+ \cdot\mathbbm{1}_{\{t>\theta\}}
\end{equation}
verifies Equation~\eqref{eq:background:indiv-sc:def} \cite[Prop. 6.4.1]{leboudecNetworkCalculus2001}.
Note that $\beta_g^{\theta}$ may not be wide-sense increasing, thus not an individual arrival curve for $g$.
Inspired by \cite[\S 5.2.1]{bouillardDeterministicNetworkCalculus2018}, we resolve this by considering the curve $\beta_g^{\theta}\overline{\oslash}\boldsymbol{0}:t\mapsto \inf_{s\ge t}\beta_g^{\theta}(s)$, where $\overline{\oslash}$ is the max-plus deconvolution\footnote{Defined as $\mathfrak{f}\overline{\oslash}\mathfrak{g}:t\mapsto \inf_{s\ge t}\mathfrak{f}(s)- \mathfrak{g}(s-t)$} and $\boldsymbol{0}$ is the zero function $t\mapsto 0$.
The function $\beta_g^{\theta}\overline{\oslash}\boldsymbol{0}$ is wide-sense increasing, smaller than $\beta_g^{\theta}$, thus an individual service curve for $g$.
From Theorem~\ref{thm:ir-service-curves:limit-ir-residual}, we obtain $\forall\theta\ge0, \forall t\ge 0$, $(\beta_g^{\theta}\overline{\oslash}\boldsymbol{0})(t)\le L^{\min}_g$.
The left-hand side of this inequation is an infimum, but the result is valid for any $\theta\ge0,t\ge 0$.
Hence, we can derive a bound on the long-term rate of $\beta$.
The formal proof is in \proofThmLimitAggreg{}.

With Theorem~\ref{thm:ir-service-curves:limit-sc}, we can conclude that for an \ac{IR} that processes more than four flows, no useful context-agnostic service curve exists to model the \ac{IR}.
Indeed, the aggregate of the four flows can exhibit a sustained rate of four times the rate of a single-flow contract, whereas a context-agnostic service curve can only guarantee a long-term service rate of three times this value.
Hence, Theorem~\ref{thm:ir-service-curves:limit-sc} concludes our search of context-agnostic service-curve models for the \ac{IR} and the results of Section~\ref{sec:ir-service-curves} are summarized in Table~\ref{tab:conclusion}.


\begin{table}
    \caption{\label{tab:conclusion}Context-Agnostic Service Curves for an \ac{IR} that Processes at least Four Flows with the Same Leaky-Bucket Shaping Curve $\gamma_{r,b}$, assuming All Packets Have Size $b$.}
    \resizebox*{\linewidth}{!}{\input{./figures/2023-05-tab-summary-sc.tex}}
\end{table}

%% file: figures/2023-05-tab-summary-sc.tex
\begin{tabular}{c|c|c}
    \makecell{Service-curve\\type} & \makecell{Curve exhibited\\in this paper} & \makecell{Limit exhibited\\in this paper}  \\
    \hline
    \makecell{Service curve\\$\beta$} & $\beta_{r,\frac{b}{r}}$ & $\liminf_{t\rightarrow+\infty} \frac{\beta(t)}{t} \le 3r$ \\
    \hline
    \makecell{Strict service\\curve $\beta^{\text{strict}}$} & $\beta_{r,\frac{b}{r}}$ & $\forall t\ge0,\ \beta^{\text{strict}}(t)\le\gamma_{r,b}(t)$ \\
    \hline
    \makecell{Individual service\\curve $\beta_g,\ \forall g\in\mathcal{F}$} & \makecell{None\\($t\mapsto0$)} & $\forall t\ge0,\ \beta_g(t) \le L^{\min}_g$\\
\end{tabular}

%% file: 8-conclusion/0-conclusion.tex
\section{Conclusion}
\label{sec:conclusion}


Network calculus is a framework for obtaining worst-case performance bounds of time-sensitive networks, as required for their validation.
Most of the mechanisms standardized by the \acf{TSN} task group of the \ac{IEEE} enjoy a network-calculus service-curve model published in the literature.
%
%
The \acf{IR}, standardized as \emph{\acl{ATS}} (\acs{ATS}) in \ac{TSN}, is an exception.
Its shaping-for-free property is instrumental in designing and analyzing time-sensitive networks but was proved without network-calculus service curves.
The existence of a service-curve model that explains the \ac{IR}'s behavior and its shaping-for-free property remained an open question.
If such a model existed, network engineers could use the \ac{IR} outside of the shaping-for-free requirements and still compute end-to-end performance bounds with service-curve-oriented tools.

In this paper, we settled the question:
Network-calculus service curves cannot explain the behavior of the \ac{IR}.
We show that the \ac{IR} still offers non-trivial functions as (strict) service curves, but (a) none of them can explain the shaping-for-free property of the \ac{IR} and (b) these curves are too weak to be helpful and cannot offer any delay guarantee in most cases.
Consequently, performance bounds cannot be obtained with service-curve-oriented approaches when the \ac{IR} is used in a context that differs from the shaping-for-free requirements, \eg after a non-\ac{FIFO} system.
We prove that these bounds do not even exist: 
The \ac{IR} can yield unbounded latencies after a non-\ac{FIFO} system.

%% file: 9-appendix/0-appendix.tex

\input{9-appendix/1-notations-proofs/0-notations-proofs.tex}

\section{Formal Proofs}
\label{appendix:proofs}
\input{9-appendix/2-proofs-section-5/0-proofs-section-5.tex}

\input{9-appendix/3-proofs-section-6/0-proofs-section-6.tex}

%% file: 9-appendix/1-notations-proofs/0-notations-proofs.tex
\section{Formal Notations for the Proofs}
\label{sec:appendix:notations-proofs}

This appendix completes the system model of Section~\ref{sec:sysmodel} and formally defines the notations used in the formal proofs in Appendix~\ref{appendix:proofs}.
Table~\ref{tab:extended-notations} lists the extended formal notations and completes Table~\ref{tab:motations}.

\begin{table}
    \caption{\label{tab:extended-notations} Extended Formal Notations for the Formal Proofs}
    \resizebox{\linewidth}{!}{\input{99-config/notations-extended.tex}}
\end{table}

\input{9-appendix/1-notations-proofs/1-marked-point.tex}
\input{9-appendix/1-notations-proofs/2-cumulative-notation.tex}
\input{9-appendix/1-notations-proofs/3-traffic-regulators.tex}

%% file: 99-config/notations-extended.tex
\begin{tabular}{r|ll}
    \multicolumn{3}{c}{Common Operators} \\
    $\mathfrak{f}^{\downarrow}$ & $w\mapsto \sup\{s\ge 0 | \mathfrak{f}(s) \ge w\}$ & Pseudo-inverse \\
    \hline
    \multicolumn{3}{c}{Common Curves} \\
    $\delta_D$      & $t\mapsto \left\lbrace \begin{aligned} 0 &\quad \text{ if } t \le D\\ +\infty &\quad \text{ otherwise}\end{aligned}\right.$ & Bounded-delay curve.\\ 
    \hline
    \multicolumn{3}{c}{Notation with Marked-Point Processes} \\
    $\mathtt{M}$ & \multicolumn{2}{l}{An observation point} \\
    $M^x$ & \multicolumn{2}{l}{Sequence of packet arrival dates\dots} \\
    $L^{x,\mathtt{M}}$ & \multicolumn{2}{l}{Sequence of packet sizes\dots} \\
    $F^{x,\mathtt{M}}$ & \multicolumn{2}{l}{Sequence of flow indexes\dots} \\
    $\mathscr{M}^x$ & $\triangleq(M^x,L^{x,\mathtt{M}},F^{x,\mathtt{M}})$ & Packet sequence\dots\\
    &  \multicolumn{2}{l}{\hfill \dots at the packetized observation point $\mathtt{M}$ in trajectory $x$} \\
    \hline
    \multicolumn{3}{c}{Notation with Cumulative Functions} \\
    $R^{x,\mathtt{M}}$ & \multicolumn{2}{l}{Cumulative function of the aggregate\dots} \\
    {[\resp $R^{x,\mathtt{M}}_f$]} & \multicolumn{2}{l}{\dots[\resp of $f$] at $\mathtt{M}$ in Trajectory $x$.} \\
    \hline
    \multicolumn{3}{c}{Mapping of the Notations} \\
    $\mathscr{R}$ & $\mathscr{M}^x\mapsto R^{x,\mathtt{M}}$ & Maps the packet sequence  $\mathcal{M}^x$\dots\\
    & \multicolumn{2}{l}{\dots at observation point $\mathtt{M}$ to the cumulative function $R^{x,\mathtt{M}}$}\\
    & \multicolumn{2}{l}{at the same observation point $\mathtt{M}$ in the same trajectory $x$.}\\
    \hline
    \multicolumn{3}{c}{Traffic-Regulator Models} \\
    $\Pi^{\gamma_{r,b}}$ & \multicolumn{2}{l}{Pi-operator \eqref{eq:sysmod:pi-gamma-regulator} associated with shaping curve $\gamma_{r,b}$}\\
    $\Lambda_{f,n}^x$ & \multicolumn{2}{l}{Number of tokens in the bucket of flow $f$ just after packet $n$.}\\
    $n\ominus1$ & \multicolumn{2}{l}{Same-flow previous packet's index, 0 if $n$ is the first of its flow.}

\end{tabular}

%% file: 9-appendix/1-notations-proofs/1-marked-point.tex
\subsection{Notation with Marked-Point Processes}
\label{sec:sysmodel:marked-point-notation}

For a trajectory $x$ and for a packetized observation point $\mathtt{M}$ in the network $\mathcal{N}$, we define the packet sequence $\mathscr{M}^x=(M^x,L^{x,\mathtt{M}},F^{x,\mathtt{M}})$ as in \cite[\S II.A]{leboudecTheoryTrafficRegulators2018}:
$M^x=(M^x_1,M^x_2,\dots)$ is the sequence of packet arrival time instants at $\mathtt{M}$, in chronological order. 
$L^{x,\mathtt{M}}=(L^{x,\mathtt{M}}_1, L^{x,\mathtt{M}}_2, \dots)$ is the sequence of packet lengths, in chronological order.
$F^{x,\mathtt{M}}=(F^{x,\mathtt{M}}_1, F^{x,\mathtt{M}}_2, \dots)$ is the sequence of flow indexes: 
$F^{x,\mathtt{M}}_n=f$ if and only if the $n$-th packet that crosses $\mathtt{M}$ in trajectory $x$ belongs to $f$.
For another observation point $\mathtt{A}$, we note $\mathscr{A}^x=(A^x,L^{x,\mathtt{A}},F^{x,\mathtt{A}})$ the packet sequence at $\mathtt{A}$.

For two infinite sequence of real numbers $A=(A_n)_{n\in\mathbb{N}^*}$, $B=(B_n)_{n\in\mathbb{N}^*}$, we note $B\ge A$ if $\forall n\in\mathbb{N}^*$, $B_n\ge A_n$.


%% file: 9-appendix/1-notations-proofs/2-cumulative-notation.tex
\subsection{Notation with Cumulative Arrival Functions}
\label{sec:sysmodel:cumulative}

For a trajectory $x$ and for an observation point $\mathtt{M}$ (that can be fluid or packetized), we denote by $R^{x,\mathtt{M}}$ [\resp $R_f^{x,\mathtt{M}}$] the left-continuous cumulative arrival function of the aggregate $\mathcal{F}$ [\resp of the individual flow $f$] at 
$\mathtt{M}$ in the trajectory $x$. 


For a packetized observation point $\mathtt{M}$ and its packet sequence $\mathscr{M}^x$ defined in Section \ref{sec:sysmodel:marked-point-notation}, we denote by $\mathscr{R}(\mathscr{M}^x)$ 
the cumulative arrival function:
\begin{align}
    \mathscr{R}(\mathscr{M}^x):t&\mapsto \sum_{n\in\mathbb{N}^*}L^{x,\mathtt{M}}_n\mathbbm{1}_{\{M^x_n < t\}}
\end{align}
where $\mathbbm{1}_{\{\texttt{cond}\}}$ equals 1 when \texttt{cond} is true, 0 otherwise.
From the definition of the packet sequence, we have $\forall x, \forall \mathtt{M}$, $\mathscr{R}(\mathscr{M}^x)=R^{x,\mathtt{M}}$ is the cumulative arrival function of the aggregate
at $\mathtt{M}$ in trajectory $x$.

%% file: 9-appendix/1-notations-proofs/3-traffic-regulators.tex
\subsection{Formal Model For Traffic Regulators}
\label{sec:sysmodel:traffic-regulators}

We briefly recall the input-output characterization of traffic regulators from Le Boudec \cite{leboudecTheoryTrafficRegulators2018}, that we restrain to leaky-bucket shaping curves.

The \ac{PFR} for the single flow $f$, configured with a leaky-bucket shaping curve $\sigma_f=\gamma_{r_f,b_f}$ is a causal, lossless and \ac{FIFO} system (Figure~\ref{fig:sysmodel:traffic-regulators:simple-ir}) that transforms, for any acceptable trajectory $x$, the input packet sequence $\mathscr{B}^x=(B^x,L^{x,\mathtt{B}},F^{x,\mathtt{B}})$ with $F^{x,\mathtt{B}}=\{f,f\dots\}$ into the output packet sequence $\mathscr{D}^x=(D^x,L^{x,\mathtt{D}},F^{x,\mathtt{D}})$ with $L^{x,\mathtt{D}}=L^{x,\mathtt{B}}$, $F^{x,\mathtt{D}}=F^{x,\mathtt{B}}$ and $\forall n\in\mathbb{N}^*$
\begin{equation}\label{eq:sysmod:pfr-model}
    D_n^x = \max\left(B_n^x,D_{n-1}^x,\Pi^{\gamma_{r_f,b_f}}(D^{x},L^{x,\mathtt{D}})_n\right)
\end{equation}
By convention, $\forall x, D_0^x=0$.
The $\Pi$ operator $\Pi^{\gamma_{r_f,b_f}}$ associated with the shaping curve $\gamma_{r_f,b_f}$ is defined in \cite[\S III.A]{leboudecTheoryTrafficRegulators2018} for two sequences $D,L$ by, $\forall n\in\mathbb{N}^*$
\begin{equation}\label{eq:sysmod:pi-gamma-regulator}
    \left(\Pi^{\gamma_{r_f,b_f}}(D,L)\right)_n = \max_{1\le m \le n-1} \left\lbrace D_m + \gamma_{r_f,b_f}^{\downarrow}\left(\sum_{j=m}^n L_j\right)\right\rbrace
\end{equation}
$\gamma_{r_f,b_f}^{\downarrow}:w\mapsto |w-b|^+\cdot\frac{1}{r}$ is the pseudo-inverse\footnote{Defined by $\mathfrak{f}^{\downarrow}:w\mapsto \sup\{s\ge 0 | \mathfrak{f}(s) \ge w\}$} of $\gamma_{r_f,b_f}$.

The \ac{IR} for the aggregate $\mathcal{F}$, configured with the leaky-bucket shaping curves $\{\gamma_{r_f,b_f}\}_{f\in\mathcal{F}}$ is a causal, lossless and \ac{FIFO} system (Figure~\ref{fig:sysmodel:traffic-regulators:simple-ir}) that transforms the input packet sequence $\mathscr{B}^x=(B^x,L^{x,\mathtt{B}},F^{x,\mathtt{B}})$ into the output sequence $\mathscr{D}^x=(D^x,L^{x,\mathtt{D}},F^{x,\mathtt{D}})$ with $L^{x,\mathtt{D}}=L^{x,\mathtt{B}}$, $F^{x,\mathtt{D}}=F^{x,\mathtt{B}}$ and $\forall n\in\mathbb{N}^*$
\begin{equation}\label{eq:sysmod:ir-model}
    D_n^x = \max\left(B_n^x,D_{n-1}^x,\Pi^{\gamma_{r_f,b_f}}([D^{x}]^f,[L^{x,\mathtt{D}}]^f)_{\text{index}(n,f)}\right)
\end{equation}
where $f=F^{x,\mathtt{D}}_n$ is the flow $f\in\mathcal{F}$ that owns the $n$-th packet that crosses the \ac{IR} and $\text{index}(n,f)$ is the index of the $n$-th packet in the sequence that contains only the packets of $f$.  By convention $\forall x, D_0^x=0$.
Equation~\eqref{eq:sysmod:ir-model} is similar to \eqref{eq:sysmod:pfr-model}, except that the flow $f\in\mathcal{F}$ that defines the applied shaping curve $\gamma_{r_f,b_f}$ changes at every new $n$ and the $\Pi$ operator is applied only to the subsequences $[D^{x}]^f$, $[L^{x,\mathtt{D}}]^f$, obtained from $D^{x}$, $L^{x,\mathtt{D}}$ by keeping only the packets that belong to $f$, \ie to the same flow as the current packet $n$.

Boyer proposes a second model \cite{boyerEquivalenceTheoreticalModel2022} for an \ac{IR} with leaky-bucket shaping curves.
It relies on the content of token buckets, one per flow.
For a flow $f\in\mathcal{F}$, $\Lambda^x_{f,n}$ is defined as the number of tokens inside the bucket for flow $f$ just after packet $n$ is released from the \ac{IR} in Trajectory $x$.
The head-of-line packet is released as soon as the bucket for the associated flow contains at least as many tokens as the packet's size.
With this equivalent model, the release time of packet $n$ (Equation~\eqref{eq:sysmod:ir-model}) is \cite[\S 3.3]{boyerEquivalenceTheoreticalModel2022}
\begin{equation}
    D_n^x = \max\left( B_n^x, D_{n-1}^x, \frac{L_n^{x,\mathtt{B}}-\Lambda^x_{f,{n\ominus 1}}}{r_f}+D_{n\ominus 1} \right)
\end{equation}
where $f=F^{x,\mathtt{B}}_n$ is the flow that owns packet $n$ and $n\ominus 1$ is the index of the last packet of $f$ just before $n$.
By convention, for any trajectory $x$ and any flow $f$, $n\ominus 1=0$ if $n$ is the first packet of $f$, and $\Lambda^x_{f,{0}}$ equals $b_f$, the shaping-curve burst configured for flow $f$.
Last, $D_0^x=0$.

Boyer's \cite{boyerEquivalenceTheoreticalModel2022} and Le Boudec's \cite{leboudecTheoryTrafficRegulators2018} models are equivalent for leaky-bucket shaping curves. 
We use either model in the formal proofs in Appendix~\ref{appendix:proofs}.

%% file: 9-appendix/2-proofs-section-5/0-proofs-section-5.tex
\input{9-appendix/2-proofs-section-5/1-proof-spring.tex}
\input{9-appendix/2-proofs-section-5/2-proof-explains.tex}
\input{9-appendix/2-proofs-section-5/3-proof-lemma.tex}
\input{9-appendix/2-proofs-section-5/4-proof-ir.tex}

%% file: 9-appendix/2-proofs-section-5/1-proof-spring.tex
\subsection{Proof of Theorem~\ref{thm:instab}}

\begin{proof}[Proof of Theorem~\ref{thm:instab}]
\label{proof:spring}
\label{proof:thm:instab}
    Consider an \ac{IR} that processes at least three flows with the same leaky-bucket shaping curve for three of them: $\forall f_i\in\{f_1,f_2,f_3\},\ \sigma_{f_i}=\gamma_{r,b}$ with $r>0$ and $b$ greater than the minimum packet size of $f_1,f_2,f_3$.

    Let the adversary Spring define the constants $I,d,\epsilon$ and $\tau$ as in Equation~\eqref{eq:spring-constants}.
    We consider the network $\mathcal{N}_1$ of Figure~\ref{fig:fifo-sff:spring:n1}, obtained by concatenating the Spring-controlled source $\phi$, the Spring-controlled system $S_1$ and the \ac{IR} (not controlled by Spring).
    We denote by $x=1$ the trajectory on the network $\mathcal{N}_1$ that results from the choices of Spring, and we formally describe it using the packet sequences of the form $\mathscr{M}^1=(M^1,L^{1,\mathtt{M}},F^{1,\mathtt{M}})$ for an observation point $\mathtt{M}$.
    
    \begin{table*}
        \caption{\label{tab:spring:sequences} Packet Sequences Used by Spring}
        \resizebox{\linewidth}{!}{\input{./figures/2022-11-spring-sequences}}
    \end{table*}

    The Spring-controlled source $\phi$ generates at observation point $\mathtt{A}$ in Figure~\ref{fig:fifo-sff:spring:n1} the packet sequence $\mathscr{A}^1=(A^1,L^{1,\mathtt{A}},F^{1,\mathtt{A}})$ defined in Table~\ref{tab:spring:sequences}.
    It comprises a sub-sequence of six packets that repeats every $\tau$ seconds.
    All packets have the same size\footnote{If $b$ is greater than the maximum packet size, then we send two or more packets at the same time and such that their lengths sum up to $b$.} $b$, thus $L^{1,\mathtt{A}}=(b,b,b,\dots)$.

    Using the definitions of $I,d,\epsilon$ and $\tau$, one can verify that the sequence $A^1$ in Table~\ref{tab:spring:sequences} is increasing. 
    For example, $$A^1_{6k+2}-A^1_{6k+1}=I+\epsilon+k\tau-d-k\tau=I+\epsilon-d \ge \epsilon > 0$$ because $d < I$ \eqref{eq:spring-constants}.
    We can also compute the minimum distance between any two packets of the same flow: 
    $$A^1_{6k+3} - A^1_{6k+1} = I = \frac{b}{r}$$ and 
    $$A^1_{6(k+1)+1} - A^1_{6k+3} = \tau - I = 2I+3\epsilon - d \ge I + 3\epsilon \ge \frac{b}{r}$$ because $d < I$. 
    We obtain similar results for $f_2$ and $f_3$.
    This proves that the minimum distance between any two packets of the same flow is $\frac{b}{r}$.
    Hence, each flow $f_i$ is $\gamma_{r,b}$-constrained at its source $\phi$ and Property~\ref{enum:spring:source-constraint} of Theorem~\ref{thm:instab} holds.
    
    The Spring-controlled system $S_1$ outputs the packet sequence $\mathscr{B}^1=(B^1,L^{1,\mathtt{B}},F^{1,\mathtt{B}})$ with $L^{1,\mathtt{B}}=(b,b,b,\dots)$ and $B^1$, $F^{1,\mathtt{B}}$ defined in Table~\ref{tab:spring:sequences}.
    For any $k\ge 0$, the $(6k+2)$th packet in sequence $\mathscr{A}^1$ is now the $(6k+3)$th packet in $\mathscr{B}^1$ and \emph{vice versa}.
    This is reflected by the sequence $F^{1,\mathtt{B}}$ that differs from $F^{1,\mathtt{A}}$.
    
    By comparing $F^{1,\mathtt{A}}$ and $F^{1,\mathtt{B}}$, we observe that the system $S_1$ is causal, lossless and \ac{FIFO}-per-flow, \ie Property~\ref{enum:spring:s-fifo-per-flow-lossless} of Theorem~\ref{thm:instab} holds.
    One can also verify from Table~\ref{tab:spring:sequences} that the delay of every packet through $S_1$ is lower-bounded by $0$ and upper-bounded by $d$. 
    This is clear for most packets.
    We do it for those that exchange their order: $B^1_{6k+3}-A^1_{6k+2}=d$ and $B^1_{6k+2}-A^1_{6k+3}=0$.
    Hence, Property~\ref{enum:spring:delay-through-s} holds.
    
    The sequence of packets $\mathscr{B}^1$ is now the input of the \ac{IR} in Figure~\ref{fig:fifo-sff:spring:n1}.
    The output $\mathscr{D}^1$ is obtained from $\mathscr{B}^1$ by applying the model of the \ac{IR} from section~\ref{sec:sysmodel:traffic-regulators}.
    Spring does not control the \ac{IR}.
    Let $\mathscr{D}^1=(D^1,L^{1,\mathtt{D}},F^{1,\mathtt{D}})$  be the packet sequence at the output of the \ac{IR}.
    As the \ac{IR} is causal, lossless and \ac{FIFO}, we have $L^{1,\mathtt{D}}=L^{1,\mathtt{B}}$ and $F^{1,\mathtt{D}}=F^{1,\mathtt{B}}$.
    Furthermore,
    
    \begin{lemma}\label{lemma:proof:spring:d-min}
        For any $h\in\mathbb{N}$, $D^1_{2h+1}\ge B^1_1+hI$, and\\$D^1_{2h+2}\ge B^1_1+hI+I$
    \end{lemma}
    \begin{proof}[Proof of Lemma~\ref{lemma:proof:spring:d-min}\label{proof:lemma:proof:spring:d-min}]
        We prove this by induction.
        \paragraph{Base case $h=0$}
        We have $F^{1,\mathtt{B}}_1=F^{1,\mathtt{B}}_2=f_1$. From Equation~\eqref{eq:sysmod:ir-model}, 
        \begin{align}
            D^1_1 &= \max\left(B_1^1, D^1_{0}, \Pi^{\gamma_{r,b}}([D^{1}]^{f_1},[L^{1,\mathtt{D}}]^{f_1})_{\text{index}(1,{f_1})} \right)\nonumber\\
            &\ge B^1_1  \nonumber\\
            D^1_1 &\ge B^1_1 + 0\cdot I \label{eq:proof:lemma:proof:spring:d-min:d1}
        \end{align}
        In addition, 
        \begin{align}
            D^1_2 &= \max\left(B_2^1, D^1_{1}, \Pi^{\gamma_{r,b}}([D^{1}]^{f_1},[L^{1,\mathtt{D}}]^{f_1})_{\text{index}(2,f_1)} \right)\nonumber \\ 
            &\ge \Pi^{\gamma_{r,b}}([D^{1}]^{f_1},[L^{1,\mathtt{D}}]^{f_1})_{\text{index}(2,{f_1})} \nonumber \\
            &\ge D_1^1 + \gamma_{r,b}^{\downarrow}(L_1^{1,\mathtt{D}}+L_2^{1,\mathtt{D}})\nonumber \\
            &\ge D_1^1 + |2b-b |^+ \cdot \frac{1}{r} \nonumber \\ 
            &\ge D_1^1 + I\nonumber \\
            &\ge B_1^1 + I &&\triangleright\ \eqref{eq:proof:lemma:proof:spring:d-min:d1}\nonumber\\
            D^1_2  &\ge B_1^1 + 0\cdot I + I \label{eq:proof:lemma:proof:spring:d-min:d2}
        \end{align}
        \paragraph{Induction step}
        Consider $h\in\mathbb{N}$, assume that 
        \begin{align}
            & D^1_{2h+1}\ge B^1_1+hI \label{eq:proof:lemma:proof:spring:d-min:2hplus1}\\
            \text{and}\quad &D^1_{2h+2}\ge B^1_1+hI+I \label{eq:proof:lemma:proof:spring:d-min:2hplus2}
        \end{align}
        We have
        \begin{align}
            D_{2(h+1)+1}^1 &= D_{2h+3}^1 && \nonumber \\
            & \ge D_{2h+2}^1 && \triangleright\ \eqref{eq:sysmod:ir-model} \nonumber \\
            & \ge B^1_1+(h+1) I && \triangleright\ \eqref{eq:proof:lemma:proof:spring:d-min:2hplus2} \label{eq:proof:lemma:proof:spring:d-min:2hplus3}
        \end{align}
        
        Now denote $f\triangleq F^{1,\mathtt{B}}_{2h+4}$.
        The index $2h+4$ is even, hence from the description of the Spring trajectory (Table~\ref{tab:spring:sequences}), we have that $F^{1,\mathtt{B}}_{2h+3}$ also equals $f$.
        Hence,
        \begin{align}
            D_{2h+4}^1 &\ge \Pi^{\gamma_{r,b}}\left( [D^1]^f, [L^{1,\mathtt{D}}]^f \right)_{\text{index}(2h+4,f)} &&\triangleright\ \eqref{eq:sysmod:ir-model}\nonumber \\
            &\ge D_{2h+3}^1 + \gamma_{r,b}^{\downarrow}(L_{2h+3}^{1,\mathtt{D}}+L_{2h+4}^{1,\mathtt{D}}) &&\triangleright\ \eqref{eq:sysmod:pi-gamma-regulator}\nonumber \\
            &\ge  D_{2h+3}^1 + |2b-b|^+ \cdot \frac{1}{r} \nonumber \\
            D_{2(h+1)+2}^1&\ge  B^1_1+(h+1) I + I &&\triangleright\ \eqref{eq:proof:lemma:proof:spring:d-min:2hplus3} \label{eq:proof:lemma:proof:spring:d-min:2hplus4}
        \end{align}
        Equation~\eqref{eq:proof:lemma:proof:spring:d-min:2hplus3} and \eqref{eq:proof:lemma:proof:spring:d-min:2hplus4} conclude the induction step of the proof of Lemma~\ref{lemma:proof:spring:d-min}.
    \end{proof}
    
    Now, for $k\in\mathbb{N}$,
    \begin{align*}
        D^1_{6k+1}-B^1_{6k+1} &\ge B_1^1 + 3kI-2d-k\tau \\
        &\qquad\qquad\qquad\qquad\triangleright\ \text{Lemma~\ref{lemma:proof:spring:d-min}}\\
        &\ge B_1^1 + 3kI-2d-3kI-3k\epsilon+kd \\
        &\qquad\qquad\qquad\qquad\triangleright\ \eqref{eq:spring-constants}\\
        &\ge k(d-3\epsilon) + B_1^1-2d
    \end{align*}
    As $\epsilon < \frac{d}{3}$, we obtain $\sup_{k\in\mathbb{N}}  D^1_{6k+1}-B^1_{6k+1} =+\infty$. 
    We can do the same for the other indices $6k+2,6k+3,\dots,6k+6$.
    Finally, this gives $\sup_{n\in\mathbb{N}^*}  D^1_{n}-B^1_{n} =+\infty$ and Property \ref{enum:spring:instable} of Theorem~\ref{thm:instab} holds.
    
    
    
    
    
    \end{proof}
    

%% file: figures/2022-11-spring-sequences.tex
\begin{tikzpicture}
    \node at (0,0) (before) {$\begin{NiceArray}{rrccccccccl}
        n: & \dots, & 6(k-1)+6 & 6k+1 & 6k+2 & 6k+3 & 6k+4 & 6k+5 & 6k+6 & 6(k+1)+1 & \dots\\
        \hline
        A^1=A^2=&(\dots, &3I+2\epsilon+(k-1)\tau, &d+k\tau, &I+\epsilon+k\tau, &I+d+k\tau, &2I+\epsilon + k\tau, &2I+2\epsilon+k\tau, &3I+2\epsilon+k\tau, &d+(k+1)\tau,&\dots) \\
        F^{1,\mathtt{A}}=F^{2,\mathtt{A}}= &(\dots, &{\color{orange}\boldsymbol{f_3}}, &{\color{blue}\boldsymbol{f_1}}, &{\color{red}\boldsymbol{f_2}}, &{\color{blue}\boldsymbol{f_1}}, &{\color{red}\boldsymbol{f_2}}, &{\color{orange}\boldsymbol{f_3}}, &{\color{orange}\boldsymbol{f_3}}, &{\color{blue}\boldsymbol{f_1}}, &\dots) \\
        \hline
        B^1=B^2=B^3=&(\dots, &3I+2\epsilon+(k-1)\tau, &2d+k\tau, &I+d+k\tau, &I+\epsilon+d+k\tau,  &2I+\epsilon+d+k\tau, &2I+2\epsilon+d+k\tau, &3I+2\epsilon+d+k\tau, &2d+(k+1)\tau,&\dots) \\
        F^{1,\mathtt{B}}=F^{3,\mathtt{B}}= &(\dots, &{\color{orange}\boldsymbol{f_3}}, &{\color{blue}\boldsymbol{f_1}}, &{\color{blue}\boldsymbol{f_1}}, &{\color{red}\boldsymbol{f_2}}, &{\color{red}\boldsymbol{f_2}}, &{\color{orange}\boldsymbol{f_3}}, &{\color{orange}\boldsymbol{f_3}}, &{\color{blue}\boldsymbol{f_1}}, &\dots)\\
        F^{2,\mathtt{B}}= &(\dots, &{\color{orange}\boldsymbol{f_3}}, &{\color{blue}\boldsymbol{f_1}}, &{\color{red}\boldsymbol{f_2}}, &{\color{blue}\boldsymbol{f_1}}, &{\color{red}\boldsymbol{f_2}}, &{\color{orange}\boldsymbol{f_3}}, &{\color{orange}\boldsymbol{f_3}}, &{\color{blue}\boldsymbol{f_1}}, &\dots)\\
        \hline
        D^1=D^3\ge&(\dots, &B^1_0+3(k-1)I+3I, &B^1_0+3kI, &B^1_0+3kI+I, &B^1_0+3kI+I,  &B^1_0+3kI+2I, &B^1_0+3kI+2I, &B^1_0+3kI+3I, &B^1_0+3(k+1)I,&\dots) \\
    \end{NiceArray}$};
    \draw[|-latex] ($(before.north)+(-5.2,0)$) -- ($(before.north west)+(2,0)$) node[pos=0.5,above]{$(k-1)^{\text{th}}$ period}; 
    \draw[|-latex] ($(before.north)+(8.4,0)$) -- ($(before.north east)$) node[pos=0.5,above]{$(k+1)^{\text{th}}$ period}; 
    \draw[latex-latex] ($(before.north)+(-5.2,0)$) -- ($(before.north)+(8.4,0)$) node[pos=0.5,above]{$k^{\text{th}}$ period}; 
    \draw[dashed] ($(before.north)+(-5.2,0)$) -- ($(before.south)+(-5.2,0)$);
    \draw[dashed] ($(before.north)+(8.4,0)$) -- ($(before.south)+(8.4,0)$);
\end{tikzpicture}

%% file: 9-appendix/2-proofs-section-5/2-proof-explains.tex
\subsection{Proof of Proposition~\ref{prop:fifo-sff:no-explain:pfr-sff-explain}}

\begin{proof}[\label{proof:fifo:pfr-sff-explain} Proof of Proposition~\ref{prop:fifo-sff:no-explain:pfr-sff-explain}]
    We first prove that the \ac{PFR} offers the fluid service curve $\beta^{\boldsymbol{\sigma}}\triangleq\sigma_f$. 
    A shaping curve $\sigma_f$ that meets the conditions of Proposition~\ref{prop:fifo-sff:no-explain:pfr-sff-explain} verifies \cite[Eq. (1.14)]{leboudecNetworkCalculus2001} (see the discussion in \cite[\S1.7.3]{leboudecNetworkCalculus2001}).
    From \cite[Thm. 1.7.3]{leboudecNetworkCalculus2001}, a \ac{PFR} configured with such a service curve can be realized as the concatenation of a greedy shaper $G$ with shaping curve $\sigma_f$, followed by a packetizer.
    From \cite[Thm 1.5.1]{leboudecNetworkCalculus2001}, $G$ offers $\sigma_f$ as a service curve.
    Hence, the \ac{PFR} offers the fluid service curve $\sigma_f$ as defined in Definition~\ref{def:fluid-service-curve}.

    We now prove that $\sigma_f$ explains the shaping-for-free property. 
    Consider two causal, lossless and \ac{FIFO} systems $S$ and $Z'$ such that $Z'$ offers the service curve $\beta^{\boldsymbol{\sigma}}=\sigma_f$ (Figure~\ref{fig:fifo-sff:no-explain:curve-explains-sff}). 
    Consider now the network made of $S$ followed by $Z'$ and the subset $X$ of the trajectories on this network such that $\forall x\in X, R_f^{\mathtt{A},x} \sim \sigma_f$.
    Denote by $D_f^{\mathtt{A}\rightarrow\mathtt{B}}$ the worst-case delay over $X$ of the flow $f$ between $\mathtt{A}$ and $\mathtt{B}$ and $D_f^{\mathtt{A}\rightarrow\mathtt{C}}$ the worst-case delay over $X$ of the flow $f$ between $\mathtt{A}$ and $\mathtt{C}$.
    
    If $D_f^{\mathtt{A}\rightarrow\mathtt{B}}=+\infty$, then by causality $D_f^{\mathtt{A}\rightarrow\mathtt{C}}=+\infty$ and the result holds.

    Assume that $D_f^{\mathtt{A}\rightarrow\mathtt{B}}$ is finite.
    By causality, 
    \begin{equation}\label{eq:proof:pfr-sff-explain:lower-bound-wc}
        D_f^{\mathtt{A}\rightarrow\mathtt{C}}\ge D_f^{\mathtt{A}\rightarrow\mathtt{B}}
    \end{equation}
    Within the subset of trajectories $X$, $S$ offers to $f$ the individual service curve $\beta_{S,f} = \delta_{D_f^{\mathtt{A}\rightarrow\mathtt{B}}}$.
    Hence, the concatenation of $S$ and $Z'$ offers to $f$ the individual service curve $\beta_{S+Z,f} = \beta_{S,f} \otimes \beta_{z,f} = \delta_{D_f^{\mathtt{A}\rightarrow\mathtt{B}}}\otimes \sigma_f$ \cite[Thm. 1.4.6]{leboudecNetworkCalculus2001}.
    Within $X$, $\sigma_f$ is an arrival curve for $f$ at $\mathtt{A}$.
    Hence
    \begin{align*}
        D_{f}^{\mathtt{A}\rightarrow\mathtt{C}} &\le h(\sigma_f, \delta_{D_f^{\mathtt{A}\rightarrow\mathtt{B}}}\otimes \sigma_f) \quad &&  \triangleright \text{\cite[Thm. 1.4.2]{leboudecNetworkCalculus2001}} \\
        & \le h(\sigma_f, \sigma_f \otimes \delta_{D_f^{\mathtt{A}\rightarrow\mathtt{B}}}) &&  \triangleright \otimes \text{ commutative}\\
        & \le h(\sigma_f, \delta_{D_f^{\mathtt{A}\rightarrow\mathtt{B}}}) &&  \triangleright \text{\cite[Lem. 1.5.2]{leboudecNetworkCalculus2001}}\\
        D_{f}^{\mathtt{A}\rightarrow\mathtt{C}} & \le D_f^{\mathtt{A}\rightarrow\mathtt{B}}
    \end{align*}
    Combined with (\ref{eq:proof:pfr-sff-explain:lower-bound-wc}), we obtain the result.

\end{proof}

%% file: 9-appendix/2-proofs-section-5/3-proof-lemma.tex
\subsection{Proof of Lemma~\ref{lem:fifo-sff:no-explain:minimum-sc-sff}}
\begin{proof}[Proof of Theorem~\ref{lem:fifo-sff:no-explain:minimum-sc-sff}\label{proof:lem:fifo-sff:no-explain:minimum-sc-sff}]

        Consider a set of shaping curves $\boldsymbol{\sigma}=\{\sigma_f\}_{f\in\mathcal{F}}$.
        Consider a wide-sense increasing function $\beta^{\boldsymbol{\sigma}}$.
        Assume that $\beta^{\boldsymbol{\sigma}}$ explains the shaping-for-free property (Definition~\ref{def:fifo-sff:no-explain:curve-explains-sff}).

        Consider the system $S$ that has no delay.
        $S$ is causal, lossless and \ac{FIFO}.
        Consider a causal, lossless and \ac{FIFO} system $Z'$ defined by, for all trajectory $x$, 
        \begin{equation}\label{eq:proof:lem:fifo-sff:no-explain:minimum-sc-sff:Z-def}
            R^{x,\mathtt{C}}=\min(R^{x,\mathtt{B}},\beta^{\boldsymbol{\sigma}})
        \end{equation}
        where $R^{x,\mathtt{B}}$ [\resp $R^{x,\mathtt{C}}$] is the cumulative function of the aggregate at the input $\mathtt{B}$ of $Z'$ [\resp at the output $\mathtt{C}$ of $Z'$], see Figure~\ref{fig:fifo-sff:no-explain:curve-explains-sff}.
        By Equation~\eqref{eq:proof:lem:fifo-sff:no-explain:minimum-sc-sff:Z-def}, $Z'$ offers $\beta^{\boldsymbol{\sigma}}$ as a service curve.

        Concatenate now $S$ and $Z'$ as in Figure~\ref{fig:fifo-sff:no-explain:curve-explains-sff} and consider the Trajectory $x_0$ with $\forall f\in\mathcal{F}, R_f^{x_0,\mathtt{A}}=\sigma_f$.
        $x_0$ belongs to the set $X$ defined in Definition~\ref{def:fifo-sff:no-explain:curve-explains-sff}.
        In addition, $S$ has no delay so
        \begin{equation}\label{eq:proof:lem:fifo-sff:no-explain:minimum-sc-sff:RB}
            R^{x_0,\mathtt{B}}=R^{x_0,\mathtt{A}}=\sum_{f\in\mathcal{F}}\sigma_f
        \end{equation}

        By Definition~\ref{def:fifo-sff:no-explain:curve-explains-sff}, the worst-case delay over the set of trajectories $X$ between $\mathtt{A}$ and $\mathtt{C}$ in Figure~\ref{fig:fifo-sff:no-explain:curve-explains-sff} equals the worst-case delay over $X$ between $\mathtt{A}$ and $\mathtt{B}$, \ie zero (by construction of $S$).
        By causality, the worst-case delay between $\mathtt{B}$  and $\mathtt{C}$ also equals 0 over $X$.

        As $x_0$ belongs to $X$, we have $R^{x_0,\mathtt{B}}=R^{x_0,\mathtt{C}}$.
        Combined with Equation~\eqref{eq:proof:lem:fifo-sff:no-explain:minimum-sc-sff:Z-def}, we have $\beta^{\boldsymbol{\sigma}}\ge R^{x_0,\mathtt{B}}$. 
        With Equation~\eqref{eq:proof:lem:fifo-sff:no-explain:minimum-sc-sff:RB}, we finally obtain
        \begin{equation}
            \beta^{\boldsymbol{\sigma}}\ge \sum_{f\in\mathcal{F}}\sigma_f
        \end{equation}

\end{proof}

%% file: 9-appendix/2-proofs-section-5/4-proof-ir.tex
\subsection{Proof of Theorem~\ref{thm:fifo-sff:no-explain:ir-has-no-explaination}}
\begin{proof}[Proof of Theorem~\ref{thm:fifo-sff:no-explain:ir-has-no-explaination}\label{proof:thm:fifo-sff:no-explain:ir-has-no-explaination}]
    Consider an \ac{IR} that shapes at least three flows $\mathcal{F}$ with the same leaky-bucket shaping curve for three of them $\forall f\in\{f_1,f_2,f_3\}\subset \mathcal{F} $, $\sigma_f = \gamma_{r,b}$ with $b\ge\max_{f\in\{f_1,f_2,f_3\}}{L^{\max}_f}$.
    Assume that there exists a wide-sense increasing curve $\beta^{\boldsymbol{\sigma}}$ such that (a) $\beta^{\boldsymbol{\sigma}}$ is a fluid service-curve of the \ac{IR} and (b) $\beta^{\boldsymbol{\sigma}}$ explains the shaping-for-free.
    Select an arbitrary $D>0$, and consider a causal, lossless, \ac{FIFO} system $S_2$ that offers the service curve $\delta_D$.
    Consider now the network $\mathcal{N}_2$ shown in Figure~\ref{fig:guarantees:sff:network-fifo}.
    Assume that the flows $f_1,f_2,f_3$ are $\gamma_{r,b}$-constrained at $\mathtt{A}$.

    We define the Spring constants $I,d,\epsilon,\tau$ as in (\ref{eq:spring-constants}).
    We define Trajectory 2 with $\mathscr{A}^2\triangleq\mathscr{A}^1$, where the supperscript $1$ denotes Trajectory $1$ described in Table~\ref{tab:spring:sequences}, and $\mathscr{B}^2\triangleq(B^2,L^{2,\mathtt{B}},F^{2,\mathtt{B}})$ with $B^2\triangleq B^1$, $L^{2,\mathtt{B}}\triangleq L^{1,\mathtt{B}}$.
    However, $F^{2,\mathtt{B}}$ differs from $F^{1,\mathtt{B}}$ by exchanging two packets in every period, as shown in Table~\ref{tab:spring:sequences}.
    We define the packet sequence $\mathscr{D}^2$ as the result of the input-output characterization of the \ac{IR} (Section~\ref{sec:sysmodel:traffic-regulators}) when it processes $\mathscr{B}^2$ as an input.

    \begin{figure}
        \resizebox{\linewidth}{!}{\input{./figures/2022-11-network-fifo.tex}}
        \caption{\label{fig:guarantees:sff:network-fifo} Several networks and their observation points used in the Proof of Theorem~\ref{thm:fifo-sff:no-explain:ir-has-no-explaination}.}
    \end{figure}
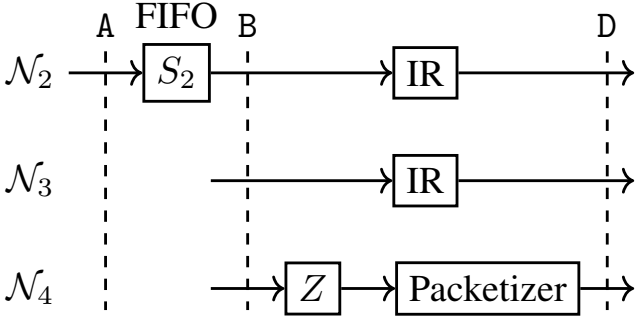

    One can check that Trajectory 2 is acceptable for the given constraints.
    For example, $\mathscr{A}^2=\mathscr{A}^1$, hence we know from Appendix~\ref{proof:thm:instab} that each flow is $\gamma_{r,b}$-constrained at $\mathtt{A}$ in Trajectory 2.
    
    In addition, $F^{2,\mathtt{B}}$ differs only slighty from $F^{1,\mathtt{B}}$, so we only need to recompute the delay through $S_2$ for the two packets with a different order: Packet $6k+2$ has a delay $B^2_{6k+2}-A^2_{6k+2}=d-\epsilon<d$ through $S_2$ and Packet $6k+3$ has a delay $B^2_{6k+3}-A^2_{6k+3}=\epsilon<d$ (Table~\ref{tab:spring:sequences}).
    Consequently, the delay through $S_2$ in Trajectory 2 is upper-bounded by $d$.
    By (\ref{eq:spring-constants}), $d<D$, so Trajectory 2 is acceptable with respect to $S_2$'s service-curve constraint $\delta_D$.

    Last, it is clear from Table~\ref{tab:spring:sequences} that the order of the packets is preserved between $\mathtt{A}$ and $\mathtt{B}$ in Trajectory 2, so Trajectory 2 is acceptable with respect to $S_2$'s \ac{FIFO} property.
    Also, $\mathscr{R}(\mathscr{B}^2)$ is constrained by the arrival curve $3\gamma_{r,b}\oslash\delta_D$ \cite[Thm. 1.4.3]{leboudecNetworkCalculus2001}.

    Now we consider the network $\mathcal{N}_3$ shown in Figure~\ref{fig:guarantees:sff:network-fifo}, and we define Trajectory 3 by $\mathscr{B}^3\triangleq\mathscr{B}^1$. 
    In particular, we have $B^3=B^2$.

    The operator $\mathscr{R}:(M^x,L^{x,\mathtt{M}},F^{x,\mathtt{M}})\mapsto R^{x,\mathtt{M}}$ defined in Appendix~\ref{sec:sysmodel:cumulative}, and that returns the cumulative function $R^{x,\mathtt{M}}$ of the traffic as a function of the packet sequence $(M^x,L^{x,\mathtt{M}},F^{x,\mathtt{M}})$, depends only on $M^x$ and $L^{x,\mathtt{M}}$, thus
    \begin{equation}\label{eq:guarantees:sff:r3r2}
        \mathscr{R}(\mathscr{B}^3)=\mathscr{R}(\mathscr{B}^2)\sim 3\gamma_{r,b}\oslash\delta_D
    \end{equation}
    

    We now use the fact that $\beta^{\boldsymbol{\sigma}}$ is a fluid service curve of the \ac{IR}.
    By Definition~\ref{def:fluid-service-curve}, the network $\mathcal{N}_3$ can be replaced by the network $\mathcal{N}_4$ and for the same input packet sequence $\mathscr{B}^4\triangleq \mathscr{B}^3$, we have the same output packet sequence $\mathscr{D}^4=\mathscr{D}^3$. 
    In particular, for any $n\in\mathbb{N}^*$, 
    \begin{equation}
        D^3_n-B^3_n = D^4_n-B^4_n
    \end{equation}
    In $\mathcal{N}_4$, the packetizer does not increase the per-packet delay and $Z$ offers the service-curve $\beta^{\boldsymbol{\sigma}}$.
    In addition, $\mathscr{B}^4=\mathscr{B}^3$ so by Equation~(\ref{eq:guarantees:sff:r3r2}), $3\gamma_{r,b}\oslash\delta_D$ is an arrival curve for the aggregate at $\mathtt{B}$ in $\mathcal{N}_4$.
    By Definition~\ref{def:fluid-service-curve}, $Z$ is also causal, lossless and \ac{FIFO}, hence from \cite[Thm 1.4.2]{leboudecNetworkCalculus2001} we obtain, $\forall n\in\mathbb{N}^*$,
    \begin{align}
        D^3_n-B^3_n &= D^4_n-B^4_n \\
        & \le h\left(3\gamma_{r,b}\oslash\delta_D, \beta^{\boldsymbol{\sigma}}\right) \\
        & \le h\left(3\gamma_{r,b}, \delta_D \otimes \beta^{\boldsymbol{\sigma}}\right) && \triangleright \text{\cite[Thm 3.1.12]{leboudecNetworkCalculus2001}}
    \end{align}
    $\beta^{\boldsymbol{\sigma}}$ explains the shaping-for-free.
    By Lemma~\ref{lem:fifo-sff:no-explain:minimum-sc-sff}, we have $\beta^{\boldsymbol{\sigma}}\ge\sum_{f\in\mathcal{F}}\sigma_f \ge 3\gamma_{r,b}$. 
    Combined with \cite[Lemma 1.5.2]{leboudecNetworkCalculus2001}, we obtain
    \begin{align}
        D^3_n-B^3_n & \le h\left(3\gamma_{r,b}, \delta_D\right) &&\\
        & \le D
    \end{align}
    But the packet sequence $\mathscr{B}^3$ in $\mathcal{N}_3$ is equal to the packet sequence $\mathscr{B}^1$ in the network $\mathcal{N}_1$ of Appendix~\ref{proof:spring} and Table~\ref{tab:spring:sequences}.
    Hence, $\mathscr{D}^1$ is equal to $\mathscr{D}^3$. 
    We have a contradiction because we show in Appendix~\ref{proof:spring} that $\sup_n D^1_n - B^1_n = +\infty$.
\end{proof}

%% file: figures/2022-11-network-fifo.tex
\begin{tikzpicture}
    \tikzstyle{n} = [draw]

    \node[n, label=above:{\acs{FIFO}}] at (0,0) (s2) {$S_2$};
    \node[n, anchor=west] at (2,0) (ir2) {\ac{IR}};
    \node[anchor=east] at (-1,0) (n2) {$\mathcal{N}_2$};
    \draw[->] (n2.east) -- (s2.west) node[pos=0.5,anchor=center](ta2){};
    \draw[->] (s2.east) -- (ir2.west) node[pos=0.5](tb2){};

    \node[n, anchor=west] at (2,-1) (ir3) {\ac{IR}};
    \node[anchor=east] at (-1,-1) (n3) {$\mathcal{N}_3$};
    \draw[->] let \p2=(s2.east) in  (\x2,-1) -- (ir3.west) node[pos=0.5](tb3){};

    \node[n, anchor=west] at (1,-2) (z4) {$Z$};
    \node[n,anchor=west] at ($(z4.east)+(0.5,0)$) (pl4) {Packetizer};
    \node[anchor=east] at (-1,-2) (n4) {$\mathcal{N}_4$};
    \draw[->] (z4.east) -- (pl4.west);
    \draw[->] let \p2=(s2.east) in  (\x2,-2) -- (z4.west) node[pos=0.5](tb4){};
    \draw[->] (pl4.east) -- ($(pl4.east)+(0.5,0)$)node[pos=0.5](td4){}node[pos=1,anchor=center](ttout){};
    
    \draw[->] let \p1=(ir3.east), \p2=(ttout.center) in (\x1,\y1) -- (\x2,\y1) node[pos=0.5](td3){};
    \draw[->] let \p1=(ir2.east), \p2=(ttout.center) in (\x1,\y1) -- (\x2,\y1) node[pos=0.5](td2){};


    \draw[dashed] let \p3=(ta2.center) in (\x3,0.2) -- (\x3,-2.2) node[pos=0,above]{$\mathtt{A}$};
    \draw[dashed] let \p3=(tb4.center) in (\x3,0.2) -- (\x3,-2.2) node[pos=0,above]{$\mathtt{B}$};
    \draw[dashed] let \p3=(td4.center) in (\x3,0.2) -- (\x3,-2.2) node[pos=0,above]{$\mathtt{D}$};
\end{tikzpicture}

%% file: 9-appendix/3-proofs-section-6/0-proofs-section-6.tex
\label{sec:appendix:proofs-ir-service-curves}

\input{9-appendix/3-proofs-section-6/1-strict-sc.tex}
\input{9-appendix/3-proofs-section-6/2-limit-strict-sc.tex}
\input{9-appendix/3-proofs-section-6/3-limit-residual.tex}
\input{9-appendix/3-proofs-section-6/4-limit-sc.tex}

%% file: 9-appendix/3-proofs-section-6/1-strict-sc.tex
\subsection{Proof of Theorem~\ref{thm:ir-service-curves:strict-sc}}

\begin{proof}[Proof of Theorem~\ref{thm:ir-service-curves:strict-sc}]\label{proof:thm:ir-service-curves:strict-sc}
    Consider an \ac{IR} that processes an aggregate $\mathcal{F}$ of flows with leaky-bucket shaping curves: $\forall f \in \mathcal{F}$, $\sigma_f=\gamma_{r_f,b_f}$ with $r_f>0$ and $b_f \ge L^{\max}_f$.
    Define $L^{\min}$ and $I^{\max}$ as in Theorem~\ref{thm:ir-service-curves:strict-sc}.
    Denote by $\mathtt{B}$ [\resp $\mathtt{D}$] an observation point located at the input [\resp at the output] of the \ac{IR} (as in Figure~\ref{fig:sysmodel:traffic-regulators:simple-ir}).

    Consider an acceptable trajectory $x$.
    Consider two time instants $s < t$ such that $]t,s]$ is a backlogged period for the \ac{IR} and consider $\epsilon\in\mathbb{R}$ such that $0<\epsilon<t-s$.
    From the definition of the backlogged period, the cumulative functions at the input and the output verify $R^{x,\mathtt{B}}(s+\epsilon)>R^{x,\mathtt{D}}(s+\epsilon)$.

    Both the input and the output of the \ac{IR} are packetized.
    Hence, consider now the packet sequences $\mathscr{B}^x=(B^x,L^{x,\mathtt{B}},F^{x,\mathtt{B}})$ and $\mathscr{D}^x=(D^x,L^{x,\mathtt{D}},F^{x,\mathtt{D}})$ in trajectory $x$, at observation point $\mathtt{B}$ and $\mathtt{D}$ (as defined in Section~\ref{sec:sysmodel:marked-point-notation}).
    The \ac{IR} is causal, lossless and \ac{FIFO}, so we have $L^{x,\mathtt{B}}=L^{x,\mathtt{D}}$ [\resp $F^{x,\mathtt{B}}=F^{x,\mathtt{D}}$]. 
    Denote by $L^x$ [\resp $F^x$] this sequence.

    At $s+\epsilon$, we have $R^{x,\mathtt{B}}(s+\epsilon)>R^{x,\mathtt{D}}(s+\epsilon)$.
    The \ac{IR} is lossless, so at $s+\epsilon$, it is non-empty.
    In addition, its input and output are packetized, so the \ac{IR} contains at least one packet at $t+\epsilon$.
    The \ac{IR} is also \ac{FIFO}, so denote by $n$ the index within the sequences $\mathscr{B}^x$, $\mathscr{D}^x$ of the head-of-line packet inside the \ac{IR} at $s+\epsilon$ and denote $g\triangleq F_n^x$ the flow to which packet $n$ belongs in trajectory $x$.

    We now use the \ac{IR} model from \cite{boyerEquivalenceTheoreticalModel2022} to compute the release time instant of packet $n$ in trajectory $x$:
    \begin{equation}
        D_n^x=\max\left\lbrace B_n^x, D_{n-1}^x, \frac{L_n^x-\Lambda^x_{g,n\ominus 1}}{r_{g}}+D^x_{n\ominus 1}\right\rbrace
    \end{equation}
    where $\Lambda^x_{g,k}$ is the number of tokens available for flow $g$ in its token-bucket filter after packet $k$ left the \ac{IR} (see Appendix~\ref{sec:sysmodel:traffic-regulators}).

    At time instant $s+\epsilon$, packet $n$ is the head-of-line packet in the \ac{IR}, so $B_n\le s+\epsilon$, $D_{n-1}\le s+\epsilon$ and $D_{n\ominus1}\le s+\epsilon$.
    Hence,
    \begin{align}
        D_n^x&\le\max\left\lbrace s+\epsilon,\frac{L_n^x-\Lambda^x_{g,n\ominus1}}{r_g}+s+\epsilon\right\rbrace\\
        &\le \max\left\lbrace s+\epsilon, \frac{L_n^x}{r_g}+s+\epsilon \right\rbrace&&\triangleright \Lambda^x_{g,n\ominus 1}\ge 0\\
        &\le \frac{L_n^x}{r_g}+s+\epsilon&&\triangleright L_n^x \ge 0 \label{eq:proof:thm:ir-service-curves:strict-sc:d-upper-bound}
    \end{align}
    With the terminology of \cite{boyerEquivalenceTheoreticalModel2022}, Equation (\ref{eq:proof:thm:ir-service-curves:strict-sc:d-upper-bound}) states that packet $n$ stays at the head of the line for a duration that is less than the time required for the associated token bucket to regain as many credits as the length of $n$.
    We obtain 
    \begin{align}
        D_n^x &\le \frac{L_g^{\max}}{r_g} + s+\epsilon\\
        &\le I^{\max} + s + \epsilon
    \end{align}
    
    The size $L_n^x$ of packet $n$ is greater than $L^{\min}$, hence the cumulative function $R^{x,\mathtt{D}}$ at the output of the \ac{IR} in trajectory $x$ verifies:
    \begin{align}
        R^{x,\mathtt{D}}(t)-R^{x,\mathtt{D}}(s) &\ge \left\lbrace
            \begin{aligned}
                0 &&\quad \text{ if } t -s \le I^{\max} + \epsilon \\
                L^{\min} &&\quad \text{ if } t - s > I^{\max} + \epsilon \\
            \end{aligned}\right.\\
            R^{x,\mathtt{D}}(t)-R^{x,\mathtt{D}}(s) &\ge L^{\min} \wedge \left(\delta_{I^{\max}+\epsilon}\right)(t-s) \label{eq:proof:thm:ir-service-curves:strict-sc:stric-epsilon}
    \end{align}
    where $\wedge$ is the minimum and $\delta$ is the bounded-delay curve (Table~\ref{tab:extended-notations}).
    Equation~\eqref{eq:proof:thm:ir-service-curves:strict-sc:stric-epsilon} is valid for any $\epsilon$ such that $0<\epsilon<t-s$, hence
    \begin{equation}
        R^{x,\mathtt{D}}(t)-R^{x,\mathtt{D}}(s) \ge L^{\min} \wedge \left(\delta_{I^{\max}}(t-s)\right) \label{eq:proof:thm:ir-service-curves:strict-sc:stric-no-epsilon}
    \end{equation}
    Equation~\eqref{eq:proof:thm:ir-service-curves:strict-sc:stric-no-epsilon} is valid for any acceptable trajectory $x$ with no assumption on the context of the \ac{IR} and for any backlogged period $]s,t]$.
    In addition, note that the function $\beta_0:t\mapsto L^{\min} \wedge \delta_{I^{\max}}(t)$ is wide-sense increasing.
    By definition, the \ac{IR} offers $\beta_0$ as a context-agnostic strict service curve.
    The curve $\beta_0$ is shown with a dashed-blue line in Figure~\ref{fig:ir-service-curves:strict-sc:result-figure}.

    Let us now define the curve $\beta_{\text{sc}}$ as
    \begin{equation}
        \beta_{\text{sc}}:t\mapsto \left\lfloor\frac{t}{I^{\max}}\right\rfloor\cdot L^{\min}
    \end{equation}
    where $\lfloor \cdot \rfloor$ denotes the floor function.
    We observe that $\beta_{\text{sc}}$ is the supper-additive closure of $\beta_0$ \cite[\S 2.4]{bouillardDeterministicNetworkCalculus2018}.
    Consequently, the \ac{IR} also offers $\beta_{\text{sc}}$ as a strict service curve \cite[Prop. 5.6.2]{bouillardDeterministicNetworkCalculus2018}.

    In addition, any curve $\beta'$ such that $\forall t\ge 0$, $\beta'(t)\le \beta_{\text{sc}}$ is also a strict service curve of the \ac{IR}.
    This is the case, in particular, for the rate-latency service curve $\beta_{R,T}$ with $R=\frac{L^{\min}}{I^{\max}}$ and $T=I^{\max}$. 

\end{proof}

%% file: 9-appendix/3-proofs-section-6/2-limit-strict-sc.tex
\subsection{Proof of Proposition~\ref{prop:ir-service-curves:strict-sc:limit}}
\begin{proof}[Proof of Proposition~\ref{prop:ir-service-curves:strict-sc:limit}]\label{proof:prop:ir-service-curves:strict-sc:limit}
    Consider an \ac{IR} that processes an aggregate $\mathcal{F}$ with leaky-bucket shaping curves $\{\sigma_f\}_{f\in\mathcal{F}}=\{\gamma_{r_f,b_f}\}_{f\in\mathcal{F}}$.
    Take a curve $\beta^{\text{strict}}\in\mathfrak{F}_0$ and assume that $\beta^{\text{strict}}$ is a context-agnostic strict service curve of the \ac{IR}.

    Take $f\in\mathcal{F}$, and consider trajectory $x_f$.
    Define $I_f\triangleq\frac{b_f}{r_f}$.
    Let the packet sequence $\mathscr{B}^{x_f}$ at the input $\mathtt{B}$ of the \ac{IR} be $B^{x_f}=(0,\frac{I_f}{2},I_f,\frac{3I_f}{2},2I_f,\dots)$, $L^{x_f,\mathtt{B}}=(b_f,b_f,b_f,b_f,b_f,\dots)$, $F^{x_f,\mathtt{B}}=(f,f,f,f,f,\dots)$.
    That is, only flow $f$ sends packets of size $b_f$ at twice the rate allowed by its shaping curve.
    There are no packets for flows $\mathcal{F}\backslash\{f\}$.

    From the input-output characterisation of the \ac{IR} (Appendix~\ref{sec:sysmodel:traffic-regulators}), we directly obtain that the packet sequence $\mathscr{D}^x$  at the output $\mathtt{D}$ of the \ac{IR} is $L^{x_f,\mathtt{D}}=L^{x_f,\mathtt{B}}$, $F^{x_f,\mathtt{D}}=F^{x_f,\mathtt{B}}$ and $D^{x_f}=(0,I_f,2I_f,3I_f,4I_f,\dots)$.

    Hence, $\sigma_f$ is an arrival curve for the output traffic $R^{x_f,\mathtt{D}}$ and the \ac{IR} is non-empty after $\frac{I_f}{2}$.
    By combining \eqref{eq:background:strict-sc-def} and \eqref{eq:background:strict-ac-def}, this gives
    \begin{equation}\label{eq:proof:prop:ir-service-curves:strict-sc:limit:ineq}
        \forall t\ge0, \quad \beta^{\text{strict}}(t) \le \sigma_f(t)
    \end{equation}
    Equation~\eqref{eq:proof:prop:ir-service-curves:strict-sc:limit:ineq} is valid for any $f\in\mathcal{F}$, hence the result.

\end{proof}

%% file: 9-appendix/3-proofs-section-6/3-limit-residual.tex
\subsection{Proof of Theorem~\ref{thm:ir-service-curves:limit-ir-residual}}

\begin{proof}[Proof of Theorem~\ref{thm:ir-service-curves:limit-ir-residual}]\label{proof:thm:ir-service-curves:limit-ir-residual}
   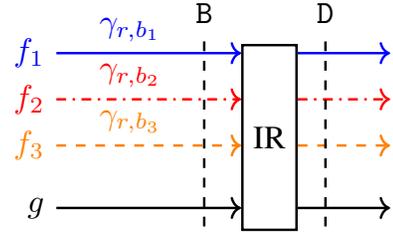
\begin{figure}
   \centering
      \resizebox{0.6\linewidth}{!}{\input{./figures/2023-04-ir-residual-situation.tex}}
      \caption{\label{fig:proof:thm:ir-service-curves:limit-ir-residual:situation} Situation for the Proof of Theorem~\ref{thm:ir-service-curves:limit-ir-residual}. The flows $f_1,f_2,f_3$ are processed by the \ac{IR} with the same leaky-bucket shaping curve $\gamma_{r,b}$. They enter the \ac{IR} with leaky-bucket arrival curves $\gamma_{r_1,b_1}$,$\gamma_{r_2,b_2}$ and $\gamma_{r_3,b_3}$. We are interested in an \emph{individual} service curve that the \ac{IR} offers to a fourth flow $g$.}
   \end{figure}
   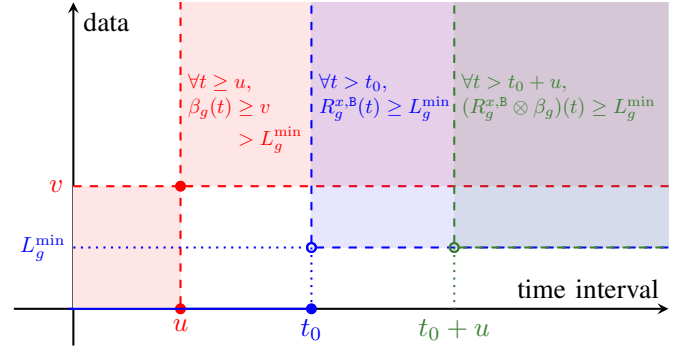
\begin{figure}
      \centering
         \resizebox{\linewidth}{!}{\input{./figures/2023-04-min-plus-convolution.tex}}
         \caption{\label{fig:proof:thm:ir-service-curves:limit-ir-residual:min-plus-conv} 
         If we assume that $\exists u \ge 0, \beta_g(u)=v>L_g^{\min}$ as in Lemma~\ref{lemma:proof:thm:ir-service-curves:limit-ir-residual:upper-bound-first-packet}, then this figure depicts the envelopes of possible values for the \emph{individual} service curve $\beta_g$ for $g$ (red area), for the cumulative function $R_g^{x,\mathtt{B}}$ of $g$ at the input $\mathtt{B}$ of the \ac{IR} with a trajectory $x\in X'$ (in blue), and for the resulting min-plus convolution $R_g^{x,\mathtt{B}}\otimes\beta_g$ (in green).}
   \end{figure}

   Consider an \ac{IR} that processes an aggregate $\mathcal{F}$ of at least four flows with the same leaky-bucket shaping curve for three of them: $\forall f_i\in\{f_1,f_2,f_3\}$, $\sigma_{f_i} = \gamma_{r,b}$.
   Consider a flow $g\in\mathcal{F}\backslash\{f_1,f_2,f_3\}$.
   For $f_i\in\{f_1,f_2,f_3\}$, assume that the flow $f_i$ is constrained at the input $\mathtt{B}$ of the \ac{IR} by a leaky-bucket arrival curve $\gamma_{r,b_i}$ with $b_1 > b$, $b_2\ge b$ and $b_3\ge b$ (Figure~\ref{fig:proof:thm:ir-service-curves:limit-ir-residual:situation}).
   Assume that any other flow $h\in\mathcal{F}\backslash\{f_1,f_2,f_3,g\}$ exhibits an arrival curve $\alpha_h^{\mathtt{B}}$ at the input $\mathtt{B}$ of the interleaved regulator.
   Finally, assume that the \ac{IR} offers to $g$ the individual service curve $\beta_g$. 
   This curve can depend on the shaping curves $\{\sigma_{f}\}_{f\in\mathcal{F}}$ and on the arrival curves of the other flows $\alpha_{f_1}^{\mathtt{B}}=\gamma_{r,b_1}$, $\alpha_{f_2}^{\mathtt{B}}=\gamma_{r,b_2}$, $\alpha_{f_3}^{\mathtt{B}}=\gamma_{r,b_3}$, $\{\alpha_h^{\mathtt{B}}\}_{h\in\mathcal{F}\backslash\{f_1,f_2,f_3,g\}}$.
   Consider $X'$, the subset of trajectories such that for each $x\in X'$, the first packet of $g$ that enters the \ac{IR} is of size $L_g^{\min}$.

   We obtain the proof of Theorem~\ref{thm:ir-service-curves:limit-ir-residual} through two lemmas.
   \begin{lemma}\label{lemma:proof:thm:ir-service-curves:limit-ir-residual:upper-bound-first-packet}
      If there exists $u\in\mathbb{R}^+$ such that $\beta(u)>L_g^{\min}$, then for any trajectory $x\in X'$, the delay of the first packet of $g$ in the \ac{IR} is upper-bounded by $u$.
   \end{lemma}
   \begin{proof}[Proof of Lemma~\ref{lemma:proof:thm:ir-service-curves:limit-ir-residual:upper-bound-first-packet}\label{proof:lemma:proof:thm:ir-service-curves:limit-ir-residual:upper-bound-first-packet}]
      Denote by $v$ the value of $\beta_g(u)$.
      We have $v>L_g^{\min}$.
      $\beta_g$ is an individual service curve of the \ac{IR} for $g$.
      By definition, it is wide-sense increasing.
      Hence, 
      \begin{equation}\label{eq:proof:thm:ir-service-curves:limit-ir-residual:beta-increase}
         \forall t\ge u, \quad \beta_g(t)\ge v> L_g^{\min}
      \end{equation}
      The red-shaded area in Figure~\ref{fig:proof:thm:ir-service-curves:limit-ir-residual:min-plus-conv} represents the possible values of $\beta_g$.

      Recall that $\mathtt{B}$, the observation point at the input of the \ac{IR} (Figure~\ref{fig:proof:thm:ir-service-curves:limit-ir-residual:situation}), is packetized.
      Consider a trajectory $x\in X'$ and denote by $t_0$ the arrival time instant of the first packet of $g$ at observation point $\mathtt{B}$ in trajectory $x$.
      The cumulative arrival function $R_g^{x,\mathtt{B}}$ of $g$ at $\mathtt{B}$ is wide-sense increasing and left-continuous, hence
      \begin{equation}\label{eq:proof:thm:ir-service-curves:limit-ir-residual:r-in-increase}
         \forall t> t_0, R_{g}^{x,\mathtt{B}}(t) \ge L_g^{\min} 
      \end{equation}
      The blue-shaded area in Figure~\ref{fig:proof:thm:ir-service-curves:limit-ir-residual:min-plus-conv} represents the possible value of $R_{g}^{x,\mathtt{B}}$.

      Consider now the min-plus convolution $R_{g}^{x,\mathtt{B}}\otimes\beta_g$. For $t> t_0+u$, we have
      \begin{align*}
         (R_{g}^{x,\mathtt{B}}\otimes\beta_g)(t) &= \inf_{0\le s \le t}\beta_g(s) + R_g^{x,\mathtt{B}}(t-s)\\
         & =\min\left(\begin{aligned}
            \inf_{0\le s \le u} \beta_g(s)+R_g^{x,\mathtt{B}}(t-s),\\
            \inf_{u<s\le t}\beta_g(s)+R_g^{x,\mathtt{B}}(t-s)\end{aligned}\right)
      \end{align*}
      
      For $s\in[0,u]$, $\beta_g(s) \ge 0$ and $t-s \ge t-u > t_0$, so $R_g^{x,\mathtt{B}}(t-s)\ge L_g^{\min}$ \eqref{eq:proof:thm:ir-service-curves:limit-ir-residual:r-in-increase}.
      For $s\in ]u,t]$, $\beta_g(s) \ge L_g^{\min}$ \eqref{eq:proof:thm:ir-service-curves:limit-ir-residual:beta-increase} and $R_g^{x,\mathtt{B}}(t-s)\ge 0$.
      Hence, $\forall t > t_0+u$,
      \begin{equation}
         (R_{g}^{x,\mathtt{B}}\otimes\beta_g)(t) \ge \min(0+L_g^{\min}, L_g^{\min}+0) = L_g^{\min}
      \end{equation}

      By assumption, $\beta_g$ is an individual service curve of the \ac{IR} for $g$, so 
      \begin{equation}\label{eq:proof:thm:ir-service-curves:limit-ir-residual:r-out-min}
         \forall t > t_0+u, \quad R_{g}^{x,\mathtt{D}}(t) \ge (R_{g}^{x,\mathtt{B}}\otimes\beta_g)(t) \ge L_g^{\min}
      \end{equation}

      The \ac{IR} is causal, \ac{FIFO} and lossless, and its output is packetized. 
      Hence Equation~\eqref{eq:proof:thm:ir-service-curves:limit-ir-residual:r-out-min} implies that the first packet of $g$ exits the \ac{IR} before $t_0 +u$, \ie its delay through the \ac{IR} is upper-bounded by $u$.      
   \end{proof}

   Let us now state the second lemma:
   \begin{lemma}\label{lemma:proof:thm:ir-service-curves:limit-ir-residual:unbounded-first-delay}
      For any duration $M\in\mathbb{R}^+$, there exists an acceptable trajectory $x_M\in X'$ such that the delay of the first packet of $g$ within the \ac{IR} is greater than $M$.
   \end{lemma}

   \begin{proof}[Proof of Lemma~\ref{lemma:proof:thm:ir-service-curves:limit-ir-residual:unbounded-first-delay} \label{proof:lemma:proof:thm:ir-service-curves:limit-ir-residual:unbounded-first-delay}]
      For the three flows $f_1,f_2,f_3$, consider the Trajectory 3 of Table~\ref{tab:spring:sequences} with
      \begin{equation}
         D \triangleq  \frac{b_1}{r} - \frac{b}{r}
      \end{equation}
      and the rest of the Spring parameters $I,d,\epsilon,\tau$ as in Equation~\eqref{eq:spring-constants}.
      $D>0$ because $b_1>b$.
      Note that 
      \begin{equation}\label{eq:proof:lemma:proof:thm:ir-service-curves:limit-ir-residual:unbounded-first-delay:D}
         D =  I -\frac{2b}{r} + \frac{b_1}{r}
      \end{equation}

      In Trajectory 3, the size of the packets is $b$ and the duration between any two consecutive packets of $f_2$ [\resp $f_3$] is at least $I=b/r$ at $\mathtt{B}$, so $\gamma_{r,b_2}$ [\resp $\gamma_{r,b_3}$] is an arrival of $f_2$ [\resp $f_3$] at $\mathtt{B}$ in Trajectory 3 (because $b_2\ge b$ and $b_3 \ge b$).
      Similarly, the duration between any two packets of $f_1$ of size $b$ is $I-d$ and
      \begin{align*}
         -D &< -d &&\quad\triangleright \eqref{eq:spring-constants}\\
         -I +\frac{2b-b_1}{r} &< -d &&\quad\triangleright\eqref{eq:proof:lemma:proof:thm:ir-service-curves:limit-ir-residual:unbounded-first-delay:D}\\
         2b &<r(I-d)+b_1 &&
      \end{align*}
      This shows that $\gamma_{r,b_1}$ is an arrival curve for $f_1$ at $\mathtt{B}$ in Trajectory 3.

      From Trajectory 3, we obtain trajectory $x_M$ by inserting a packet of $f_3$, of size $L_g^{\min}$, between the ($6k+1$)th and ($6k+2$)th packets of Trajectory 3, with 
      \begin{equation}\label{eq:proof:lemma:proof:thm:ir-service-curves:limit-ir-residual:unbounded-first-delay:k}
         k \triangleq \left\lceil \frac{M-d+I}{d-3\epsilon}\right\rceil + 1
      \end{equation}
      where $\lceil\cdot\rceil$ is the ceiling function.
      This packet takes the index $6k+2$ in trajectory $x_M$, and we have
      \begin{equation}\label{eq:proof:lemma:proof:thm:ir-service-curves:limit-ir-residual:unbounded-first-delay:input}
         B^3_{6k+1} \le B^{x_M}_{6k+2} \le B^{3}_{6k+2} 
      \end{equation}
      No other flow $h\in\mathcal{F}\backslash\{f_1,f_2,f_3,g\}$ sends any packet in Trajectory $x_M$.
      For the three flows $f_1$, $f_2$, $f_3$, Trajectory $x_M$ describes the same packet sequence in $\mathtt{B}$ as in Trajectory 3.
      Hence, $\gamma_{r_,b_1},\gamma_{r,b_2},\gamma_{r,b_3}$ are also arrival curves of $f_1$, $f_2$, $f_3$ in Trajectory $x_M$.
      We define $\mathscr{D}^{x_M}$ as the result of the \ac{IR}'s input-output model (Section~\ref{sec:appendix:notations-proofs}) on input $\mathscr{B}^{x_M}$.
      Hence, $x_M$ is an acceptable trajectory, and, by definition $x_M \in X'$.

      The packet sequences $\mathscr{B}^{x_M}$ and $\mathscr{B}^{3}$ are identical until index $6k+1$ included and the \ac{IR} is causal.
      Hence, the packet sequences $\mathscr{D}^{x_M}$ and $\mathscr{D}^{3}$ are also identical until index $6k+1$ included.
      As the \ac{IR} is \ac{FIFO}, we obtain
      \begin{equation}\label{eq:proof:lemma:proof:thm:ir-service-curves:limit-ir-residual:unbounded-first-delay:output}
         D^{x_M}_{6k+2} \ge D^{x_M}_{6k+1} = D^{3}_{6k+1}
      \end{equation}
      
      By combining \eqref{eq:proof:lemma:proof:thm:ir-service-curves:limit-ir-residual:unbounded-first-delay:input} and \eqref{eq:proof:lemma:proof:thm:ir-service-curves:limit-ir-residual:unbounded-first-delay:output}, we obtain
      \begin{align}
         D^{x_M}_{6k+2} - B^{x_M}_{6k+2} &\ge D^{3}_{6k+1} - B^{3}_{6k+2} &&\nonumber\\
         & \ge (d-I) + k(3I-\tau) \nonumber\\ 
            &\qquad\qquad\qquad\qquad\qquad\triangleright\text{Table }\ref{tab:spring:sequences}\nonumber\\
         &\ge (d-I) + k(d-3\epsilon) \nonumber\\ 
            &\qquad\qquad\qquad\qquad\qquad\triangleright\eqref{eq:spring-constants}\nonumber\\
         & \ge (d-I) + \frac{M-d+I}{d-3\epsilon} (d-3\epsilon) \nonumber\\
            &\qquad\qquad\qquad\qquad\qquad\triangleright\eqref{eq:proof:lemma:proof:thm:ir-service-curves:limit-ir-residual:unbounded-first-delay:k}\nonumber\\
         D^{x_M}_{6k+2} - B^{x_M}_{6k+2} & \ge M &&
      \end{align}
      This gives the result of the Lemma.
   \end{proof}

   To prove Theorem~\ref{thm:ir-service-curves:limit-ir-residual}, we now simply combine Lemma~\ref{lemma:proof:thm:ir-service-curves:limit-ir-residual:unbounded-first-delay} with the contrapose of Lemma~\ref{lemma:proof:thm:ir-service-curves:limit-ir-residual:upper-bound-first-packet}.

   Lemma~\ref{lemma:proof:thm:ir-service-curves:limit-ir-residual:unbounded-first-delay} states that for any $u\in\mathbb{R}^+$, there exists a trajectory $x\in X'$ such that the delay of the first packet of $g$ is not upper-bounded $u$ (it is strictly larger than $u$).
   This means that the consequent of Lemma~\ref{lemma:proof:thm:ir-service-curves:limit-ir-residual:upper-bound-first-packet} is false.
   By contrapose, the antecedent is also false.

\end{proof}

%% file: figures/2023-04-ir-residual-situation.tex
\begin{tikzpicture}
    \node[draw, minimum height=2cm] at (0,0) (ir) {IR};
    \draw[->, blue] ($(ir.north west) + (-2,-0.1)$) -- ($(ir.north west)+(0,-0.1)$)             node[pos=0,left] {$f_1$}        node[pos=0.4,above] {$\gamma_{r,b_1}$} node[pos=0.8,anchor=center](tbup){};
    \draw[->, red, dashdotted] ($(ir.north west) + (-2,-0.6)$) -- ($(ir.north west)+(0,-0.6)$)  node[pos=0,left] {$f_2$}        node[pos=0.4,above] {$\gamma_{r,b_2}$};
    \draw[->, orange, dashed] ($(ir.north west) + (-2,-1.1)$) -- ($(ir.north west)+(0,-1.1)$)   node[pos=0,left] {$f_3$}        node[pos=0.4,above] {$\gamma_{r,b_3}$};
    \draw[->] ($(ir.south west) + (-2,0.25)$) -- ($(ir.south west)+(0,0.25)$)                   node[pos=0,left] {$g$}                                                 node[pos=0.8,anchor=center](tbdown){};

    \draw[->, blue] ($(ir.north east) + (0,-0.1)$) -- ($(ir.north east)+(1,-0.1)$)                  node[pos=0.3,anchor=center](tdup){};
    \draw[->, red, dashdotted] ($(ir.north east) + (0,-0.6)$) -- ($(ir.north east)+(1,-0.6)$);
    \draw[->, orange, dashed] ($(ir.north east) + (0,-1.1)$) -- ($(ir.north east)+(1,-1.1)$);
    \draw[->] ($(ir.south east) + (0,0.25)$) -- ($(ir.south east)+(1,0.25)$)                        node[pos=0.3,anchor=center](tddown){};

    \draw[dashed] ($(tbdown)+(0,-0.2)$) -- ($(tbup)+(0,0.2)$) node[pos=1,above]{$\mathtt{B}$};
    \draw[dashed] ($(tddown)+(0,-0.2)$) -- ($(tdup)+(0,0.2)$) node[pos=1,above]{$\mathtt{D}$};
\end{tikzpicture}

%% file: figures/2023-04-min-plus-convolution.tex
\begin{tikzpicture}
    \tikzstyle{ma} = [
        shorten >=-2pt,
        shorten <=-2pt
    ]
    \begin{axis}[
        axis x line=center,
        axis y line=center,
        xmin=-0.5,
        ymin=-0.6,
        xmax=5,
        ymax=5,
        xlabel={time interval},
        ylabel={data},
        ticks=none,
        width=10cm,
        height=6cm
    ]    
    \def\xBeta{0.9}
    \def\xR{2}
    \def\xRStar{3.2}
    
    \path[fill=red!10] (0,0) -- (\xBeta,0) -- (\xBeta,2) -- (0,2) -- cycle;
    \path[fill=red, opacity=0.1] (\xBeta,2) -- (5,2) -- (5,5) -- (\xBeta,5) -- cycle;

    \draw[red, ma, {Circle[length=4pt,width=4pt]}-, dashed] (\xBeta,2) -- (5,2);
    \draw[red, dashed] (0,2) -- (1,2) node[pos=0,left]{$v$};
    \draw[red, ma, {Circle[length=4pt,width=4pt]}-, dashed] (\xBeta,0) -- (\xBeta,5) node[pos=0,below]{$u$};

    \draw[blue, ma, -{Circle[length=4pt,width=4pt]}] (0,0) -- (\xR,0);
    \draw[blue, ma, {Circle[open,length=4pt,width=4pt]}-, dashed] (\xR,1) -- (5,1);
    \draw[blue, ma, {Circle[open,length=4pt,width=4pt]}-, dashed] (\xR,1) -- (\xR,5);
    \draw[dotted, blue] (\xR,1) -- (\xR,0) node[pos=1,below]{$t_0$};
    \draw[dotted, blue] (0,1) -- (\xR,1) node[pos=0,left,scale=0.8]{$L_g^{\min}$};
    \path[fill=blue, opacity=0.1] (\xR,1) -- (5,1) -- (5,5) -- (\xR,5) -- cycle;

    \draw[OliveGreen, ma, {Circle[open, length=4pt,width=4pt]}-, dashed] (\xRStar,1) -- (5,1);
    \draw[OliveGreen, ma, {Circle[open, length=4pt,width=4pt]}-, dashed] (\xRStar,1) -- (\xRStar,5);
    \draw[dotted, OliveGreen] (\xRStar,1) -- (\xRStar,0) node[pos=1,below] {$t_0+u$};
    \path[fill=OliveGreen, opacity=0.1] (\xRStar,1) -- (5,1) -- (5,5) -- (\xRStar,5) -- cycle;

    \node[scale=0.75, anchor=north west, red] at  (\xBeta,4) {\makecell[l]{$\forall t\ge u,$\\$\begin{aligned}\beta_g(t)&\ge v\\ &>L_g^{\min}\end{aligned}$}};
    \node[scale=0.75, anchor=north west, blue] at (\xR,4) {\makecell[l]{$\forall t>t_0,$\\$R^{x,\mathtt{B}}_g(t) \ge L_g^{\min}$}};
    \node[scale=0.75, anchor=north west, OliveGreen] at (\xRStar,4) {\makecell[l]{$\forall t>t_0+u,$\\$(R^{x,\mathtt{B}}_g\otimes \beta_g)(t) \ge L_g^{\min}$}};

    \end{axis}
\end{tikzpicture}

%% file: 9-appendix/3-proofs-section-6/4-limit-sc.tex
\subsection{Proof of Theorem~\ref{thm:ir-service-curves:limit-sc}}

\begin{proof}[Proof of Theorem~\ref{thm:ir-service-curves:limit-sc}]\label{proof:thm:ir-service-curves:limit-sc}

    Consider an \ac{IR} that processes an aggregate $\mathcal{F}$ of at least four flows with leaky-bucket shaping shaping curves $\{\sigma_h\}_{h\in\mathcal{F}}$.
    Assume that at least three flows $f_1,f_2,f_3 \in \mathcal{F}$ share the same leaky-bucket shaping curve: $\forall f_i\in\{f_1,f_2,f_3\}, \sigma_{f_i}=\gamma_{r,b}$.
    Consider now a function $\beta\in\mathfrak{F}_0$. 
    $\beta$ can depend on $\{\sigma_h\}_{h\in\mathcal{F}}$.
    Assume that the \ac{IR} offers $\beta$ as a context-agnostic service curve.
 
    For any acceptable trajectory $x$ with no assumption on the context of the \ac{IR}, we have
    \begin{equation}
         \forall t\ge0,\quad R^{x,\mathtt{D}}(t) \ge (R^{x,\mathtt{B}}\otimes \beta)(t)
    \end{equation}
    where $R^{x,\mathtt{B}}:t\mapsto \sum_{h\in\mathcal{F}}R_h^{x,\mathtt{B}}(t)$ is the cumulative function of the aggregate at $\mathtt{B}$, the \ac{IR}'s input (Figure~\ref{fig:sysmodel:traffic-regulators:simple-ir}), in Trajectory $x$.
    And $R^{x,\mathtt{D}}:t\mapsto \sum_{h\in\mathcal{F}}R_h^{x,\mathtt{D}}(t)$ is the cumulative function at $\mathtt{D}$, the \ac{IR}'s output, in Trajectory $x$.
 
    In $\mathcal{F}\backslash\{f_1,f_2,f_3\}$, choose arbitrarly one flow and call it $g$.
    For a given set of arrival curves $\{\alpha_h^{\mathtt{B}}\}_{h\in\mathcal{F}\backslash\{g\}}$, we define
    \begin{equation}
        \alpha_{\cancel{g}} \triangleq \sum_{h\in\mathcal{F}\backslash\{g\}}\alpha_h^{\mathtt{B}}
    \end{equation}
    where $\alpha_h^{\mathtt{B}}$ is an arrival curve for $h$ at the input $\mathtt{B}$ of the \ac{IR}. 
    For each $\theta\in\mathbb{R}^+$, we also define $\beta_{\theta}$ as:
    \begin{equation}\label{eq:proof:thm:ir-service-curves:limit-sc:beta-theta}
        \beta_{\theta}: t\mapsto \left|\beta(t)-\alpha_{\cancel{g}}(t-\theta)\right|^+ \mathbbm{1}_{\{t>\theta\}}
    \end{equation}
    where $\mathbbm{1}_{t>\theta}$ equals 1 when $t>\theta$, 0 otherwise.
    We finally define $\overline{\beta}_{\theta}$ as:
    \begin{equation}\label{eq:proof:thm:ir-service-curves:limit-sc:beta-theta:overline}
        \overline{\beta}_{\theta} = \beta_{\theta} \overline{\oslash} \boldsymbol{0} : t \mapsto \inf_{s\ge t}\{\beta_{\theta}(s)\}
    \end{equation}
    where $\overline{\oslash}$ is the max-plus deconvolution (Table~\ref{tab:motations}) and $\boldsymbol{0}$ is the zero function: $\forall t\ge 0,\ \boldsymbol{0}(t)=0$.

    \begin{lemma}\label{lemma:proof:thm:ir-service-curves:limit-sc:beta-theta-is-indiv-sc}
        If, for each flow $h\in\mathcal{F}\backslash\{g\}$, $\alpha_h^{\mathtt{B}}$ is an arrival curve of $h$ at the input of the \ac{IR}, then for each $\theta \ge 0$, the \ac{IR} offers $\overline{\beta}_{\theta}$ as an individual service curve for $g$.
    \end{lemma}
    \begin{proof}[Proof of Lemma~\ref{lemma:proof:thm:ir-service-curves:limit-sc:beta-theta-is-indiv-sc}]
        The \ac{IR} is a \acs{FIFO} system that offers the service curve $\beta$ to the aggregate. 
        By \cite[Prop. 6.4.1]{leboudecNetworkCalculus2001}, for each $\theta\in\mathbb{R}^+$, the function $\beta_{\theta}$ verifies Equation~\eqref{eq:background:indiv-sc:def}.
        
        However, for $\theta\in\mathbb{R}^+$, $\beta_{\theta}$ may not be wide-sense increasing, thus not an individual service curve for $g$.
        But $\overline{\beta}_{\theta}$ is wide-sense increasing by construction and $\forall \theta \in\mathbb{R}^+, \forall t\ge 0, \overline{\beta}_{\theta}(t)\le\beta_{\theta}(t)$.
        Hence, for $\theta\in\mathbb{R}^+$,  $\overline{\beta}_{\theta}$ also verifies Equation~\eqref{eq:background:indiv-sc:def}, thus is an individual service curve of $S$ for $g$.
    \end{proof}

    Take now $\epsilon >0$ and denote  by $X_{\epsilon}$ the subset of acceptable trajectories such that:
    \begin{itemize}
        \item $\gamma_{r,b+\epsilon}$ is an arrival curve for $f_1$ at $\mathtt{B}$, and
        \item $\gamma_{r,b}$ is an arrival curve for $f_2$ at $\mathtt{B}$, and
        \item $\gamma_{r,b}$ is an arrival curve for $f_3$ at $\mathtt{B}$, and
        \item $\forall h\in\mathcal{F}\backslash\{f_1,f_2,f_3,g\}$, $\forall t\ge 0$, $R_h^{x,\mathtt{B}}=R_h^{x,\mathtt{D}}=0$
    \end{itemize}

    For all trajectories of $X_{\epsilon}$, we have
    \begin{equation}\label{eq:proof:thm:ir-service-curves:limit-sc:alpha-2}
        \alpha_{\cancel{g}} = \gamma_{3r,3b+\epsilon}
    \end{equation}

    By applying Lemma~\ref{lemma:proof:thm:ir-service-curves:limit-sc:beta-theta-is-indiv-sc}, for each $\theta \in \mathbb{R}^+$, $\overline{\beta}_{\theta}$ is an individual service curve of the \ac{IR} for $g$.
    In addition, $b+\epsilon>b$ so the subset of trajectories $X_{\epsilon}$ meets the assumptions of Theorem~\ref{thm:ir-service-curves:limit-ir-residual}.
    By Theorem~\ref{thm:ir-service-curves:limit-ir-residual}, we have, $\forall\theta\ge0, \forall t\ge 0$,
    \begin{alignat*}{3}
        \overline{\beta}_{\theta}(t) &\le L^{\min}_g & \\
        \inf_{s\ge t}\{\beta_{\theta}(s)\}&\le L_g^{\min} &\quad \triangleright\ \eqref{eq:proof:thm:ir-service-curves:limit-sc:beta-theta:overline}\\
        \inf_{s\ge t}\left\lbrace\left|\beta(s)-\alpha_{\cancel{g}}(s-\theta)\right|^+ \mathbbm{1}_{\{s>\theta\}}\right\rbrace&\le L^{\min}_g &\quad \triangleright\ \eqref{eq:proof:thm:ir-service-curves:limit-sc:beta-theta}\\
    \end{alignat*}
    If we limit to $t>\theta$, we obtain, $\forall\theta\ge 0, \forall t>\theta$,
    \begin{equation}\label{eq:proof:thm:ir-service-curves:limit-sc:inf-un-removed}
        \inf_{s\ge t}\left\lbrace\left|\beta(s)-\alpha_{\cancel{g}}(s-\theta)\right|^+\right\rbrace \le L^{\min}_g
    \end{equation}
    For any $x\in\mathbb{R}$, $x\le |x|^+$, hence $\forall\theta\ge 0, \forall t>\theta$, 
    \begin{equation}\label{eq:proof:thm:ir-service-curves:limit-sc:removing-bars}
        \inf_{s\ge t}\left\lbrace\beta(s)-\alpha_{\cancel{g}}(s-\theta)\right\rbrace \le \inf_{s\ge t}\left\lbrace\left|\beta(s)-\alpha_{\cancel{g}}(s-\theta)\right|^+\right\rbrace
    \end{equation}
    By combining Equations~\eqref{eq:proof:thm:ir-service-curves:limit-sc:inf-un-removed} and \eqref{eq:proof:thm:ir-service-curves:limit-sc:removing-bars}, we obtain $\forall\theta\ge 0, \forall t>\theta$,
    \begin{alignat}{3}
        \inf_{s\ge t}\left\lbrace\beta(s)-\alpha_{\cancel{g}}(s-\theta)^+\right\rbrace &\le L^{\min}_g & \nonumber \\
        \inf_{s\ge t}\left\lbrace\beta(s)-\gamma_{3r,3b+\epsilon}(s-\theta)\right\rbrace &\le L^{\min}_g &\quad \triangleright\ \eqref{eq:proof:thm:ir-service-curves:limit-sc:alpha-2} \nonumber\\
        \inf_{s\ge t}\left\lbrace\beta(s)-3rs\right\rbrace &\le L^{\min}_g -3r\theta + 3b+\epsilon & \label{eq:proof:thm:ir-service-curves:limit-sc:with-epsilon}
    \end{alignat}
    Equation~\eqref{eq:proof:thm:ir-service-curves:limit-sc:with-epsilon} is valid for any $\epsilon > 0$. Hence,
    \begin{alignat}{3}
        \forall \theta\ge0,\forall t>\theta, \quad &\inf_{s\ge t}\left\lbrace\beta(s)-3rs\right\rbrace &&\le L^{\min}_g -3r\theta + 3b \nonumber\\
        \text{In particular,} \qquad & && \nonumber \\
        \forall t > \frac{L^{\min}_g}{3r}+\frac{b}{r}, \quad &\inf_{s\ge t}\left\lbrace\beta(s)-3rs\right\rbrace &&\le 0 \label{eq:proof:thm:ir-service-curves:limit-sc:inf-diff-0}
    \end{alignat}

    Denote $t_1 \triangleq  \frac{L^{\min}_g}{3r}+\frac{b}{r}$. 
    With the definition of the infimum, Equation~\eqref{eq:proof:thm:ir-service-curves:limit-sc:inf-diff-0} gives
    \begin{alignat}{2}
        & \forall t>t_1, \forall \epsilon' >0, \exists s\ge t, \quad &\beta(s)-3rs\le \epsilon'   \nonumber \\
        & \forall t>t_1, \forall \epsilon' >0, \exists s\ge t, \quad &\frac{\beta(s)}{s}\le \frac{\epsilon'}{s} + 3r \nonumber\\
        & \forall t>t_1, &\inf_{s\ge t}\frac{\beta(s)}{s} \le 3r \nonumber\\
        &\text{This shows} &\liminf_{t\rightarrow +\infty}\frac{\beta(t)}{t} \le 3r
    \end{alignat}
\end{proof}